\documentclass[12pt,twoside]{article}

\usepackage[
letterpaper,
left=2.5cm,
right=2.5cm,
top=2cm,
bottom=2cm
]{geometry}

\usepackage{setspace}
\usepackage{graphicx,amssymb,amsmath,color}
\usepackage{mathrsfs}
\usepackage{bbm}
\usepackage{mathptmx}
\usepackage{multirow}
\usepackage{booktabs}
\usepackage{tikz}
\usepackage{bm}
\usepackage{color}
\usepackage{eurosym}
\usepackage{textcomp}
\usepackage{fancyhdr}
\usepackage{epsfig}
\usepackage{times}
\usepackage[latin1]{inputenc}
\usepackage[english]{babel}
\usepackage{amsmath}

\usepackage{amssymb}
\usepackage{amsmath}
\usepackage{graphicx}
\usepackage{amsthm}
\usepackage{natbib}
\usepackage{amssymb}
\usepackage{bbding}

\usepackage{array}
\usepackage{tabularx}

\newcolumntype{Y}{>{\raggedright\arraybackslash}X}

\bibpunct[, ]{[}{]}{,}{a}{}{,}

\newtheorem{remark}{Remark}
\newtheorem{theorem}{Theorem}
\newtheorem{lemma}{Lemma}
\newtheorem{proposition}{Proposition}
\newtheorem{assumption}{Assumption}
\newtheorem{corollary}{Corollary}
\newtheorem{example}{Example}
\newtheorem{definition}{Definition}

\newenvironment{proof-theorem}[1][\sc \textbf{Proof of Theorem 4}]{\textbf{#1.} }{\ \rule{0.5em}{0.5em}}

\begin{document}



\title{Endogenous Information Design and Informational Evolution in Dynamic Economic Systems with Incentives and Learning\thanks{I am grateful to Bo Chen, Zhen Huang, Qingmin Liu, Dawen Meng, Alessandro Pavan, Yingyi Qian, Ning Sun, Guofu Tan, Neng Wang, Yijiang Wang, Mingjun Xiao, Jun Zhang, and Sijia Zhang for helpful comments and suggestions.}}

\author{Guoqiang Tian\footnote{Department of Economics, Texas A\&M University, USA. Email: gtian@tamu.edu.}\\
}

\date{September 2026}
\maketitle

\begin{abstract}

This paper studies the endogenous informational evolution of dynamic economic systems. Building on Hurwicz's conceptual framework, such a system is defined as a dynamic joint adjustment process of resource and information allocation linked by recursive feedback. Economic theory gains precision by isolating particular margins, but when information production, disclosure, learning, incentives, actions, and future feasibility interact across these two processes, conclusions obtained by holding other margins fixed need not survive their joint determination. The paper develops a system-level dynamic equilibrium framework for this interaction. Holding particular links fixed yields familiar environments in dynamic mechanism design, information design, experimentation and learning, feedback and disclosure, and information acquisition. The analysis establishes three central failures relative to fixed or partially exogenous informational environments: informational efficiency under exogenous feasibility is generally incompatible with endogenous dynamic feasibility; incentive and disclosure design cannot generally be separated; and informational monotonicity can fail without implementation replicability. It also identifies the corresponding positive boundaries. Optimal disclosure may be interior or extreme; bang--bang disclosure arises under additional conditions, and a further specialization generates an endogenous transition from opacity to transparency. Information may therefore be optimally withheld even when its direct decision value is strictly positive. Extensions to private payoff-relevant information and limited commitment preserve the recursive structure.

\medskip

\noindent
\textbf{JEL Classification:} D82, D83, D86, C73, D81.

\noindent
\textbf{Keywords:} endogenous informational evolution, system-level economic analysis, endogenous information production, information design, dynamic mechanism design, incentive--disclosure nonseparability, informational monotonicity, dynamic feasibility.

\end{abstract}

\setcounter{tocdepth}{2}
\tableofcontents

\section{Introduction}

\par
Economic theory often studies choice, incentives, and institutional design within an informational environment that is treated as given. This is analytically indispensable: a theory becomes useful by isolating a particular margin and holding other margins fixed. But an economic system does not merely operate \emph{within} an informational environment. Economic activity may determine what information is produced; institutional rules determine what generated information becomes available and usable; learning changes beliefs; beliefs change incentives and decisions; and the resulting actions alter both future information and future economic opportunities. When these feedbacks are economically important, the informational environment is itself an endogenous outcome of the economic system.

\par
This paper studies economic systems at this level, defining an economic system as a dynamic joint adjustment process of resource and information allocation linked by recursive feedback. Markets, organizations, contracts, and institutions shape both processes: they govern resource allocation and how information is produced, disclosed, incorporated into beliefs, and used. Economic actions therefore change both informational conditions and future opportunities, which in turn shape subsequent incentives and actions. The central object is their joint dynamic evolution rather than optimization within a fixed informational environment.

\par
This distinction implies a corresponding methodological distinction between \emph{partial endogenization} and a \emph{system-level approach}. Modern economic theory has progressively endogenized many components that earlier models treated as fixed. Dynamic mechanism design endogenizes intertemporal incentives and continuation relationships. Information design endogenizes signal or disclosure structures. Experimentation and learning-by-doing make information production depend on actions. Information-acquisition and rational-inattention models endogenize precision, search, or attention. Dynamic disclosure and feedback design make the timing and content of information transmission strategic. Each of these theories appropriately endogenizes the margin needed for the question it studies.

\par
Partial endogenization, however, is not yet a system-level analysis when the endogenized component interacts materially with components that remain fixed. The distinction is not the number of endogenous variables, but whether relevant margins of resource and information allocation are studied jointly through recursive feedback. Information production affects what can be learned; disclosure affects the return to producing information; incentives affect the actions that generate information; those actions change posterior beliefs and future feasibility; and future beliefs and feasibility in turn change continuation incentives and future actions. When these interactions change the feasible set, the implementable set, or the welfare comparison, conclusions derived by considering the margins separately need not survive their joint determination.

\par
This logic also shows why integrating interacting margins can be essential. A useful illustration is \citet{BergemannHeumannMorris2026}, who jointly optimize mechanism and information design and show that solving the two problems separately can miss the joint optimum. Closely related, \citet{ElyGeorgiadisRayo2025} study dynamic incentives and feedback when effort affects the production of performance information, while future feasibility remains exogenous. Under the system-level approach developed here, resource feasibility and informational evolution are instead jointly endogenous: the same information-producing action may also change future economic opportunities. As a result, an informational frontier attainable under exogenous feasibility may become unattainable once feasibility evolves endogenously.

This perspective is rooted in Hurwicz's treatment of resource allocation and informational decentralization within a common adjustment-process framework. Since \citet{Hurwicz1960}, realization theory has asked whether social objectives can be supported through informationally decentralized processes and how the informational requirements of such processes can be compared \citep{Hurwicz1972a,Hurwicz1972b,MountReiter1974,Reiter1977,HurwiczReiter2006,Walker1977,Jordan1982,Tian2004,Tian2006}. Implementation theory adds the complementary incentive question of whether desired outcomes can be sustained when participants behave strategically \citep{Hurwicz1973,Maskin1977,GrovesLedyard1977,DasguptaHammondMaskin1979,Myerson1981,PostlewaiteSchmeidler1986,PalfreySrivastava1987,PostlewaiteWettstein1989,Tian1989,MooreRepullo1990,AbreuSen1990,Jackson1991}, while \citet{ReichelsteinReiter1988} show that incentive requirements can themselves alter the informational complexity required for decentralized implementation. Hurwicz did not develop the modern stochastic-control, dynamic-contracting, persuasion, or experimentation models used below. The connection is instead methodological: resource allocation and its informational requirements are treated as parts of a common adjustment process, making it natural to ask how information itself is generated and how the resulting informational state feeds back into subsequent allocation.

\par
The present paper takes a further step by studying resource allocation and informational evolution jointly within a recursive dynamic system. Specialized theories remain indispensable because they isolate particular mechanisms and generate sharp conclusions; the claim is not that every economic problem should be modeled at this level. A system-level treatment becomes relevant when interactions among those mechanisms determine feasibility, implementation, or welfare. The paper identifies such interactions formally and shows that their consequences cannot in general be recovered by treating information production, information use, incentives, and feasibility as separate blocks.

A minimal example makes the distinction transparent. Consider a firm that delegates an innovation or experimentation project to a manager. The manager chooses research, testing, or implementation intensity. This action generates evidence about an initially unknown environment, is costly, and may consume scarce testing capacity or other future opportunities. The firm chooses compensation and the rules governing how generated evidence is validated, disclosed, and incorporated into future decisions. More experimentation generates more information but may reduce future capacity. More disclosure makes generated information more useful for learning but also changes the continuation return to experimentation. Compensation affects the manager's willingness to undertake the information-producing action. The resulting posterior belief changes future decisions and continuation incentives, which in turn change subsequent experimentation.

The resulting system can be represented schematically as the feedback loop
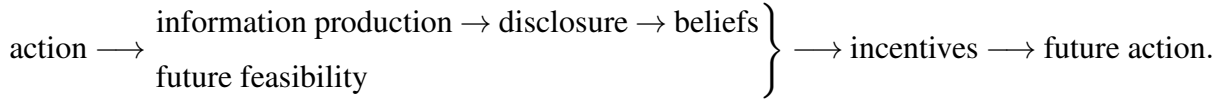
\[
\text{action}
\longrightarrow
\left.
\begin{aligned}
&\text{information production}\to\text{disclosure}\to\text{beliefs}\\
&\text{future feasibility}
\end{aligned}
\right\}
\longrightarrow
\text{incentives}
\longrightarrow
\text{future action}.
\]
Thus action jointly determines what can be learned and what remains feasible; when the action represents research, experimentation, monitoring, search, or attention, it also determines information-acquisition intensity. Posterior beliefs and feasibility then feed through continuation incentives into future action.

\par
This example also clarifies three objects that must be kept analytically distinct. The first is the \emph{primitive economic and informational environment}: preferences, technologies, initial uncertainty, action sets, the information-production technology, and physical feasibility. The second is the \emph{mechanism}: transfers, disclosure rules, recommended or target actions, communication rules, and other institutional choices. The third is the \emph{equilibrium informational evolution}: the path of actions, signals, disclosure, posterior beliefs, transfers, and feasibility induced when participants respond strategically to the mechanism. The mechanism is designed; the informational evolution is induced. Confusing these two objects obscures what is actually endogenous.

\par
This distinction also changes the relevant meaning of implementation. In a fixed informational environment, one can ask whether a desired allocation or action rule can be supported subject to incentive constraints defined on that environment. Here the informational path itself must also be generated by the economic system. The designer cannot in general select an arbitrary desired posterior path and then attach incentives to it. A posterior path must be generated by an admissible information-production technology, by actions that participants can be induced to take, and by disclosure rules available to the institution; at the same time, those actions must preserve the feasibility required for continuation. An informational path is therefore economically relevant only if it is technologically attainable, informationally admissible, dynamically feasible, and strategically implementable.

\par
Continuation values provide the natural recursive connection among these components. Current action changes future information and feasibility; these changes alter continuation values; and continuation values determine current incentives. The recursive formulation therefore links current behavior to the endogenous informational and economic states that the behavior itself helps create. The dynamic-programming and filtering methods used to characterize this recursion are analytical tools rather than the source of the paper's novelty. Their role is to make the feedback structure sufficiently precise that feasibility, implementation, and optimal information use can be studied jointly.

\par
Seen from this perspective, the paper contributes at three related levels.

\par
First, and most fundamentally, the paper develops a system-level conceptual framework in which resource and information allocation evolve jointly through recursive feedback. Information is produced through economic activity, allocated and transformed through institutional arrangements, incorporated into posterior beliefs through learning, and used in decisions and continuation incentives. The relevant economic object is consequently their joint evolution rather than either process considered independently.

\par
Second, and as a consequence, the paper develops an analytical framework in which information production, disclosure, learning, incentives, actions, and feasibility are jointly determined within an implementable dynamic system. The framework separates the primitive environment, the mechanism, and the induced informational evolution and connects them through continuation values and endogenous state dynamics. Familiar benchmark environments arise when particular links in this feedback structure are held fixed: automatic disclosure moves the model toward experimentation and learning; exogenous information production toward information or disclosure design; fixed informational dynamics toward dynamic contracting; and the absence of strategic incentives toward controlled learning or information acquisition. The purpose is not to relabel these literatures as special cases, but to make explicit which feedback is suppressed when a particular component is treated as exogenous.

\par
Third, and most importantly, the paper shows that these system-level interactions change substantive economic conclusions. Three results identify dimensions along which conclusions obtained under fixed or partially exogenous informational environments need not survive once the relevant feedback is endogenous.

\par
The first concerns informational attainment. Relative to environments such as \citet{ElyGeorgiadisRayo2025}, in which information-producing activity affects learning and incentives while future feasibility remains exogenous, the system-level formulation here makes resource feasibility endogenous together with informational evolution. When information-producing activity consumes scarce future opportunities, informational efficiency under exogenous feasibility is generally incompatible with endogenous dynamic feasibility. The paper derives a general informational-feasibility bound linking attainable information to the remaining feasibility stock and time horizon. Under linear--Gaussian learning, it obtains the exact attainable terminal-precision set and the sharp boundary at which the exogenous-feasibility informational frontier remains attainable. The obstruction is technological rather than contractual: transfers, disclosure rules, and continuation promises cannot generate information that the remaining economic opportunities do not permit the system to produce.

\par
The second concerns the organization of design. Incentive design and disclosure design cannot generally be separated because disclosure changes the continuation return to information-producing action. A transfer and continuation-value plan that implements an action under one disclosure rule need not implement the same action under another. The paper establishes the failure of generally lossless sequential incentive--disclosure design, provides an exact pointwise separation criterion, and identifies sufficient conditions under which sequential design is valid. The result provides a system-level counterpart to the joint-design insight emphasized by \citet{BergemannHeumannMorris2026}: here nonseparability arises because the use of information changes the incentives governing the behavior that produces the information itself.

\par
The third concerns the value of a more informative technology. Blackwell-more-informative information production need not increase the principal's maximal implementable value. The conventional monotonicity argument relies on the ability to garble or ignore additional information and thereby reproduce outcomes attainable under a less informative technology. That argument can fail when institutional restrictions prevent the more informative environment from reproducing the informational and implementation path generated under the less informative one. The paper identifies \emph{implementation replicability} as the relevant positive condition: when implementable outcomes under the less informative technology can be reproduced under the more informative technology, informational monotonicity is restored.

\par
These three results share a common economic logic. Improving one component of the system can alter another component of the jointly feasible and implementable set. More experimentation may generate more information but consume future opportunities; more disclosure may improve learning but change the incentives supporting information-producing action; and a more informative technology may create additional information while changing the implementation requirements needed to use it. The relevant object is therefore the performance of the endogenous system as a whole, not the isolated quality of any one component.

\par
The disclosure analysis provides a complementary positive characterization. In the general model, optimal disclosure is determined by its contribution to optimized implementable continuation value and may be interior or extreme. Bang--bang disclosure can arise under nonlinear reduced disclosure objectives, while an affine reduced implementable disclosure Hamiltonian yields a general bang--bang characterization. Under the additional conditions stated in the paper, optimal disclosure has an endogenous two-phase structure in which an initial period of opacity is followed by transparency. The economic comparison is between the direct benefit created by additional public information and the implementation cost required to preserve incentives and participation. Thus information may have positive direct decision value and nevertheless be optimally withheld when its implementation cost is sufficiently large.

\par
A closed-form delegated-experimentation application illustrates the mechanism. Experimentation produces evidence about an unknown technology, is costly, and consumes scarce testing capacity. Disclosure determines whether generated evidence becomes publicly usable and changes the incentive to experiment. Under the stated conditions, an optimal path preserves capacity and remains inactive initially, then switches to maximal experimentation and full disclosure as the terminal decision approaches. The cutoff is generated jointly by informational benefits, experimentation costs, implementation requirements, and the value of remaining capacity rather than imposed as an exogenous disclosure pattern.

\par
The same logic applies whenever information-producing activity also changes future opportunities. Examples include delegated innovation and organizational learning, regulatory monitoring and reporting, lending and borrower monitoring, product experimentation and quality learning, and digital platforms in which behavior simultaneously generates data and changes future recommendations or interactions. In these environments, institutions shape not only resource allocation or information disclosure, but also the processes through which information is produced, incorporated into beliefs, and fed back into future behavior.

\par
The benchmark model uses a common but initially unobserved informational-environment state that governs information production rather than entering primitive payoffs directly. This differs from standard Bayesian mechanism-design models in which privately observed types directly affect preferences, valuations, costs, or feasible allocations. The benchmark isolates the endogenous information-production channel. Section~\ref{subsec:payoff_relevant_information_states} extends the framework to a composite state space in which private payoff-relevant information and informational-environment uncertainty coexist, so that reporting and action incentive constraints interact within the same dynamic economic system. The paper also studies limited commitment and the resulting continuation participation constraints.

\par
The model is formulated in continuous time. The agent chooses actions that affect signal production and a publicly observable feasibility state, while the principal chooses transfer and disclosure rules and seeks to implement a target action process. Posterior beliefs evolve through Bayesian updating based on disclosed information, and feasibility evolves through the accumulated consequences of past actions. The benchmark uses an action-revealing specification to isolate the endogenous informational-evolution mechanism. Genuinely hidden information-producing actions can be incorporated by introducing noisy observable performance and the associated inference and moral-hazard problem.

\par
The remainder of the paper proceeds as follows. Section~2 develops the system-level perspective more explicitly. It distinguishes partial endogenization from system-level analysis, uses a minimal economic environment to display the feedback among action, information production, disclosure, beliefs, continuation incentives, and feasibility, and then compares this structure with realization and implementation theory, dynamic mechanism design, information design, experimentation and learning, feedback and disclosure design, information acquisition, and joint mechanism--information design. Section~3 introduces the formal dynamic economic environment with endogenous information production, distinguishes resource allocation from informational allocation, and defines the mechanism and induced state dynamics. Section~4 characterizes implementable informational evolution through continuation values, pointwise incentive conditions, and induced state dynamics. Section~5 formulates the principal's recursive contracting problem, establishes the three main impossibility results and their positive boundaries, and develops the marginal-value decomposition. Section~6 characterizes optimal disclosure under general and affine reduced disclosure technologies, establishes the bang--bang and opacity-to-transparency results under the stated conditions, and derives comparative statics. Section~7 provides the closed-form delegated-experimentation application with scarce testing capacity. Section~8 studies private payoff-relevant information and limited commitment. Section~9 concludes by drawing out the broader implications of the system-level perspective and identifying open questions toward a more general theory of dynamic economic systems with endogenous informational evolution.

\section{Related Literature and System-Level Comparison}
\label{sec:related_literature}

The Introduction has already stated the paper's system-level perspective. The purpose of this section is therefore narrower: to locate the framework relative to the main literatures that endogenize different components of an informational economic system. The useful comparison is not simply whether ``information'' is endogenous, but which objects are primitive, which are designed, and which are induced through equilibrium behavior. The literatures below deliberately isolate different margins. The present paper studies an environment in which the interactions among those margins---information production, disclosure, posterior evolution, incentives, actions, and feasibility---are themselves economically consequential.

\subsection{Informational Systems, Realization, and Implementation}
\label{subsec:literature_realization_implementation}

The conceptual foundation is the informational perspective initiated by \citet{Hurwicz1960}, in which an economic system is studied as an informationally decentralized adjustment process. General equilibrium theory characterizes the joint determination of allocations and prices under specified preferences, technologies, and resources \citep{ArrowDebreu1954,McKenzie1954,ArrowHurwicz1958,Debreu1959,McKenzie1959,HurwiczUzawa1971}. Hurwicz's realization approach asks whether desired allocations can be generated through informationally decentralized processes and how the informational requirements of alternative processes can be compared \citep{Hurwicz1972a,Hurwicz1972b,MountReiter1974,Reiter1977,HurwiczReiter2006,Walker1977,Jordan1982,Tian2004,Tian2006}.

Implementation theory adds the incentive requirement that desired outcomes remain attainable when participants behave strategically \citep{Hurwicz1973,Maskin1977,DasguptaHammondMaskin1979,Myerson1979,Postlewaite1979,Myerson1981,Thomson1984,PostlewaiteSchmeidler1986,PalfreySrivastava1987,PalfreySrivastava1989,PostlewaiteWettstein1989,Tian1989,MooreRepullo1990,AbreuSen1990,Jackson1991}. \citet{ReichelsteinReiter1988} show that incentive requirements can raise the informational complexity needed for decentralized implementation relative to decentralized realization, while important public-good environments attain implementation with minimal message spaces \citep{GrovesLedyard1977,Walker1981,Tian1990,Tian1991}. Recent work also studies mechanism design when the informational environment itself is imperfectly known \citep{McLeanPostlewaite2025}.

These theories make information and incentives properties of the economic mechanism rather than invisible background assumptions. The present paper changes a different margin: the informational state on which adjustment and implementation operate is itself generated through economic activity. Actions produce evidence and alter future feasibility, while disclosure determines how evidence enters public beliefs. The one-principal--one-agent benchmark therefore abstracts from decentralized communication and minimal-message-space questions and focuses instead on the endogenous evolution of the informational environment.

The term \emph{informational efficiency} used later in this paper should accordingly be distinguished from informational efficiency in realization theory. Here it concerns attainment of an informational frontier under a specified information-production technology; in realization theory it concerns the informational or dimensional requirements of decentralized allocation processes. The common methodological point is that informational performance is part of the economic performance of a mechanism.

\subsection{Dynamic Mechanism Design and Recursive Contracting}
\label{subsec:literature_dynamic_mechanism}

The paper is closely related to dynamic mechanism design and recursive contracting. This literature studies continuing regulatory relationships, sequential screening, long-term contracting with evolving private information, dynamic moral hazard, and intertemporal implementation \citep{BaronBesanko1984,CourtyLi2000,Battaglini2005,Sannikov2008,AtheySegal2013}. \citet{PavanSegalToikka2014} develop a general Myersonian approach to dynamic mechanism design with evolving private information, while \citet{GarrettPavan2015} provide a variational approach to dynamic managerial compensation. Related work studies dynamic institutional and network effects \citep{MengTian2021,MengSunTian2022}.

These models endogenize transfers, allocations, reports, actions, continuation utilities, and dynamic incentive constraints. Current actions or reports may also affect future states and information. The information-production and observation technologies, however, are generally specified as part of the contracting environment. The present framework separates an additional institutional margin: action generates evidence, while disclosure determines how that evidence becomes part of the public informational state. The same action may also alter future feasibility. Disclosure can therefore affect the continuation return to information-producing action and hence implementation itself.

Continuation values play the familiar recursive role of summarizing future incentive consequences, while posterior beliefs and feasibility states summarize what the system has learned and what economic opportunities remain. Methodologically, the paper draws on recursive contracting and dynamic agency \citep{SpearSrivastava1987,Sannikov2008,Williams2011}, continuous-time agency \citep{HolmstromMilgrom1987,SchattlerSung1993,He2011,CvitanicPossamaiTouzi2018}, and filtering and stochastic control \citep{LiptserShiryaev2001,BainCrisan2009,YongZhou1999,FlemingSoner2006}. These methods provide the analytical infrastructure for characterizing an implementable informational process generated jointly by action, disclosure, incentives, and feasibility.

The marginal-value analysis is similarly related to continuation-value and envelope methods in dynamic mechanism design \citep{PavanSegalToikka2014,GarrettPavan2015}. The additional object here is the informational state itself: a change in public information can improve decisions while simultaneously changing the adjustments required to preserve implementation.

\subsection{Information Design, Disclosure, and Joint Design}
\label{subsec:literature_information_design}

The paper is also related to information design and Bayesian persuasion. Foundational work on statistical experiments and information structures includes \citet{Blackwell1953} and \citet{AumannMaschler1995}. Modern information-design approaches study how a designer selects signals, experiments, posterior distributions, or informational environments subject to Bayes plausibility and obedience constraints \citep{KamenicaGentzkow2011,BergemannMorris2016,BergemannMorris2019}.

A related dynamic literature studies persuasion, strategic information provision, disclosure, feedback, ratings, and informational intermediation \citep{Ely2017,RenaultSolanVieille2017,HornerSkrzypacz2017,ElySzydlowski2020,HornerLambert2021,Ball2023,BergemannBonattiSmolin2018,AliEtAl2022,ChopraEly2025,ChopraEly2026,ElyGeorgiadisRayo2025,HornerSamuelson2026}. These models show that the timing and content of information can affect strategic behavior, continuation incentives, and welfare. Particularly related, \citet{HornerSamuelson2026} show that providing information can distort incentives and that withholding information may therefore be beneficial.

\par
Within this literature, \citet{ElyGeorgiadisRayo2025} provide a particularly close comparison on the dynamic information--incentive margin. In their model, hidden effort endogenously affects the arrival of a coarse, at-most-once performance signal under a primitive information-production technology, while the principal jointly designs incentives and feedback. Information production and feedback are therefore endogenous, but the agent does not separately choose an experiment, signal precision, or attention, and future feasibility is taken as given. The present framework shares this action-dependent information-production channel but also allows the action itself to represent information-acquisition intensity, permits the informational state to evolve continuously, and lets the same action alter future economic opportunities. This broader coupling among information production or acquisition, disclosure, incentives, learning, and feasibility can make an informational frontier attainable under exogenous feasibility unattainable once feasibility evolves endogenously.

The main distinction concerns how information becomes available for design. In canonical information-design models, an experiment, signal structure, disclosure policy, or posterior distribution is a direct design object, while the underlying environment is specified. In the present framework, raw evidence must first be generated through economic action. Disclosure is designed, but the public posterior process is induced jointly by information production, disclosure, incentives, and Bayesian updating, and the action producing information may also change future feasibility.

A complementary comparison is the joint mechanism--information design framework of \citet{BergemannHeumannMorris2026}. They show that mechanism and information design can be nonseparable and should in general be optimized jointly rather than sequentially. The present paper shares this nonseparability logic but studies a different architecture. The principal chooses transfers and disclosure, whereas information production is generated by the agent's action and may alter future feasibility. The informational state is therefore not simply another design variable; it is an equilibrium process generated by the economic system.

This difference also affects comparisons across information technologies. In standard settings with free informational disposal, a more informative technology can reproduce a less informative one through garbling or nonuse. Here the relevant question is whether the less informative technology's \emph{implementable} path can be reproduced. Institutional restrictions may prevent such replication because greater informativeness changes public learning and hence continuation incentives. The positive condition identified later is therefore implementation replicability rather than Blackwell dominance alone.

The disclosure results arise from the same interaction. Optimal disclosure is determined by its contribution to the implementable continuation value rather than by the direct decision value of information alone. General disclosure may be interior or extreme; additional structure yields bang--bang disclosure and, under further conditions, an endogenous transition from opacity to transparency.

\subsection{Experimentation, Learning, and Information Acquisition}
\label{subsec:literature_experimentation_acquisition}

The paper is related to learning-by-doing, experimentation, bandit problems, and endogenous informational accumulation. Early contributions include \citet{Arrow1962} and \citet{Rothschild1974}. Subsequent work studies strategic experimentation, learning-by-doing, bandit allocation, and learning in industries and markets \citep{BoltonHarris1999,KellerRadyCripps2005,BergemannValimaki2008,JovanovicNyarko1996,
EricsonPakes1995,BesankoWu2013,GuoHeLiu2023}. A central feature of these models is that economic activity changes what can be learned. A related contracting literature studies how incentives shape experimentation and innovation, including \citet{Manso2011} on incentives for exploration and innovation and \citet{HalacKartikLiu2016} on optimal contracting for experimentation.

\par
The present framework shares this endogenous information-production channel and can also
incorporate information acquisition when the action represents search, monitoring,
experimentation, attention, or precision investment. In that case, information-acquisition
intensity is endogenous through action, while the technology mapping action into evidence
remains primitive. Generated evidence need not enter the public posterior automatically:
disclosure governs its informational use, transfers determine whether the information-producing action is implementable, and the same action may also change future feasibility. The attainable informational path therefore depends jointly on information production or acquisition, informational usage, incentives, and remaining economic opportunities.

The paper is also related to information acquisition, endogenous precision choice, costly search, and rational inattention. Early work emphasizes costly search and the value of information \citep{Stigler1961,McCall1970,Arrow1974,GrossmanStiglitz1980}. Subsequent contributions study endogenous information acquisition and precision in markets and organizations \citep{Verrecchia1982,Persico2000,Veldkamp2006,Veldkamp2011}, while rational-inattention models make attention and information processing endogenous \citep{Sims2003,Woodford2009,Moscarini2004,MackowiakWiederholt2009,MatejkaMcKay2015,CaplinDeanLeahy2022,HebertWoodford2023}.

Related contracting models study strategic or covert information acquisition
\citep{CremerKhalil1992,LewisSappington1994,BergemannValimaki2002,
JohnsonMyatt2006,Szalay2009}. \citet{BergemannValimaki2002} study the interaction
between costly information acquisition and efficient mechanism design, while in
\citet{Szalay2009}, for example, an agent privately selects an information structure
before reporting to the principal, so information acquisition and incentive provision
interact. The present framework instead considers information that accumulates
continuously through economic activity, with disclosure separately controlled by the
principal and feasibility potentially affected by the same action.

Experimentation and information-acquisition models therefore provide the main foundations for endogenous information production. The present framework adds the institutional determination of informational use and the requirement that the resulting informational path be dynamically implementable.

\subsection{Comparison of Frameworks}
\label{subsec:literature_comparison}

Table~\ref{tab:framework_comparison} summarizes the comparison. The entries describe the principal emphasis of canonical formulations rather than every contribution within each literature. The key distinction is among information production, informational design or disclosure, incentive design, posterior evolution, and feasibility dynamics.

\begin{table}[!htbp]
\centering
\scriptsize
\setlength{\tabcolsep}{3pt}
\renewcommand{\arraystretch}{1.12}

\caption{Comparison of Informational Frameworks}
\label{tab:framework_comparison}

\begin{tabularx}{\textwidth}{
    >{\raggedright\arraybackslash}p{2.25cm}
    Y Y Y Y Y
}
\hline
Framework
& Information production
& Disclosure or information structure
& Incentive mechanism
& Posterior evolution
& Feasibility dynamics
\\
\hline

Realization theory
& Usually primitive or not central
& Usually primitive
& Not central in the canonical formulation
& Usually not central
& Usually primitive
\\[0.35em]

Implementation and mechanism design
& Usually primitive
& Usually primitive
& Designed to implement desired outcomes
& Usually induced under a given information structure
& Usually primitive
\\[0.35em]

Dynamic mechanism design
& Usually primitive; sometimes action-dependent
& Usually primitive
& Designed subject to dynamic incentive constraints
& Induced under a specified informational environment
& Sometimes endogenous
\\[0.35em]

Information design
& Usually primitive
& Designed
& Given or incorporated through obedience constraints
& Designed or induced by the experiment
& Usually absent
\\[0.35em]

Joint mechanism--information design
& Usually primitive
& Jointly designed
& Jointly designed
& Induced by the joint design
& Usually absent
\\[0.35em]

Dynamic moral hazard and feedback design
& Endogenous through effort
& Designed through feedback
& Jointly designed with incentives
& Induced by effort and feedback
& Usually exogenous or absent
\\[0.35em]

Experimentation and learning-by-doing
& Endogenous through action
& Usually automatic or fixed
& Usually given or implicit
& Endogenous
& Sometimes endogenous
\\[0.35em]

Information acquisition and rational inattention
& Chosen or endogenously produced
& Chosen or implicit in the acquisition technology
& Usually given; sometimes jointly studied
& Induced
& Usually absent
\\[0.35em]

Present framework
& Endogenous through action
& Designed through disclosure
& Designed subject to IC and IR
& Jointly induced by action and disclosure
& Endogenous through action
\\
\hline
\end{tabularx}

\vspace{0.15cm}

\begin{minipage}{0.96\textwidth}
\footnotesize
\noindent
Notes: ``Primitive'' means specified as part of the economic environment; ``designed'' means directly selected by the principal or designer; ``endogenous'' means determined through behavior or state dynamics; and ``induced'' means generated by the interaction of primitives, designed rules, and equilibrium behavior. The entries describe canonical formulations and admit important exceptions.
\end{minipage}

\end{table}

\par
The table makes clear that ``endogenous information'' can refer to different objects.
Information may be produced by action, acquired through costly choice, selected through an
experiment, governed through disclosure, or embodied in a posterior induced by equilibrium
behavior. The present framework connects these operations in one causal structure: action
produces or acquires raw evidence and changes feasibility; disclosure governs its public use; Bayesian updating generates posterior beliefs; and transfers and continuation values determine whether the resulting action and informational paths are implementable.

Several benchmark environments are obtained by holding particular links fixed. If information production is independent of action, the model approaches information or disclosure design under a specified signal technology. If action produces information but disclosure is automatic, the separate disclosure margin disappears and the model approaches experimentation or learning-by-doing. If the informational dynamics are fixed, the environment approaches dynamic contracting under a specified information structure. If strategic incentives are absent, the problem approaches controlled learning or information acquisition. These comparisons identify which feedback is suppressed; they are not claims that the corresponding literatures are literally nested in every other respect.

The term \emph{informational evolution} also differs from its use in evolutionary game theory and learning in games, where strategies, beliefs, or population states evolve through selection, imitation, reinforcement, or payoff feedback \citep{MaynardSmith1982,Weibull1995,FudenbergLevine1998,Young1998,Sandholm2010}. Here the evolving object is the public informational state of an implementable economic system, generated through information-producing actions, institutional disclosure, Bayesian updating, and feasibility dynamics.

Taken together, the related literatures provide complementary components of the architecture studied here. Realization theory makes informational requirements explicit; implementation and dynamic contracting make strategic incentives and continuation relationships endogenous; information design and disclosure make informational rules objects of design; experimentation and information acquisition make information production endogenous; dynamic feedback design links information-producing effort to incentives and informational feedback; and joint mechanism--information design establishes the importance of nonseparability across design margins. The present paper studies the system-level feedback among these components when that feedback changes feasibility, implementation, or welfare.

\section{Dynamic Economic Environment with Endogenous Information}
\label{sec:dynamic_environment}

We now introduce a dynamic principal--agent environment in which economic activity produces information. As a benchmark interpretation, a firm, the principal, delegates an innovation or experimentation project to a manager, the agent, who chooses the intensity of research, testing, or implementation. The action affects current performance, generates evidence about an initially unknown environment, and may consume scarce future capacity. The principal designs transfers and the rules governing how the generated evidence is validated, disclosed, and incorporated into future decisions and incentives.

The same structure applies to a regulator delegating inspections to a monitoring unit, a research sponsor assigning experimentation to a scientific team, or a financial institution delegating borrower monitoring to a loan officer. In each case, information production and use depend on actions, incentive provision, disclosure, and future feasibility. The analysis distinguishes the primitive information-production environment from the mechanism designed within that environment and from the informational and feasibility processes induced in equilibrium. The general formulation is presented first, followed by the linear--Gaussian benchmark and the dynamic mechanism.

\subsection{Informational Economic Environment}
\label{subsec:informational_environment}

Time is continuous, with $t\in[0,\bar T]$ and $\bar T<\infty$. Continuous time is used for analytical convenience and for the recursive characterization developed below; the underlying economic mechanism is not specific to continuous time.

The economy consists of a principal and an agent. The principal designs transfer and disclosure rules, while the agent chooses the intensity of the delegated activity. At each date, the agent has at most $\bar a<\infty$ units of productive capacity and chooses $a_t\in\mathcal A=[0,\bar a]$ for the delegated activity, leaving $\bar a-a_t$ for alternative uses. The pair $(a_t,\bar a-a_t)$ constitutes the contemporaneous \textbf{resource allocation}.

Depending on the application, $a_t$ may represent research effort, experimentation, monitoring, information acquisition, implementation, production, search, or learning-by-doing. The action entails a flow cost $c(a_t)$, where $c(0)=0$, $c'(a)>0$, and $c''(a)\geq0$.

Let $x_t\in X=[0,\bar x]$ denote a publicly observable \textbf{feasibility state} inherited from past activity. It may represent remaining testing capacity, organizational slack, financial resources, inventories, productive capacity, or future experimentation opportunities. Unlike the flow allocation $a_t$, the state $x_t$ summarizes the accumulated consequences of past activity and determines future economic opportunities. Its law of motion is specified in Subsection~\ref{subsec:informational_dynamics}.

The economy contains two linked allocation processes. The \textbf{resource allocation process} determines how productive capacity is divided between the delegated activity and alternative uses. The \textbf{informational allocation process} determines how information is produced and how generated information is disclosed, incorporated into beliefs, and used in decisions and incentives. The two processes are jointly determined because the agent's action affects information production, while the principal's informational policy affects the return to that action.

A defining feature of the environment is \textbf{information production}: economic activity generates information as well as current economic outcomes. Let $\theta\in\Theta$ denote a common but initially unobserved informational-environment state. In the benchmark model, $\theta$ governs the information-production technology through which actions generate signals rather than entering primitive payoffs directly. It is drawn from a common prior distribution $F$ on $\Theta$, known to both the principal and the agent, which admits a density with respect to a reference measure and has support equal to $\Theta$. Random variables are distinguished from their realizations by a tilde; thus
\(\widetilde{\theta}_t\) denotes the random state and \(\theta_t\) its realization,
with the same convention used for signals, actions, and outcomes.

The resulting evidence is represented by a raw signal process $(S_t)$. Conditional on the current action, feasibility state, and informational-environment state, signals are distributed according to $Q(\cdot\mid a_t,x_t,\theta)$. The conditional distribution $Q$ is the \textbf{information-production technology}: it describes how economic activity converts productive resources into evidence about the unknown environment.

Generated signals need not enter the public information set automatically. The mechanism determines how they are validated, communicated, disclosed, or otherwise made available for contracting and economic use. Posterior beliefs consequently evolve through the joint effects of action, information production, disclosure, and Bayesian updating. This interaction generates the \textbf{endogenous informational evolution} studied in the paper.

\begin{remark}[Private payoff-relevant uncertainty and common informational-environment uncertainty]
\label{rem:private_types_informational_states}

The informational-environment state in the benchmark differs from the private payoff-relevant types commonly studied in Bayesian mechanism design. In standard screening environments, privately observed types or signals directly affect preferences, valuations, costs, or feasible allocations, and the principal designs a mechanism subject to truthful-reporting constraints. Here, the informational-environment state $\theta$ is initially common but unobserved and governs the information-production technology rather than directly representing the agent's payoff-relevant type.

This distinction isolates the information-production and informational-adjustment channel rather than imposing a substantive restriction. Section~\ref{subsec:payoff_relevant_information_states} extends the framework to environments in which private payoff-relevant information and informational-environment uncertainty coexist, so that reporting and action incentive constraints operate within the same dynamic system.
\end{remark}

The incorporation of generated information into decision making, contracting, forecasting, investment, organizational coordination, and incentive provision is referred to as \textbf{informational usage}. Information production and informational usage jointly determine \textbf{informational allocation}: the former governs how evidence is generated, while the latter governs how that evidence affects beliefs, decisions, and continuation incentives.

Generated and disclosed information induces a posterior belief process $(\mu_t)$ that summarizes the public informational state of the relationship. Let $T_t$ denote the transfer from the principal to the agent. Transfers may depend on publicly observable information, including disclosed signal histories and public state variables. The agent's flow utility is $u(T_t,\mu_t,x_t)-c(a_t)$, where $u$ is weakly increasing in $T_t$ and $x_t$. The principal's flow utility is $v(T_t,a_t,\mu_t,x_t)$, which is weakly decreasing in transfers. Let $U_{\bar T}(\mu,x)$ and $V_{\bar T}(\mu,x)$ denote the agent's and principal's terminal payoff functions, respectively, which serve as the terminal conditions for their continuation-value problems at date $\bar T$. If no separate terminal payoff is present, the corresponding terminal-payoff function is identically zero.

The timing is as follows. The principal first commits to the transfer and disclosure rules, and the agent decides whether to participate. Conditional on participation, the agent subsequently chooses the action. The realized action and generated information then determine information production, posterior updating, and feasibility evolution, while the committed rules govern transfers and disclosure. The principal designs these rules so that the agent's optimal action coincides with the action the principal seeks to implement.

\begin{remark}[Action observability and hidden-action extensions]
\label{rem:hidden_action_extension}

The benchmark assumes that the realized action path is available for posterior updating and feasibility accounting. This action-revealing specification isolates the endogenous informational-evolution mechanism studied in the paper. More generally, hidden information-producing actions can be accommodated by introducing an observable performance process and the associated inference and dynamic-agency problem. For example, following the continuous-time moral-hazard formulation of \citet{HolmstromMilgrom1987}, one may introduce
\[
q_t=a_t+\varepsilon_t,
\]
where $\varepsilon_t$ represents performance noise. The benchmark corresponds to the perfectly revealing case $\varepsilon_t=0$, so that $q_t=a_t$. With nondegenerate performance noise, posterior updating and incentive provision depend on observable performance. The present paper does not pursue this additional layer in order to focus on the endogenous production, disclosure, and use of information.
\end{remark}

Dependence on beliefs, actions, and feasibility is left general to accommodate different applications. The institutional arrangement introduced below specifies transfer and disclosure rules, while the principal also selects the action process she seeks to implement. Together with the agent's action, the signal technology, Bayesian updating, and feasibility dynamics, these objects induce the informational and economic state processes analyzed in the remainder of the paper.

\subsection{Informational Dynamics and Feasibility States}
\label{subsec:informational_dynamics}

The dynamic state consists of posterior beliefs and the feasibility state, both introduced in Subsection~3.1. The former records what the economic system has learned about the informational environment, while the latter records the accumulated economic consequences of past activity.

Generated signals need not become public or contractible automatically. Their economic use is governed by the disclosure structure. In many applications, disclosure can be represented in reduced form by an intensity $\delta_t\in\mathcal D$, where $\mathcal D$ is a nonempty compact subset of $[0,1]$, measuring the extent to which generated evidence is validated, processed, communicated, or made available for contracting and decision making. The unrestricted benchmark is $\mathcal D=[0,1]$. When feasible, $\delta_t=1$ corresponds to full informational use, $\delta_t=0$ to complete suppression of informative content, and intermediate values to partial validation, selective communication, aggregation, noisy disclosure, or limited institutional use. Restricted disclosure sets accommodate institutional requirements on informational usage. For example, $\mathcal D=[\underline\delta,1]$ with $\underline\delta>0$ imposes a minimum-disclosure requirement, while $\mathcal D=\{1\}$ represents automatic or mandatory disclosure. This reduced-form representation is useful when the economic effect of disclosure depends primarily on the informativeness of the public signal rather than on the detailed communication protocol. Richer deterministic or randomized disclosure, reporting, validation, and information-processing rules are represented by $\pi_t$.\footnote{Allowing the feasible action or disclosure sets to depend on time or public states would not change the recursive logic, but is omitted to keep the notation focused on the interaction among information production, disclosure, incentives, and implementability.}

Let $Y_t=\pi_t(S_t)$ denote the publicly disclosed signal at date $t$. Through $Y_t$, the mechanism governs how raw information enters the public history and becomes available for contracting and economic use. Because action changes the distribution of the raw signal and disclosure changes its public informational content, action and disclosure jointly determine the path of learning.

The first component of the dynamic state is the \textbf{informational state} $\mu_t$. Let $h_t^\pi$ denote the public history generated by disclosed signals $(Y_s)_{s\leq t}$, transfers, and public state variables up to date $t$. We assume that, for each $\theta\in\Theta$, the history $h_t^\pi$ admits a likelihood $L_t(h_t^\pi\mid\theta)$ induced by the information-production technology $Q(\cdot\mid a_s,x_s,\theta)_{s\leq t}$, the disclosure structure $(\pi_s)_{s\leq t}$, the action process, transfer rules, and public state dynamics.

For every measurable set $B\subseteq\Theta$, the posterior belief is determined by Bayes' rule:
\begin{equation}
\label{eq:belief}
\mu_t(B)
=
\mathbb P(\theta\in B\mid h_t^\pi)
=
\frac{
\displaystyle\int_B L_t(h_t^\pi\mid\theta)\,dF(\theta)
}{
\displaystyle\int_\Theta L_t(h_t^\pi\mid\theta)\,dF(\theta)
},
\end{equation}
whenever the denominator is strictly positive. Equivalently, when $\Theta=\{\theta_1,\ldots,\theta_N\}$ is finite,
\[
\mu_t(\theta_i)
=
\frac{
F_iL_t(h_t^\pi\mid\theta_i)
}{
\sum_{j=1}^NF_jL_t(h_t^\pi\mid\theta_j)
},
\qquad i=1,\ldots,N,
\]
where $F_i=F(\theta_i)$.\footnote{In a discrete-time model with conditionally independent disclosed signals $(Y_s)$, the likelihood takes the product form
\[
L_t(h_t^\pi\mid\theta)
=
\prod_{s\leq t}
q_s^\pi(Y_s\mid a_s,x_s,\theta),
\]
where $q_s^\pi(\cdot\mid a_s,x_s,\theta)$ is the disclosed-signal distribution induced by the raw signal technology $Q(\cdot\mid a_s,x_s,\theta)$ and the disclosure rule $\pi_s$.}

Thus, $\mu_t$ is a probability measure over the state space of $\theta$ that summarizes the public informational state. It is induced by the information-production technology, action, disclosure, and Bayesian updating rather than directly selected by the principal. The updating and disclosure structure uses the standard language of statistical experiments and information structures; see \citet{Blackwell1953}, \citet{AumannMaschler1995}, and \citet{BergemannMorris2016, BergemannMorris2019}. The informational-environment state $\theta$ differs from the state in canonical information-design environments, where the underlying state is typically directly payoff relevant and the information structure itself is the principal design object.

The second component is the \textbf{feasibility state} $x_t$. It should be distinguished from the contemporaneous resource allocation $(a_t,\bar a-a_t)$. Whereas $a_t$ is a flow variable describing current productive-resource use, $x_t$ is a stock variable summarizing inherited economic conditions and determining future opportunities. Depending on the application, it may represent remaining inventories, testing capacity, productive capacity, organizational capital, knowledge stocks, reputational capital, financial resources, or future experimentation opportunities.

The feasibility state evolves according to
\begin{equation}
\label{eq:feasibility}
\dot x_t
=
G(x_t,a_t),
\qquad
x_0\in X=[0,\bar x],
\end{equation}
where $G:X\times\mathcal A\to\mathbb R$ is continuously differentiable and satisfies $G_a(x,a)\leq0$ for all $(x,a)$. Thus, higher current actions weakly reduce the evolution of future feasibility through resource use, capacity utilization, inventory reduction, financial-resource consumption, or the exhaustion of future opportunities. The dependence of $G$ on $x$ accommodates depreciation, recovery, congestion, scale effects, and technology-driven renewal. The feasibility path is required to satisfy $0\leq x_t\leq\bar x$ for all $t\in[0,\bar T]$. The depletion specification $\dot x_t=-g(a_t)$ is obtained as the special case $G(x,a)=-g(a)$, where $g(a)\geq0$ and $g'(a)\geq0$.

Thus, the pair $(\mu_t,x_t)$ constitutes the \textbf{dynamic state} of the informational economy. The informational state summarizes the accumulated informational consequences of past action and disclosure, while the feasibility state summarizes the accumulated economic consequences of past activity. Together they determine future opportunities, continuation incentives, feasible actions, and the value of additional information.

The preceding formulation permits general information-production technologies and posterior states. To obtain a tractable characterization of endogenous informational dynamics, part of the analysis specializes to a continuous-time \textbf{linear--Gaussian information-production technology}. This specification provides a finite-dimensional belief representation and permits the interaction among information production, disclosure, incentives, and feasibility to be characterized transparently. The conceptual framework and recursive logic do not depend on Gaussianity. Continuous-time Brownian environments with Gaussian updating are canonical in dynamic agency and contracting; see \citet{HolmstromMilgrom1987}, \citet{SchattlerSung1993}, \citet{Sannikov2008}, and \citet{Williams2011}. Related learning and experimentation models also use continuous-time belief-state representations; see \citet{BoltonHarris1999}, \citet{KellerRadyCripps2005}, and \citet{BergemannValimaki2008}.

Under this specification, economic activity generates the raw signal
\begin{equation}
\label{eq:linear_gaussian_raw_signal}
dS_t
=
a_t\theta\,dt+\sigma\,dW_t,
\end{equation}
where $\sigma>0$ and $(W_t)$ is a standard Brownian motion. Higher action increases the informativeness of the evidence generated about the unknown informational-environment state $\theta$.

To represent partial disclosure, let $(B_t)$ be a standard Brownian motion independent of $(W_t)$ and $\theta$. For each disclosure intensity $\delta_t\in[0,1]$, define the effective public signal process by
\begin{equation}
\label{eq:linear_gaussian_public_signal}
dY_t
=
\sqrt{\delta_t}\,dS_t
+
\sqrt{1-\delta_t}\,\sigma\,dB_t
=
\sqrt{\delta_t}\,a_t\theta\,dt+\sigma\,dW_t^Y,
\end{equation}
where $dW_t^Y=\sqrt{\delta_t}\,dW_t+\sqrt{1-\delta_t}\,dB_t$ defines a standard Brownian innovation process. Full disclosure corresponds to $\delta_t=1$, in which case $dY_t=dS_t$, whereas $\delta_t=0$ produces a public signal independent of $\theta$. Intermediate values represent Gaussian garblings of the action-generated evidence. Thus, $\delta_t$ is a reduced-form representation of the general disclosure rule $\pi_t$, rather than an additional primitive of the environment.

Suppose that the prior is Gaussian. Posterior beliefs then remain Gaussian, with $\mu_t=N(m_t,P_t)$. Conditional on the realized action path, standard filtering arguments give
\begin{equation}
\label{eq:gaussian_filter}
dm_t
=
\frac{P_t\sqrt{\delta_t}\,a_t}{\sigma^2}
\left(
dY_t-\sqrt{\delta_t}\,a_tm_t\,dt
\right),
\qquad
\dot P_t
=
-\delta_t\frac{P_t^2a_t^2}{\sigma^2}.
\end{equation}

Letting $\rho_t:=P_t^{-1}$ denote posterior precision and defining $\chi:=\sigma^{-2}$, we obtain
\begin{equation}
\label{eq:precision}
\dot\rho_t
=
\delta_t\chi a_t^2,
\qquad
\rho_t
=
\rho_0
+
\int_0^t
\delta_s\chi a_s^2\,ds.
\end{equation}
Thus, $\delta_t$ is a designed disclosure control, whereas $\rho_t$ is the induced informational state. Its evolution depends jointly on the action that produces evidence and the disclosure rule that determines its public informational content.

When continuation payoffs depend on posterior beliefs only through uncertainty, the posterior mean can be suppressed and posterior precision becomes the relevant informational state. The dynamic state is then represented by $(\rho_t,x_t)$. When payoffs depend on the estimated level of $\theta$ as well as its uncertainty, the posterior mean and precision, or more generally the full posterior distribution, must be retained.

The framework contains the principal benchmark cases summarized in Table~\ref{tab:framework_comparison}. If feasibility dynamics are absent, $G(x,a)\equiv0$, the model isolates information production, disclosure, and incentives and is related to dynamic feedback and disclosure design \citep{Ely2017,HornerLambert2021,ElyGeorgiadisRayo2025}. If action-generated information is automatically public, the separate disclosure control disappears and the model approaches learning-by-doing and strategic experimentation \citep{Arrow1962,BoltonHarris1999, KellerRadyCripps2005,BergemannValimaki2008,GuoHeLiu2023}. The framework also accommodates information acquisition, precision choice, monitoring, experimentation, and attention allocation, as well as private payoff-relevant information through the extension in Section~\ref{subsec:payoff_relevant_information_states}. The recursive characterization and value decomposition extend beyond Gaussian learning to finite-state, Poisson, and other informational environments, although their posterior-state laws and sufficient statistics differ.

\subsection{Institutional Arrangements and Dynamic Mechanisms}
\label{subsec:institutional_mechanisms}

The informational environment and the induced state dynamics do not by themselves determine economic outcomes. An institutional arrangement must specify how generated information is disclosed and how the agent is compensated, while the principal also determines which action process she seeks to implement.

Let $\mathcal F_t^\pi$ denote the public information available up to date $t$, formally represented by the filtration generated by the disclosed history $h_t^\pi$. Transfers must depend only on information available by date $t$, that is, they must be adapted to $\mathcal F_t^\pi$, and may depend on publicly observable state variables. In the general formulation, let $T_t=T(t,h_t^\pi,x_t)$ denote the flow transfer from the principal to the agent. Under the Markov state representation introduced below, this rule can be written as $T_t=T(t,\mu_t,x_t)$. Under the linear--Gaussian information-production technology, when posterior precision is a sufficient informational state, it reduces further to $T_t=T(t,\rho_t,x_t)$.

For any mechanism $\mathcal M$ and admissible action process $a$, let
$\mathbb P^{\mathcal M,a}$ denote the probability law of the induced stochastic
path generated jointly by the prior, the information-production technology,
the disclosure rule, Bayesian updating, and the feasibility dynamics.
Conditional laws are understood relative to the public information
$\mathcal F_t^\pi$.

A dynamic mechanism specifies disclosure and transfer rules:
\begin{equation}
\label{eq:mechanism}
\mathcal M
=
(\pi_t,T_t)_{t\in[0,\bar T]}.
\end{equation}
The disclosure and transfer rules are the principal's institutional instruments. Separately, let $(\hat a_t)_{t\in[0,\bar T]}$ denote the action process the principal seeks to implement. The agent chooses the actual action $(a_t)$. Public histories, posterior beliefs, and feasibility paths are not independently selected controls; they are induced by the primitive environment, the mechanism, the agent's action, Bayesian updating, and the feasibility law of motion.

Given a mechanism $\mathcal M$, the agent's problem is
\begin{equation}
\label{eq:agent_problem}
\max_{(a_t)}
\;
\int
\left[
\int_0^{\bar T}
\bigl(
u(T_t,\mu_t,x_t)-c(a_t)
\bigr)\,dt
+
U_{\bar T}(\mu_{\bar T},x_{\bar T})
\right]
\,d\mathbb P^{\mathcal M,a},
\end{equation}
where $(\mu_t,x_t)$ is the state process induced by the chosen action under $\mathcal M$. A deviation therefore changes not only the current cost of action, but also the induced probability law of future public signals, posterior beliefs, continuation values, and feasible opportunities.

The principal's problem is
\begin{equation}
\label{eq:principal_problem}
\max_{(\hat a_t,\pi_t,T_t)}
\;
\int
\left[
\int_0^{\bar T}
v(T_t,\hat a_t,\mu_t,x_t)\,dt
+
V_{\bar T}(\mu_{\bar T},x_{\bar T})
\right]
\,d\mathbb P^{\mathcal M,\hat a}
\end{equation}
subject to the agent's participation and action incentive constraints:
\[
\begin{array}{ll}
\displaystyle
\int
\left[
\int_0^{\bar T}
\bigl(
u(T_t,\mu_t,x_t)-c(\hat a_t)
\bigr)\,dt
+
U_{\bar T}(\mu_{\bar T},x_{\bar T})
\right]
\,d\mathbb P^{\mathcal M,\hat a}
\geq
\bar U,
& (\mathrm{IR}),
\\[1.2em]
\displaystyle
(\hat a_t)_{t\in[0,\bar T]}
\in
\arg\max_{(a_t)}
\int
\left[
\int_0^{\bar T}
\bigl(
u(T_t,\mu_t,x_t)-c(a_t)
\bigr)\,dt
+
U_{\bar T}(\mu_{\bar T},x_{\bar T})
\right]
\,d\mathbb P^{\mathcal M,a},
& (\mathrm{IC}).
\end{array}
\]
The benchmark imposes an ex ante participation constraint. The limited-commitment extension strengthens this requirement by imposing participation constraints along the continuation path.

An action process $(\hat a_t)$ is \textbf{implementable} under a mechanism $\mathcal M$ if it satisfies the action incentive-compatibility requirement under the induced transfer and disclosure rules. It is \textbf{individually rational} if the participation constraint is satisfied. The principal's feasible design set consists of the target action processes and mechanisms satisfying both requirements. Since action deviations affect information production and feasibility, implementability depends jointly on posterior evolution, continuation incentives, and future economic opportunities.

Disclosure affects the return to information-producing action, while action determines both the information available for disclosure and the evolution of future feasibility. The target action process, transfers, and disclosure rules must therefore be determined jointly. The mechanism and the agent's equilibrium behavior induce the resource allocation, informational allocation, and implementable continuation path analyzed recursively in the next section.

Throughout the general analysis, we maintain the following regularity conditions.

\begin{assumption}[Regularity Conditions for Primitives, State Dynamics, and Continuation Values]\rm
\label{ass:regularity_primitives}
The following regularity conditions hold:
\begin{enumerate}
\item[(i)] The flow and terminal payoff functions are continuously differentiable in their arguments, with uniformly bounded derivatives satisfying polynomial-growth conditions sufficient to justify the differentiation and conditional-expectation operations used below.

\item[(ii)] The information-production technology, disclosure and updating rules, and feasibility dynamics are measurable in time, continuous in controls, and Lipschitz continuous in state variables with linear-growth bounds sufficient to ensure well-defined state processes with the required finite moments.

\item[(iii)] Whenever a classical HJB representation or first-order condition is invoked, the associated continuation-value functions belong to the domains of the relevant controlled generators and possess the required continuous differentiability. The relevant generator mappings are continuous in controls and continuously differentiable on the interiors of their control sets.
\end{enumerate}
\end{assumption}

\begin{assumption}[Markov State Representation]\rm
\label{ass:markov_state_representation}
For every admissible mechanism, the induced posterior belief and feasibility state $(\mu_t,x_t)$ constitute a controlled Markov state for the continuation problems of the agent and the principal. Conditional on $(t,\mu_t,x_t)$ and the current controls, future payoff-relevant distributions depend on past histories only through $(\mu_t,x_t)$.
\end{assumption}

Assumption~\ref{ass:regularity_primitives}, together with the primitive restrictions above, ensures that the induced state dynamics and stochastic-control problems are well defined and provides the classical regularity required for HJB equations, continuation-value derivatives, and first-order conditions; see \citet{YongZhou1999}. Assumption~\ref{ass:markov_state_representation} provides the state-sufficiency condition for recursive analysis. Together, the two assumptions support the dynamic-programming principle and, where applicable, analysis through the associated Hamilton--Jacobi--Bellman equations and first-order conditions; see \citet{FlemingSoner2006} and \citet{YongZhou1999}.

\section{Characterization of Informational Implementability}
\label{sec:informational_implementability}

This section develops a recursive characterization of informational implementability when actions jointly affect information production and feasibility. The characterization combines the agent's continuation-value equation, a pointwise incentive condition, and the induced state dynamics. It provides the foundation for the principal's recursive problem and the value decomposition developed in the next section.

\subsection{Continuation Values and Informational Implementability}
\label{subsec:continuation_implementability}

Given a mechanism $\mathcal M=(\pi_t,T_t)_{t\in[0,\bar T]}$, the agent chooses
an admissible action process taking the induced disclosure and transfer rules as
given. For any continuation state $(t,\mu,x)$ and admissible continuation
action process $(a_s)_{s\in[t,\bar T]}$, let
$\mathbb P_{t,\mu,x}^{\mathcal M,a}$ denote the conditional law of the
continuation path induced by the mechanism $\mathcal M$ and action process
$a$, starting from $(\mu_t,x_t)=(\mu,x)$. Let $U(t,\mu,x)$ denote the agent's
continuation value at that state. Formally,
\begin{equation}
\label{eq:continuation_value}
U(t,\mu,x)
=
\sup_{(a_s)_{s\in[t,\bar T]}}
\int
\left[
\int_t^{\bar T}
\bigl(
u(T_s,\mu_s,x_s)-c(a_s)
\bigr)\,ds
+
U_{\bar T}(\mu_{\bar T},x_{\bar T})
\right]
\,d\mathbb P_{t,\mu,x}^{\mathcal M,a}.
\end{equation}

Let \(\mathcal L^{\pi,a}\) denote the infinitesimal generator associated with the posterior-belief dynamics induced by disclosure rule \(\pi\) and action \(a\).\footnote{For any function \(\phi\) in the domain \(\mathcal D(\mathcal L^{\pi,a})\) of the infinitesimal generator,
\[
\mathcal L^{\pi,a}\phi(t,\mu,x)
=
\lim_{h\downarrow0}
\frac{1}{h}
\mathbb E
\left[
\phi(t,\mu_{t+h},x)-\phi(t,\mu,x)
\,\middle|\,
(\mu_t,x_t)=(\mu,x)
\right].
\]
Time and the feasibility state are held fixed in this operator. Their evolution is accounted for separately through \(\partial_t\phi\) and \(\phi_xG(x,a)\). Thus the generator summarizes the local drift and uncertainty induced in the posterior process by the current disclosure and action controls. When the posterior admits a finite-dimensional diffusion representation, \(\mathcal L^{\pi,a}\) reduces to the standard controlled-diffusion generator.}
Economically, \(\mathcal L^{\pi,a}U\) measures the instantaneous expected change in the agent's continuation value generated by posterior evolution under the current disclosure rule and action. The notation \(\partial_t\) denotes the partial derivative with respect to time, holding \((\mu,x)\) fixed.

Under Assumptions~\ref{ass:regularity_primitives} and \ref{ass:markov_state_representation}, the continuation problem satisfies the dynamic programming principle, and its Hamilton--Jacobi--Bellman equation is
\begin{equation}
\label{eq:agent_hjb_general}
-\partial_t U(t,\mu,x)
=
\sup_{a\in\mathcal A}
\left\{
u(T(t,\mu,x),\mu,x)
-c(a)
+\mathcal L^{\pi,a}U(t,\mu,x)
+U_x(t,\mu,x)G(x,a)
\right\},
\end{equation}
with terminal condition
\begin{equation}
\label{eq:agent_terminal_general}
U(\bar T,\mu,x)
=
U_{\bar T}(\mu,x).
\end{equation}
Thus, the HJB separates the continuation effects of endogenous informational evolution and feasibility dynamics, represented by $\mathcal L^{\pi,a}U(t,\mu,x)$ and $U_x(t,\mu,x)G(x,a)$, respectively. Along an implementable path, the target action must attain the supremum in \eqref{eq:agent_hjb_general}.

Define the action-relevant continuation-value Hamiltonian by
\[
\mathcal H^A(t,\mu,x,a;U)
=
-c(a)
+\mathcal L^{\pi,a}U(t,\mu,x)
+U_x(t,\mu,x)G(x,a).
\]
Because $u(T(t,\mu,x),\mu,x)$ does not depend directly on the current action, it can be omitted from $\mathcal H^A$ without changing the agent's pointwise choice.


The following result identifies the pointwise action incentive condition implied by implementability.

\begin{lemma}[Pointwise Incentive Condition for Implementable Actions]
\label{lem:agent_local_ic}

Suppose Assumptions~\ref{ass:regularity_primitives} and \ref{ass:markov_state_representation} hold. Then an implementable action process $(\hat a_t)$ must satisfy, for almost every $t$,
\begin{equation}
\label{eq:local_ic_general}
\hat a_t
\in
\arg\max_{a\in\mathcal A}
\left\{
-c(a)
+\mathcal L^{\pi,a}U(t,\mu_t,x_t)
+U_x(t,\mu_t,x_t)G(x_t,a)
\right\}.
\end{equation}

If the maximizer is interior and the action-relevant continuation-value Hamiltonian is differentiable with respect to $a$, then it necessarily satisfies
\begin{equation}
\label{eq:agent_foc_general}
\frac{\partial}{\partial a}
\left[
-c(a)
+\mathcal L^{\pi,a}U(t,\mu_t,x_t)
+U_x(t,\mu_t,x_t)G(x_t,a)
\right]
=
0.
\end{equation}

If, in addition, the action-relevant continuation-value Hamiltonian is concave in $a$, or more generally its derivative satisfies a single-crossing property in $a$, being nonnegative below a candidate optimum and nonpositive above it, then the first-order optimality condition, interpreted as the corresponding Kuhn--Tucker condition on the boundary of $\mathcal A$, is sufficient for global pointwise optimality.

\end{lemma}

\begin{proof}
See Appendix~\ref{app:proof_agent_local_ic}.
\end{proof}

The pointwise incentive condition identifies two interacting continuation channels. First, action changes the distribution of posterior beliefs and therefore affects learning, informational evolution, and future incentives. Second, action changes the feasibility state and hence future economic opportunities. The sign of this feasibility channel depends on the marginal continuation value $U_x(t,\mu,x)$. When the continuation problem is monotone in the feasibility state, so that $U_x(t,\mu,x)\geq0$ wherever the derivative exists, $G_a(x,a)\leq0$ implies $-U_x(t,\mu,x)G_a(x,a)\geq0$, and the feasibility effect can be interpreted as a marginal continuation cost of action.

Under the linear--Gaussian information-production technology, posterior precision $\rho_t$ summarizes the relevant informational state, and the continuation value can be written as $U(t,\rho,x)$. For an interior action, the first-order condition becomes
\begin{equation}
\label{eq:agent_ic_gaussian}
c'(a_t)
-
U_x(t,\rho_t,x_t)G_a(x_t,a_t)
=
\frac{2\delta_ta_t}{\sigma^2}
U_\rho(t,\rho_t,x_t),
\end{equation}
where $U_\rho$ is the marginal continuation value of posterior precision.

Equation~\eqref{eq:agent_ic_gaussian} makes the informational incentive channel explicit. A higher action raises the drift of posterior precision at the marginal rate $2\delta_ta_t/\sigma^2$ and therefore generates the marginal informational continuation effect $(2\delta_ta_t/\sigma^2)U_\rho(t,\rho_t,x_t)$. The pointwise first-order condition balances this effect against the current marginal action cost $c'(a_t)$ and the feasibility term $-U_x(t,\rho_t,x_t)G_a(x_t,a_t)$. When $U_x(t,\rho_t,x_t)\geq0$, the feasibility term is nonnegative and can be interpreted as the marginal feasibility cost of action.

Because disclosure changes the informational return to action, it changes the set of actions supported by a given continuation plan. Information production, disclosure, and incentive provision therefore enter the implementability condition jointly.

\subsection{Recursive Characterization of Implementability}
\label{subsec:recursive_implementability}

The pointwise incentive condition can be combined with the continuation-value equation and the induced state dynamics to characterize implementable action and mechanism rules recursively.

\begin{proposition}[Recursive Characterization of Implementability]
\label{prop:general_characterization}

Suppose Assumptions~\ref{ass:regularity_primitives}--\ref{ass:markov_state_representation} hold. An action rule $a$ is implementable under mechanism rules $(T,\pi)$ if and only if there exists a continuation-value function $U(t,\mu,x)$ satisfying the following conditions:

\begin{enumerate}
\item[(i)]
The continuation value satisfies
\begin{equation}
\label{eq:prop_agent_hjb}
-\partial_tU(t,\mu,x)
=
u(T(t,\mu,x),\mu,x)
-c(a)
+\mathcal L^{\pi,a}U(t,\mu,x)
+U_x(t,\mu,x)G(x,a),
\end{equation}
with terminal condition
\begin{equation}
\label{eq:prop_agent_terminal}
U(\bar T,\mu,x)=U_{\bar T}(\mu,x).
\end{equation}

\item[(ii)]
The action rule satisfies the pointwise action incentive condition
\begin{equation}
\label{eq:prop_agent_lic}
a
\in
\arg\max_{\tilde a\in\mathcal A}
\left\{
-c(\tilde a)
+\mathcal L^{\pi,\tilde a}U(t,\mu,x)
+U_x(t,\mu,x)G(x,\tilde a)
\right\}
\end{equation}
for every admissible date-state tuple $(t,\mu,x)$.

\item[(iii)]
The informational and feasibility states evolve according to the state dynamics induced by $(a,T,\pi)$.
\end{enumerate}

Consequently, implementability admits a recursive continuation-value representation over the endogenous dynamic state $(\mu,x)$.

\end{proposition}

\begin{proof}
See Appendix~\ref{app:proof_general_characterization}.
\end{proof}

Proposition~\ref{prop:general_characterization} replaces the agent's global action problem with a continuation-value equation, a pointwise incentive condition, and the induced informational and feasibility dynamics. It characterizes incentive-compatible implementability; individual rationality remains a separate restriction on the principal's feasible design set, as specified in Subsection~\ref{subsec:institutional_mechanisms}. The same characterization therefore applies under both the benchmark ex ante participation constraint and the dynamic participation constraints introduced under limited commitment.

The recursive system also reveals why disclosure and incentive design need not be separable. Disclosure enters the informational generator $\mathcal L^{\pi,a}$ and changes both the agent's continuation value and the action incentive condition, while the resulting action changes information production and feasibility. Disclosure must therefore be evaluated together with the continuation plan required to implement the induced action and state path.

This characterization provides the foundation for the principal's recursive problem and the marginal-value decomposition in the next section. It is related to continuation-value approaches in dynamic mechanism design and dynamic information provision; see \citet{PavanSegalToikka2014,Ball2023}. The additional feature here is that information is produced through action, governed through disclosure, and linked to future feasibility within the implementability system.

\section{Optimal Contracting and the Value of Information}
\label{sec:optimal_contracting_information}

Building on the recursive characterization above, this section formulates the principal's problem over implementable action and mechanism rules. It establishes the incompatibility between exogenous-feasibility informational efficiency and endogenous dynamic feasibility, the general nonseparability of incentive and disclosure design, and the decomposition of the marginal value of information into a direct informational benefit and an implementation cost. Feasibility enters through the continuation state, the shadow value of future opportunities, and the implementable action path.

\subsection{Optimal Contracting under Endogenous Informational Evolution}
\label{subsec:principal_recursive_problem}

By Proposition~\ref{prop:general_characterization}, the principal's contracting
problem can be formulated directly over action and mechanism rules satisfying
the recursive continuation-value equation, pointwise incentive condition, and
induced state dynamics. For any continuation state \((t,\mu,x)\) and admissible
continuation plan \((\pi_s,T_s,a_s)_{s\in[t,\bar T]}\), let
\(\mathbb P_{t,\mu,x}^{\mathcal M,a}\) denote the conditional law of the
continuation path induced by the mechanism \(\mathcal M=(\pi,T)\) and action
process \(a\), starting from \((\mu_t,x_t)=(\mu,x)\). Let
\(J(t,\mu,x)\) denote the principal's continuation value at that state.

\begin{proposition}[Continuation-Value Formulation of the Principal's Problem]
\label{prop:principal_reduction}

Suppose Assumptions~\ref{ass:regularity_primitives}--\ref{ass:markov_state_representation} hold.  Then the principal's dynamic contracting problem is
equivalent to
\begin{equation}
J(t,\mu,x)
=
\sup_{(\pi_s,T_s,a_s)_{s\in[t,\bar T]}}
\int
\left[
\int_t^{\bar T}
v(T_s,a_s,\mu_s,x_s)\,ds
+
V_{\bar T}(\mu_{\bar T},x_{\bar T})
\right]
\,d\mathbb P_{t,\mu,x}^{\mathcal M,a},
\end{equation}
subject to the informational and feasibility dynamics induced by $(a,T,\pi)$,
the applicable participation requirement, and the existence of a
continuation-value function $U$ satisfying
\begin{align}
a_s
&\in
\arg\max_{\tilde a\in\mathcal A}
\left\{
-c(\tilde a)
+
\mathcal L^{\pi,\tilde a}U(s,\mu_s,x_s)
+
U_x(s,\mu_s,x_s)G(x_s,\tilde a)
\right\},
\label{eq:principal_local_ic_constraint}
\\
-\partial_s U(s,\mu,x)
&=
u(T(s,\mu,x),\mu,x)
-
c(a_s)
+
\mathcal L^{\pi,a_s}U(s,\mu,x)
+
U_x(s,\mu,x)G(x,a_s),
\label{eq:principal_agent_continuation}
\\
U(\bar T,\mu,x)
&=
U_{\bar T}(\mu,x).
\label{eq:principal_agent_terminal}
\end{align}

\end{proposition}

\begin{proof}
See Appendix~\ref{app:proof_principal_reduction}.
\end{proof}

Here and below, the applicable participation requirement depends on the commitment regime: under full commitment, only the initial participation constraint $U(0,\mu_0,x_0)\geq\bar U$ is imposed, whereas under limited commitment the continuation participation constraint $U(s,\mu_s,x_s)\geq\bar U(s,\mu_s,x_s)$ must hold along the continuation path. Proposition~\ref{prop:principal_reduction} reduces the principal's problem to a recursive optimization over implementable action and mechanism rules and endogenous informational and feasibility states. This representation is used below to study informational feasibility, the joint design of incentives and disclosure, and the marginal value of information.

\subsection{Two Fundamental Impossibility Results}
\label{subsec:fundamental_impossibility_results}

Because the informational state evolves endogenously, informational design must be evaluated jointly with incentive provision and feasibility management. This raises two questions: whether a mechanism can attain the informational frontier available under exogenous feasibility while preserving endogenous dynamic feasibility, and whether incentive and disclosure design can be conducted sequentially without loss. The results below show that neither property can be guaranteed throughout the class of environments considered in this paper.

\subsubsection{Impossibility of Exogenous-Feasibility Informational Efficiency and Endogenous Dynamic Feasibility}
\label{subsubsec:informational_feasibility}

Fix a primitive informational environment $(\Theta,F,Q)$, a continuation state $(t,\mu,x)$, and a terminal informational criterion $\Phi(\mu_{\bar T})$, normalized so that a larger value represents a better terminal informational outcome. Let $\mathfrak P^{\mathrm{exo}}(t,\mu,x)$ denote the set of admissible action--disclosure paths satisfying the primitive action restrictions, disclosure restrictions, and informational dynamics when feasibility is exogenous to information-producing activity. Define the \textbf{exogenous-feasibility informational frontier} by
\[
\mathcal I^{\mathrm{exo}}(t,\mu,x)
:=
\sup_{(a,\pi)\in\mathfrak P^{\mathrm{exo}}(t,\mu,x)}
\mathbb E_t
\left[
\Phi(\mu_{\bar T})
\right].
\]

Let $\mathfrak P^{\mathrm{endo}}(t,\mu,x) \subseteq \mathfrak P^{\mathrm{exo}}(t,\mu,x)$ denote the subset also satisfying the endogenous feasibility dynamics and $x_s\in X$ for all $s\in[t,\bar T]$. Define the \textbf{endogenous-feasibility informational frontier} by
\[
\mathcal I^{\mathrm{endo}}(t,\mu,x)
:=
\sup_{(a,\pi)\in\mathfrak P^{\mathrm{endo}}(t,\mu,x)}
\mathbb E_t
\left[
\Phi(\mu_{\bar T})
\right].
\]
Since  $\mathfrak P^{\mathrm{endo}}(t,\mu,x) \subseteq \mathfrak P^{\mathrm{exo}}(t,\mu,x)$, we have
\[
\mathcal I^{\mathrm{endo}}(t,\mu,x)
\leq
\mathcal I^{\mathrm{exo}}(t,\mu,x).
\]

A mechanism is \textbf{informationally efficient under exogenous feasibility} at $(t,\mu,x)$ if its induced continuation path attains $\mathcal I^{\mathrm{exo}}(t,\mu,x)$. It satisfies \textbf{endogenous dynamic feasibility} if its induced feasibility path remains in $X$ throughout the continuation horizon. Informational efficiency here refers to attainment of the specified terminal informational criterion and is distinct from informational-efficiency notions in realization theory based on the informational or dimensional requirements of decentralized processes.

If information-producing activity does not restrict feasibility, the endogenous-feasibility restriction is vacuous and  $\mathfrak P^{\mathrm{endo}}(t,\mu,x) = \mathfrak P^{\mathrm{exo}}(t,\mu,x)$, so the two frontiers coincide. When information-producing activity consumes a scarce continuation state, however, the endogenous-feasibility information frontier may lie strictly below the exogenous-feasibility frontier.

\begin{theorem}[Impossibility of Exogenous-Feasibility Informational Efficiency and Endogenous Dynamic Feasibility]
\label{thm:informational_scarcity_impossibility}

For the class of informational environments considered in this paper, there exists an informational environment and a continuation state $(t,\mu,x)$ such that
\[
\mathcal I^{\mathrm{endo}}(t,\mu,x)
<
\mathcal I^{\mathrm{exo}}(t,\mu,x).
\]
Consequently, no dynamic mechanism can guarantee both of the following properties throughout this class:
\begin{enumerate}
\item[(i)] informational efficiency under exogenous feasibility;

\item[(ii)] endogenous dynamic feasibility, $x_s\in X$ for all $s\in[t,\bar T]$.
\end{enumerate}

In particular, no implementable and individually rational dynamic mechanism can guarantee both properties throughout the class.

\end{theorem}

\begin{proof}
See Appendix~\ref{app:proof_informational_scarcity_impossibility}.
\end{proof}

Theorem~\ref{thm:informational_scarcity_impossibility} identifies a technological rather than contractual obstruction. When information production consumes scarce future opportunities, no transfer, disclosure rule, or continuation promise can reproduce an informational frontier that the endogenous feasibility state cannot sustain.

The Gaussian counterexample constructed in the proof of Theorem~\ref{thm:informational_scarcity_impossibility} establishes a strict frontier gap for every continuation state in the open set
\[
0<x<\bar a(\bar T-t).
\]
If $x\geq\bar a(\bar T-t)$, the remaining feasibility stock can sustain maximal action over the entire continuation horizon. If $x<\bar a(\bar T-t)$, it cannot, and the endogenously feasible informational frontier is strictly lower.

The magnitude of the informational restriction can be characterized for a broader scalar subclass. Suppose the endogenous posterior process admits a scalar sufficient statistic $z_s=Z(\mu_s)$ whose induced dynamics satisfy
\[
\dot z_s
=
\delta_s h(a_s),
\qquad
\dot x_s
=
-g(a_s),
\]
where $a_s\in[0,\bar a]$, $\delta_s\in[0,1]$, $h(a)\geq0$, $g(a)\geq0$, $h(0)=g(0)=0$, and $g(a)>0$ for $a>0$.\footnote{The linear--Gaussian benchmark is the special case $z_s=\rho_s$, $h(a)=\chi a^2$, and $g(a)=a$, so that $\dot\rho_s=\delta_s\chi a_s^2$ and $\dot x_s=-a_s$.} Define
\[
\Lambda
:=
\sup_{a\in(0,\bar a]}
\frac{h(a)}{g(a)}
<
\infty,
\qquad
\bar h
:=
\max_{a\in[0,\bar a]}h(a).
\]

\begin{lemma}[General Informational Feasibility Bound]
\label{lem:general_informational_feasibility}

For any admissible continuation path from date $t$ to $\bar T$, the increase in the scalar informational sufficient statistic satisfies
\begin{equation}
\label{eq:general_information_feasibility_bound}
z_{\bar T}-z_t
\leq
\min
\left\{
\Lambda\bigl(x_t-x_{\bar T}\bigr),
\bar h(\bar T-t)
\right\}
\leq
\min
\left\{
\Lambda x_t,
\bar h(\bar T-t)
\right\}.
\end{equation}

\end{lemma}

\begin{proof}
See Appendix~\ref{app:proof_general_informational_feasibility}.
\end{proof}

The first component of the bound is determined by the remaining feasibility stock, whereas the second is determined by the remaining time horizon. Under full disclosure, the feasibility-stock component is attained when the supremum defining $\Lambda$ is attained at some $a^*>0$ and the horizon permits the relevant stock to be used at that action. The time component is attained when $\bar h$ is attained at some $a^h$ and the feasibility stock can sustain that action until $\bar T$. The bound is therefore sharp under the conditions associated with its binding component.

In the exogenous-feasibility benchmark, full disclosure and the action attaining $\bar h$ generate the informational increment $\bar h(\bar T-t)$. Under endogenous feasibility, the increment cannot exceed $\Lambda x_t$. Consequently, $\Lambda x_t < \bar h(\bar T-t)$ implies a strict gap between the two informational frontiers. The incompatibility therefore applies to a broad class of information-production and feasibility technologies.

The linear--Gaussian information-production technology yields the exact attainable informational set. Let $\mathcal R(t,\rho,x)$ denote the set of technologically attainable terminal precision levels from continuation state $(t,\rho,x)$.

\begin{corollary}[Sharpness of the Gaussian Informational Frontier]
\label{cor:gaussian_informational_frontier}

Under
\[
\dot\rho_s
=
\delta_s\chi a_s^2,
\qquad
\dot x_s
=
-a_s,
\qquad
a_s\in[0,\bar a],
\qquad
\delta_s\in[0,1],
\qquad
x_s\geq0,
\]
the attainable terminal-precision set is
\begin{equation}
\label{eq:gaussian_attainable_precision_set}
\mathcal R(t,\rho,x)
=
\left[
\rho,\,
\rho
+
\chi\bar a
\min
\left\{
x,\,
\bar a(\bar T-t)
\right\}
\right].
\end{equation}
Equivalently, a target $\widehat\rho\geq\rho$ is technologically attainable if and only if
\begin{equation}
\label{eq:gaussian_target_feasibility}
\widehat\rho-\rho
\leq
\chi\bar a
\min
\left\{
x,\,
\bar a(\bar T-t)
\right\}.
\end{equation}
If the corresponding full-disclosure maximal-action path is implementable and individually rational for the required active duration, then the target is implementable.
\end{corollary}

\begin{proof}
See Appendix~\ref{app:proof_gaussian_informational_feasibility}.
\end{proof}

For the terminal-precision criterion,
\[
\mathcal I^{\mathrm{endo}}(t,\rho,x)
=
\rho
+
\chi\bar a
\min
\left\{
x,\,
\bar a(\bar T-t)
\right\},
\]
whereas
\[
\mathcal I^{\mathrm{exo}}(t,\rho,x)
=
\rho
+
\chi\bar a^2(\bar T-t).
\]
Hence
\[
\mathcal I^{\mathrm{endo}}(t,\rho,x)
<
\mathcal I^{\mathrm{exo}}(t,\rho,x)
\quad\Longleftrightarrow\quad
x
<
\bar a(\bar T-t).
\]
Corollary~\ref{cor:gaussian_informational_frontier} therefore gives the exact boundary between the region in which endogenous feasibility restricts informational attainment and the region in which the exogenous-feasibility frontier is technologically attainable.

Theorem~\ref{thm:informational_scarcity_impossibility} and Corollary~\ref{cor:gaussian_informational_frontier} concern the amount of information that can be generated. Section~6 shows how scarcity of future informational opportunities also affects the timing of information production and disclosure.

\subsubsection{Impossibility of Separating Incentive and Disclosure Design}
\label{subsubsec:design_nonseparability}

In environments with a fixed informational structure, incentive design can be conducted conditional on that structure, and disclosure may be treated separately when it does not alter implementability. Here, disclosure changes the continuation return to information-producing action and may therefore change both the implemented action and the resulting informational and feasibility paths.

For a continuation-value function $U$, define the implementable-action correspondence under disclosure rule $\pi$ by
\begin{equation}
\label{eq:implementable_action_correspondence}
\mathcal A^I(t,\mu,x;U,\pi)
:=
\arg\max_{a\in\mathcal A}
\left\{
-c(a)
+
\mathcal L^{\pi,a}U(t,\mu,x)
+
U_x(t,\mu,x)G(x,a)
\right\}.
\end{equation}
A prescribed action $a^I$ satisfies the \textbf{pointwise incentive condition} under $\pi$ at $(t,\mu,x)$ for continuation plan $U$ if $a^I\in\mathcal A^I(t,\mu,x;U,\pi)$. The \textbf{local separation property} between disclosure rules $\pi$ and $\pi'$ holds for $a^I$ at $(t,\mu,x)$ if the same action satisfies the pointwise incentive condition under both rules with $U$ held fixed. Thus, ``local'' refers to separation at a fixed date-state tuple rather than to a neighborhood of actions, states, or disclosure rules. A sequential design procedure is \textbf{disclosure-independent} if it first selects an action and incentive-continuation plan and subsequently optimizes disclosure while retaining that action and plan.

\begin{theorem}[Impossibility of Separating Incentive and Disclosure Design]
\label{thm:design_nonseparability}

For the class of informational environments considered in this paper, no sequential dynamic design procedure can guarantee both of the following properties:
\begin{enumerate}
\item[(i)] an action and incentive-continuation plan specified independently of the disclosure rule remains implementable as the disclosure rule varies;

\item[(ii)] holding that action and incentive-continuation plan fixed, optimizing disclosure attains the value of the joint incentive-and-disclosure design problem.
\end{enumerate}

Consequently, incentive and disclosure design cannot generally be separated.
\end{theorem}

\begin{proof}
See Appendix~\ref{app:proof_design_nonseparability}.
\end{proof}

Theorem~\ref{thm:design_nonseparability} compares joint and sequential incentive-and-disclosure design. Under joint design, the principal optimizes disclosure and the incentive-continuation plan together, so the continuation plan may adjust with the disclosure rule to preserve implementation. Under sequential design, the principal first selects an action and incentive-continuation plan and then optimizes disclosure while holding that plan fixed. The theorem shows that this sequential procedure may be lossy because changing disclosure can alter the agent's incentives, so the initially selected action and continuation plan need not remain implementable or optimal.

The following proposition gives the exact pointwise criterion for retaining a prescribed action and identifies stronger conditions for lossless sequential design.

\begin{proposition}[Criterion for Sequential Separation]
\label{prop:separation_criterion}

Suppose Assumptions~\ref{ass:regularity_primitives}--\ref{ass:markov_state_representation} hold.
\begin{enumerate}
\item[(i)] For a prescribed action $a^I$, the local separation property between disclosure rules $\pi$ and $\pi'$ holds at $(t,\mu,x)$ under continuation plan $U$ if and only if
\begin{equation}
\label{eq:separation_invariance}
a^I
\in
\mathcal A^I(t,\mu,x;U,\pi)
\cap
\mathcal A^I(t,\mu,x;U,\pi').
\end{equation}

\item[(ii)] A sufficient condition for invariance of the pointwise incentive condition across disclosure rules is
\[
\mathcal A^I(t,\mu,x;U,\pi)
=
\mathcal A^I(t,\mu,x;U,\pi')
\]
for every pair of admissible disclosure rules $\pi$ and $\pi'$ relevant to the continuation problem.

\item[(iii)] The joint incentive-and-disclosure design problem can be solved sequentially without loss if, at every reachable continuation state, the retained incentive-continuation plan remains feasible and implementable as disclosure varies, and the principal's feasible set and continuation objective are separable between the incentive-continuation plan and the disclosure rule.

\item[(iv)] If
\[
\mathcal A^I(t,\mu,x;U,\pi)
\cap
\mathcal A^I(t,\mu,x;U,\pi')
=
\varnothing
\]
for some pair of admissible disclosure rules at a reachable continuation state, then no action satisfies the pointwise incentive condition under both disclosure rules with the same continuation-value function. The local separation property therefore fails for every prescribed action at that state.
\end{enumerate}

\end{proposition}

\begin{proof}
See Appendix~\ref{app:proof_separation_criterion}.
\end{proof}

Part~(i) gives the exact pointwise condition for retaining a prescribed action. The equality in part~(ii) is stronger because it preserves the entire implementable-action correspondence rather than only a particular prescribed action.

Under the linear--Gaussian information-production technology, the exact local separation boundary for a positive interior action has a particularly simple form.

\begin{corollary}[Exact Local Separation Condition under Gaussian Learning]
\label{cor:gaussian_nonseparability}

Suppose Assumptions~\ref{ass:regularity_primitives}--\ref{ass:markov_state_representation} hold and the information-production technology satisfies the linear--Gaussian specification. Fix a date-state tuple $(t,\rho,x)$, a differentiable continuation-value function $U(t,\rho,x)$, and two disclosure intensities $\delta,\delta'\in[0,1]$. Let $a\in\operatorname{int}\mathcal A$ satisfy the pointwise incentive condition under disclosure intensity $\delta$. Suppose further that, under disclosure intensity $\delta'$, the derivative of the action-relevant continuation-value Hamiltonian satisfies the single-crossing property stated in Lemma~\ref{lem:agent_local_ic}.

Then the same action $a$ satisfies the pointwise incentive condition under disclosure intensity $\delta'$ with the same continuation-value function if and only if
\begin{equation}
\label{eq:gaussian_separation_condition}
(\delta-\delta')aU_\rho(t,\rho,x)
=
0.
\end{equation}

\end{corollary}

\begin{proof}
See Appendix~\ref{app:proof_gaussian_nonseparability}.
\end{proof}

Corollary~\ref{cor:gaussian_nonseparability} gives the exact local separation boundary for positive interior actions. Because $a>0$, the same action can satisfy the pointwise incentive condition under distinct disclosure intensities only if posterior precision has no marginal continuation value for the agent. Whenever $\delta\neq\delta'$ and $U_\rho(t,\rho,x)\neq0$, disclosure changes the agent's marginal return to action and the local separation property fails.

The corollary does not characterize boundary actions such as $a=0$, whose incentive compatibility is governed by boundary inequalities rather than the interior first-order condition. More generally, separation is restored when disclosure does not enter the action incentive condition, information production is independent of action, or posterior beliefs do not affect the agent's continuation value.

\subsection{The Marginal Value of Information}
\label{subsec:marginal_value_information}

The continuation-value formulation developed above applies to a general informational environment. A perturbation of the informational state affects the principal's continuation value through two broad channels. The first is the \textbf{direct informational-benefit channel}: posterior beliefs affect future decisions, learning opportunities, continuation payoffs, and the value of future feasibility. The second is the \textbf{implementability-adjustment channel}: changes in the informational state may require adjustments in transfers, actions, disclosure rules, continuation promises, participation margins, and feasibility paths to preserve implementability.

Throughout this subsection, we restrict attention to informational perturbations evaluated around an optimal implementable continuation plan. Fix such an optimal continuation plan $(a_s,T_s,\pi_s)_{s\in[t,\bar T]}$, with associated continuation-value process $(U_s)_{s\in[t,\bar T]}$ and induced feasibility-state trajectory $(x_s)_{s\in[t,\bar T]}$ satisfying $\dot x_s=G(x_s,a_s)$. Consider an admissible differentiable perturbation $\{\mu_s^\varepsilon\}_{s\in[t,\bar T]}$ of the posterior process, with $\mu_s^0=\mu_s$. Let $(a_s^\varepsilon,T_s^\varepsilon,\pi_s^\varepsilon)_{s\in[t,\bar T]}$, $(U_s^\varepsilon)_{s\in[t,\bar T]}$, and $(x_s^\varepsilon)_{s\in[t,\bar T]}$ denote the corresponding implementability-preserving continuation choices and state trajectory, with
\[
(a_s^0,T_s^0,\pi_s^0,U_s^0,x_s^0)
=
(a_s,T_s,\pi_s,U_s,x_s).
\]
Represent the perturbation family on a common conditional probability space and write
$\mathbb P_t^*:=\mathbb P_{t,\mu_t,x_t}^{\mathcal M,a}$ for the conditional continuation law generated by the unperturbed optimal plan. All first-order expressions below are evaluated at $\varepsilon=0$ under this baseline law.

Under the standard envelope condition, the first-order effect of further reoptimization around the optimal unperturbed continuation plan vanishes. The marginal value generated by the implementability-preserving informational perturbation can therefore be written directly as the directional G\^{a}teaux derivative of the principal's optimized continuation value:
\begin{equation}
\label{eq:general_information_derivative}
D_\mu J(t,\mu_t,x_t)
:=
\left.
\frac{d}{d\varepsilon}
J(t,\mu_t^\varepsilon,x_t)
\right|_{\varepsilon=0}.
\end{equation}

Define the associated first-order perturbation directions by
\[
\begin{aligned}
h_s^\mu
&:=
\left.
\frac{d\mu_s^\varepsilon}{d\varepsilon}
\right|_{\varepsilon=0},
&
h_s^T
&:=
\left.
\frac{dT_s^\varepsilon}{d\varepsilon}
\right|_{\varepsilon=0},
\\
h_s^a
&:=
\left.
\frac{da_s^\varepsilon}{d\varepsilon}
\right|_{\varepsilon=0},
&
h_s^x
&:=
\left.
\frac{dx_s^\varepsilon}{d\varepsilon}
\right|_{\varepsilon=0}.
\end{aligned}
\]

The resulting marginal value of information admits the accounting representation
\begin{equation}
\label{eq:general_information_decomposition}
D_\mu J(t,\mu_t,x_t)
=
\mathcal B_\mu(t,\mu_t,x_t)
-
\mathcal C_\mu(t,\mu_t,x_t),
\end{equation}
where
\[
\mathcal B_\mu(t,\mu_t,x_t)
=
\int
\left[
\int_t^{\bar T}
\left\langle
D_\mu v(T_s,a_s,\mu_s,x_s),
h_s^\mu
\right\rangle\,ds
+
\left\langle
D_\mu V_{\bar T}(\mu_{\bar T},x_{\bar T}),
h_{\bar T}^\mu
\right\rangle
\right]
\,d\mathbb P_t^*
\]
collects the direct marginal informational effects on the principal's primitive flow and terminal payoffs, and
\[
\begin{aligned}
\mathcal C_\mu(t,\mu_t,x_t)
=
-\int
\bigg[
&\int_t^{\bar T}
\Bigl(
v_T(T_s,a_s,\mu_s,x_s)h_s^T
+
v_a(T_s,a_s,\mu_s,x_s)h_s^a
\\
&\hspace{5.5em}
+
v_x(T_s,a_s,\mu_s,x_s)h_s^x
\Bigr)\,ds
+
V_{\bar T,x}(\mu_{\bar T},x_{\bar T})h_{\bar T}^x
\bigg]
\,d\mathbb P_t^*,
\end{aligned}
\]
where $\langle\cdot,\cdot\rangle$ denotes the duality pairing between the derivative with respect to the posterior and the corresponding perturbation direction, playing the role of the usual inner-product pairing when the posterior state is infinite-dimensional.

The term $\mathcal C_\mu$ collects the continuation-plan adjustments required to preserve implementability. Because these adjustments may operate through transfers, actions, and feasibility, $\mathcal C_\mu$ may be positive or negative and need not admit a purely monetary interpretation. A negative value means that the implementability-preserving adjustments themselves raise the principal's continuation value. The marginal value of information may therefore differ from its direct decision value.

To obtain the sharper monetary decomposition used in the disclosure analysis, we now impose quasi-linear utility. Under the envelope evaluation above, the implemented action and disclosure rules are held at their unperturbed optimal values, while transfers and continuation promises adjust to preserve incentive compatibility and participation. Because the feasibility law depends on the implemented action, the feasibility path is also unchanged.

\begin{assumption}[Quasi-Linear Utility Specification]
\label{ass:quasilinear}

The agent's flow utility and the principal's flow payoff take the form
\[
u(\mu_t,x_t)+T_t-c(a_t)
\]
and
\[
v(a_t,\mu_t,x_t)-T_t,
\]
respectively. The terminal payoff functions are quasi-linear in the terminal transfer:
\[
U_{\bar T}(\mu,x)
=
u_{\bar T}(\mu,x)+T_{\bar T}(\mu,x),
\qquad
V_{\bar T}(\mu,x)
=
v_{\bar T}(\mu,x)-T_{\bar T}(\mu,x),
\]
where $u_{\bar T}$ and $v_{\bar T}$ are transfer-independent terminal payoff components.
\end{assumption}

Along an optimal implementable continuation path $(a_s,\pi_s,T_s)_{s\in[t,\bar T]}$, the principal's continuation value takes the form
\begin{equation}
\label{eq:principal_quasilinear_value}
J(t,\mu,x)
=
\int
\left[
\int_t^{\bar T}
\bigl(
v(a_s,\mu_s,x_s)-T_s
\bigr)\,ds
+
v_{\bar T}(\mu_{\bar T},x_{\bar T})
-
T_{\bar T}
\right]
\,d\mathbb P_{t,\mu,x}^{\mathcal M,a}.
\end{equation}
Equivalently, $J$ satisfies the recursive equation
\begin{equation}
\label{eq:principal_hjb_quasilinear}
-\partial_tJ(t,\mu,x)
=
v(a_t,\mu,x)-T_t
+
\mathcal L^{\pi,a_t}J(t,\mu,x)
+
J_x(t,\mu,x)G(x,a_t),
\end{equation}
with terminal condition
\begin{equation}
\label{eq:principal_terminal_quasilinear}
J(\bar T,\mu,x)
=
v_{\bar T}(\mu,x)-T_{\bar T}(\mu,x).
\end{equation}

Defining the full-state infinitesimal generator by
\[
\overline{\mathcal L}^{\pi,a}J(t,\mu,x)
:=
\mathcal L^{\pi,a}J(t,\mu,x)
+
J_x(t,\mu,x)G(x,a),
\]
equation~\eqref{eq:principal_hjb_quasilinear} becomes
\[
-\partial_tJ(t,\mu,x)
=
v(a_t,\mu,x)-T_t
+
\overline{\mathcal L}^{\pi,a_t}J(t,\mu,x).
\]

\begin{theorem}[General Decomposition of the Marginal Value of Information]
\label{thm:general_decomposition}

Suppose Assumptions~\ref{ass:regularity_primitives}--\ref{ass:quasilinear} hold. Then, along an optimal implementable continuation path,
\begin{equation}
\label{eq:general_value_decomposition}
D_\mu J(t,\mu_t,x_t)
=
B_\mu(t,\mu_t,x_t)
-
K_\mu(t,\mu_t,x_t),
\end{equation}
where
\[
B_\mu
=
\left.
\frac{d}{d\varepsilon}
\int
\left[
\int_t^{\bar T}
v(a_s,\mu_s^\varepsilon,x_s)\,ds
+
v_{\bar T}(\mu_{\bar T}^\varepsilon,x_{\bar T})
\right]
\,d\mathbb P_t^*
\right|_{\varepsilon=0},
\]
and
\[
K_\mu
=
\left.
\frac{d}{d\varepsilon}
\int
\left[
\int_t^{\bar T}
T_s^\varepsilon\,ds
+
T_{\bar T}^\varepsilon
\right]
\,d\mathbb P_t^*
\right|_{\varepsilon=0}.
\]
Consequently,
\begin{equation}
\label{eq:general_value_boundary}
D_\mu J(t,\mu_t,x_t)
\gtreqless 0
\quad\Longleftrightarrow\quad
B_\mu(t,\mu_t,x_t)
\gtreqless
K_\mu(t,\mu_t,x_t).
\end{equation}

\end{theorem}

\begin{proof}
See Appendix~\ref{app:proof_general_decomposition}.
\end{proof}

The term $B_\mu$ is the direct decision benefit of the informational perturbation along the optimal implemented path, while $K_\mu$ is the monetary adjustment required to preserve incentive compatibility and participation. Feasibility affects both components through the continuation state, the shadow value of future opportunities, and the transfers required to implement the fixed action and feasibility paths. The decomposition is related to continuation-value and envelope approaches in dynamic mechanism design; see \citet{PavanSegalToikka2014} and \citet{GarrettPavan2015}. It is also related to dynamic information provision, where continuation information itself can serve as an incentive instrument; see \citet{Ball2023}. Its distinct feature is that the informational state is produced through action, governed through disclosure, and linked to future feasibility.

Under the linear--Gaussian information-production technology, posterior beliefs admit a finite-dimensional representation. For the remainder of the precision-based analysis, suppose continuation payoffs depend on posterior beliefs only through posterior precision. Under this restriction, $\rho_t$ is the relevant informational state and the principal's optimized continuation value can be written as $J(t,\rho,x)$. Normalize the perturbation parameter so that $\left.d\rho_t^\varepsilon/d\varepsilon\right|_{\varepsilon=0}=1$. Under this normalization, the directional derivative with respect to the informational state reduces to the ordinary partial derivative with respect to $\rho$:
\begin{equation}
\label{eq:J_rho_def}
D_\rho J(t,\rho,x)
=
\frac{\partial J(t,\rho,x)}{\partial\rho}
\equiv
J_\rho(t,\rho,x).
\end{equation}

\begin{proposition}[Value Decomposition under the Linear--Gaussian Benchmark]
\label{prop:gaussian_decomposition}
Suppose Assumptions~\ref{ass:regularity_primitives}--\ref{ass:quasilinear} hold and the information-production technology satisfies the linear--Gaussian specification
\[
dS_t=a_t\theta\,dt+\sigma\,dW_t.
\]
Then, along the optimal implementable continuation path and for any admissible differentiable compensated perturbation of posterior precision of the type described above,
\begin{equation}
\label{eq:gaussian_value_decomposition}
J_\rho(t,\rho_t,x_t)
=
B_\rho(t,\rho_t,x_t)
-
K_\rho(t,\rho_t,x_t),
\end{equation}
where
\[
B_\rho
=
\left.
\frac{\partial}{\partial\varepsilon}
\int
\left[
\int_t^{\bar T}
v(a_s,\rho_s^\varepsilon,x_s)\,ds
+
v_{\bar T}(\rho_{\bar T}^\varepsilon,x_{\bar T})
\right]
\,d\mathbb P_t^*
\right|_{\varepsilon=0},
\]
and
\[
K_\rho
=
\left.
\frac{\partial}{\partial\varepsilon}
\int
\left[
\int_t^{\bar T}
T_s^\varepsilon\,ds
+
T_{\bar T}^\varepsilon
\right]
\,d\mathbb P_t^*
\right|_{\varepsilon=0}.
\]

\end{proposition}

\begin{proof}
See Appendix~\ref{app:proof_gaussian_decomposition}.
\end{proof}

Proposition~\ref{prop:gaussian_decomposition} identifies $B_\rho$ as the gross decision value of posterior precision and $K_\rho$ as the monetary adjustment required to preserve implementability and participation. Since $\dot\rho_t=\delta_ta_t^2/\sigma^2$, disclosure changes the rate of public learning, while feasibility affects the informational benefit and implementation cost through continuation values and the implementable action path.

When feasibility dynamics are absent, so that $G(x,a)\equiv0$, the state argument $x$ can be suppressed. If information is generated exogenously and disclosure does not enter the agent's incentive condition, the framework reduces to information design or dynamic information provision under a fixed signal technology. If disclosure is fixed but action determines information production, the problem corresponds to delegated experimentation or learning; if neither action nor disclosure affects information, it reduces further to dynamic contracting under a fixed information structure.

\subsection{Impossibility of General Informational Monotonicity}
\label{subsec:informational_monotonicity}

The decomposition in Subsection~\ref{subsec:marginal_value_information} shows that the value of information depends on both its direct decision benefit and the cost of preserving implementability. This raises a more global question: does a more informative information-production technology necessarily increase the principal's maximal implementable continuation value?

In canonical information-design models, free informational disposal ensures monotonicity because the designer can garble or ignore additional information and reproduce any outcome available under a less informative technology. Here, however, the informational state evolves jointly with action, incentives, and feasibility. A more informative technology may change implementation costs or the evolution of the economic system, while institutional restrictions may prevent replication of outcomes attainable under the less informative technology.

Consider two information-production technologies $Q^H$ and $Q^L$, holding all other payoff, action, feasibility, contracting, and commitment primitives fixed. We say that $Q^H$ \textbf{Blackwell dominates} $Q^L$, and write $Q^H\succeq_B Q^L$, if, conditional on every admissible action and state path, the experiment over future raw-signal histories induced by $Q^L$ is a garbling of the experiment induced by $Q^H$. Write $J(t,\mu,x;Q)$ for the principal's maximal implementable continuation value under information-production technology $Q$.

\begin{definition}[Informational Monotonicity Principle]
\label{def:informational_monotonicity}

The \textbf{informational monotonicity principle} holds if, for every pair of information-production technologies $Q^H$ and $Q^L$ satisfying $Q^H\succeq_B Q^L$ and every continuation state $(t,\mu,x)$,
\begin{equation}
\label{eq:informational_monotonicity_principle}
J(t,\mu,x;Q^H)
\geq
J(t,\mu,x;Q^L).
\end{equation}

\end{definition}

\begin{theorem}[Impossibility of General Informational Monotonicity]
\label{thm:informational_monotonicity_impossibility}

For the class of informational environments considered in this paper, the informational monotonicity principle cannot be guaranteed: there exist two information-production technologies $Q^H$ and $Q^L$, holding all other primitives fixed, and a continuation state $(t,\mu,x)$ such that
\begin{equation}
\label{eq:blackwell_dominance_reversal}
Q^H
\succeq_B
Q^L,
\qquad
J(t,\mu,x;Q^H)
<
J(t,\mu,x;Q^L).
\end{equation}

\end{theorem}

\begin{proof}
See Appendix~\ref{app:proof_informational_monotonicity_impossibility}.
\end{proof}

The reversal is not a knife-edge construction: subject to the stated interiority condition, it holds throughout an open set of parameter values in the counterexample constructed in Appendix~\ref{app:proof_informational_monotonicity_impossibility}. Nor does Theorem~\ref{thm:informational_monotonicity_impossibility} contradict the conventional value-of-information principle. In the counterexample, disclosure is bounded away from zero, so greater raw-signal informativeness necessarily increases public learning for every positive action. Because public learning imposes a continuation cost on the agent, preserving participation and implementation requires additional compensation. When this implementation cost exceeds the direct decision benefit, the principal's optimized value falls.

The counterexample identifies the missing condition behind the conventional monotonicity argument: the ability to reproduce under the more informative technology the implementable outcomes available under the less informative technology. This motivates the following notion.

\begin{definition}[Implementation Replicability]
\label{def:implementation_replicability}

Technology $Q^H$ \textbf{implementation-replicates} technology $Q^L$ at continuation state $(t,\mu,x)$ if every implementable continuation plan under $Q^L$ can be replicated under $Q^H$ so as to induce the same joint distribution of actions, public posterior beliefs, feasibility states, transfers, and payoff-relevant outcomes.

\end{definition}

Implementation replicability implies that the principal's feasible set under $Q^L$ is contained in its feasible set under $Q^H$, and therefore
\begin{equation}
\label{eq:replicability_monotonicity}
J(t,\mu,x;Q^H)
\geq
J(t,\mu,x;Q^L).
\end{equation}
The substantive question is thus whether the information-production technology and available disclosure instruments permit such replication.

\begin{proposition}[Gaussian Replicability and Informational Monotonicity]
\label{prop:replicability_informational_monotonicity}

Consider two linear--Gaussian information-production technologies that differ only in their informativeness parameters, with $\chi_H>\chi_L$. Suppose disclosure can be chosen freely from $[0,1]$, and payoffs and implementation constraints depend on the technology only through the induced public posterior process. Then $Q^H$ implementation-replicates $Q^L$, and informational monotonicity holds:
\[
J(t,\mu,x;Q^H)
\geq
J(t,\mu,x;Q^L).
\]

\end{proposition}

\begin{proof}
See Appendix~\ref{app:proof_gaussian_replication_monotonicity}.
\end{proof}

Taken together, Theorem~\ref{thm:informational_monotonicity_impossibility} and Proposition~\ref{prop:replicability_informational_monotonicity} identify the boundary of informational monotonicity. Blackwell dominance alone is insufficient, while implementation replicability restores monotonicity. Replicability may fail when the more informative technology cannot be attenuated or otherwise adjusted to reproduce the informational and implementation outcomes attainable under the less informative technology. This can occur under automatic or mandatory disclosure, a positive lower bound on disclosure, restrictions on garbling or attenuation, indivisible information production, private exposure to raw signals, or other institutional constraints that cause greater informativeness to alter implementation or feasibility.

These results concern the principal's optimized value across information-production technologies. Proposition~\ref{prop:signal_informativeness} below addresses the distinct local question of how signal informativeness changes disclosure timing along an interior cutoff path.

\section{Optimal Disclosure and Informational Usage}
\label{sec:optimal_disclosure}

We now characterize optimal disclosure along implementable continuation paths. Under general disclosure technologies, optimal disclosure can be interior or extreme, and bang--bang disclosure can arise even under nonlinear reduced objectives. We then specialize to an affine reduced-Hamiltonian representation, which yields a general bang--bang characterization. We next establish conditions for an endogenous transition from opacity to transparency and the failure of the positive-value disclosure principle. Comparative statics describe how signal informativeness, action costs, and risk aversion affect the disclosure cutoff. Section~\ref{subsec:closed_form_experimentation} then provides a closed-form delegated-experimentation application that derives the optimal two-phase experimentation and disclosure path, its endogenous cutoff, and the primitive implementability and feasibility conditions when information production consumes scarce testing capacity.

\subsection{Optimal Disclosure under General Disclosure Technologies}
\label{subsec:optimal_disclosure_general}

In the general formulation, disclosure affects continuation values through the controlled informational generator $\mathcal L^{\pi,a}$. When the information-production technology satisfies the linear--Gaussian specification, disclosure is represented by an intensity $\delta_t\in\mathcal D$, where $\mathcal D$ is the compact feasible set introduced in Subsection~\ref{subsec:informational_dynamics}, and the generator is written as $\mathcal L^{\delta,a}$. The unrestricted disclosure case is $\mathcal D=[0,1]$; restricted sets accommodate institutional constraints on feasible disclosure.

For each admissible $(t,\mu,x,a,\delta)$, let $\Gamma^I(t,\mu,x;a,\delta)$ denote the set of admissible current and continuation choices under disclosure intensity $\delta$ that implement action $a$ and satisfy the applicable participation requirement. For $\gamma\in\Gamma^I(t,\mu,x;a,\delta)$, let $T_\gamma$ denote its current transfer component. Define the \textbf{reduced implementable disclosure Hamiltonian} by
\begin{equation}
\label{eq:reduced_implementable_disclosure_hamiltonian}
\mathfrak H^I(t,\mu,x;a,\delta)
:=
\sup_{\gamma\in\Gamma^I(t,\mu,x;a,\delta)}
\left\{
v(T_\gamma,a,\mu,x)
+
\mathcal L^{\delta,a}J(t,\mu,x)
+
J_x(t,\mu,x)G(x,a)
\right\}.
\end{equation}

\begin{theorem}[General Characterization of Optimal Disclosure]
\label{theorem:general_disclosure}

Suppose Assumptions~\ref{ass:regularity_primitives}--\ref{ass:markov_state_representation} hold. Suppose further that, for every admissible $(t,\mu,x,a)$ and every $\delta\in\mathcal D$, the set $\Gamma^I(t,\mu,x;a,\delta)$ is nonempty, the supremum in \eqref{eq:reduced_implementable_disclosure_hamiltonian} is attained, and the mapping $\delta\mapsto\mathfrak H^I(t,\mu,x;a,\delta)$ is continuous on $\mathcal D$ and differentiable on $\operatorname{int}\mathcal D$.

Then, along any optimal implementable continuation path and for almost every $t$,
\begin{equation}
\label{eq:optimal_delta_general_revised}
\delta_t^*
\in
\arg\max_{\delta\in\mathcal D}
\mathfrak H^I(t,\mu_t,x_t;a_t,\delta).
\end{equation}

If $\delta_t^*\in\operatorname{int}\mathcal D$, then
\begin{equation}
\label{eq:delta_FOC_revised}
\left.
\frac{\partial}{\partial\delta}
\mathfrak H^I(t,\mu_t,x_t;a_t,\delta)
\right|_{\delta=\delta_t^*}
=
0.
\end{equation}

\end{theorem}

\begin{proof}
See Appendix~\ref{app:proof_general_disclosure}.
\end{proof}

The characterization given by Theorem~\ref{theorem:general_disclosure} allows optimal disclosure to be interior or at an endpoint, depending on the shape of the reduced implementable disclosure Hamiltonian. It also accommodates continuous-time filtering and disclosure settings in which disclosure continuously modifies the effective information flow; see \citet{LiptserShiryaev2001,BainCrisan2009,Sannikov2008}.

Define the generator contribution associated with a change from nondisclosure to full disclosure by
\[
\Delta^aJ(t,\mu,x)
:=
\mathcal L^{1,a}J(t,\mu,x)
-
\mathcal L^{0,a}J(t,\mu,x).
\]

The next two examples illustrate, respectively, bang--bang and interior disclosure under nonlinear reduced disclosure objectives.

\begin{example}[Bang--Bang Disclosure under a Non-Affine Reduced Objective]

Suppose
\begin{equation}
\label{eq:nonaffine_bangbang_example}
\mathfrak H^I(t,\mu,x;a,\delta)
=
\mathfrak H^I(t,\mu,x;a,0)
+
\delta^2\Delta^aJ(t,\mu,x),
\qquad
\delta\in[0,1].
\end{equation}
The disclosure-dependent term is $\delta^2\Delta^aJ$. Hence the unique optimum is $\delta=1$ when $\Delta^aJ>0$ and $\delta=0$ when $\Delta^aJ<0$; every $\delta\in[0,1]$ is optimal when $\Delta^aJ=0$. Thus bang--bang disclosure can arise even when the reduced disclosure objective is non-affine.

\end{example}

\begin{example}[Interior Disclosure under a Nonlinear Disclosure Objective]

Suppose disclosure entails a convex processing, validation, or communication cost $\zeta\delta^2/2$, where $\zeta>0$, so that
\begin{equation}
\label{eq:quadratic_generator_revised}
\mathfrak H^I(t,\mu,x;a,\delta)
=
\mathfrak H^I(t,\mu,x;a,0)
+
\delta\Delta^aJ(t,\mu,x)
-
\frac{\zeta}{2}\delta^2.
\end{equation}
The objective is strictly concave, and the constrained optimum is
\begin{equation}
\label{eq:delta_star_revised}
\delta_t^*
=
\min
\left\{
\max
\left\{
\frac{\Delta^{a_t}J(t,\mu_t,x_t)}{\zeta},
0
\right\},
1
\right\}.
\end{equation}
Disclosure is interior whenever $0<\Delta^{a_t}J(t,\mu_t,x_t)<\zeta$, and $\zeta$ governs its responsiveness to the informational contribution.

\end{example}

We next impose an affine dependence of the informational generator on disclosure intensity. Consider the unrestricted benchmark $\mathcal D=[0,1]$ and suppose
\begin{equation}
\label{eq:affine_disclosure_generator}
\mathcal L^{\delta,a}
=
(1-\delta)\mathcal L^{0,a}
+
\delta\mathcal L^{1,a},
\qquad
\delta\in[0,1].
\end{equation}

Because disclosure may change the continuation choices required to implement a prescribed action, affinity of the disclosure generator does not by itself imply affinity of the reduced implementable problem. Accordingly, suppose also that
\begin{equation}
\label{eq:reduced_affine_disclosure_condition}
\mathfrak H^I(t,\mu,x;a,\delta)
=
(1-\delta)\mathfrak H^I(t,\mu,x;a,0)
+
\delta\mathfrak H^I(t,\mu,x;a,1),
\qquad
\delta\in[0,1].
\end{equation}

Define
\[
\Delta^{I,a}\mathfrak H(t,\mu,x)
:=
\mathfrak H^I(t,\mu,x;a,1)
-
\mathfrak H^I(t,\mu,x;a,0).
\]
Equivalently, the affine reduced-Hamiltonian condition can be written as
\begin{equation}
\label{eq:affine_representation}
\mathfrak H^I(t,\mu,x;a,\delta)
=
\mathfrak H^I(t,\mu,x;a,0)
+
\delta\Delta^{I,a}\mathfrak H(t,\mu,x).
\end{equation}

The affine representation yields the following bang--bang characterization.

\begin{proposition}[Bang--Bang Disclosure under an Affine Reduced Hamiltonian]
\label{thm:bang_bang_general}

Suppose Assumptions~\ref{ass:regularity_primitives}--\ref{ass:markov_state_representation} hold, the feasible disclosure set is $\mathcal D=[0,1]$, and the reduced implementable disclosure Hamiltonian satisfies the affine representation \eqref{eq:affine_representation}. Then, along any optimal implementable continuation path and for almost every $t$,
\begin{equation}
\label{eq:bangbang_rule_revised}
\delta_t^*
=
\begin{cases}
1, & \text{if }\Delta^{I,a_t}\mathfrak H(t,\mu_t,x_t)>0,\\[0.4em]
0, & \text{if }\Delta^{I,a_t}\mathfrak H(t,\mu_t,x_t)<0,\\[0.4em]
\text{any value in }[0,1], & \text{if }\Delta^{I,a_t}\mathfrak H(t,\mu_t,x_t)=0.
\end{cases}
\end{equation}

Hence optimal disclosure is bang--bang.

\end{proposition}

\begin{proof}
See Appendix~\ref{app:proof_bang_bang_general}.
\end{proof}

Proposition~\ref{thm:bang_bang_general} establishes a general bang--bang characterization under the affine reduced-Hamiltonian representation; the result does not rely on an exogenous binary restriction on disclosure. Affinity of the disclosure generator alone is not sufficient because disclosure may alter the continuation choices required to preserve implementation. The switching term incorporates the informational, incentive, and feasibility consequences of disclosure. The threshold structure is related to dynamic disclosure and information-provision models such as \citet{ElySzydlowski2020,HornerSkrzypacz2017}, but here information is produced through incentive-sensitive action and linked to future feasibility.

When the information-production technology satisfies the linear--Gaussian specification, posterior precision evolves according to $\dot\rho_t=\delta_ta_t^2/\sigma^2$. If the compensated implementability adjustments satisfy the envelope identity, the resulting switching term satisfies
\begin{equation}
\label{eq:Delta_gaussian_revised}
\Delta^{I,a_t}\mathfrak H(t,\rho_t,x_t)
=
\Delta^{a_t}J(t,\rho_t,x_t)
=
\frac{a_t^2}{\sigma^2}
J_\rho(t,\rho_t,x_t).
\end{equation}

\begin{corollary}[Bang--Bang Disclosure under Gaussian Learning]
\label{cor:bang_bang_gaussian_revised}

Suppose Assumptions~\ref{ass:regularity_primitives}--\ref{ass:markov_state_representation} hold. Suppose further that:
\begin{enumerate}
\item[(i)] the feasible disclosure set is $\mathcal D=[0,1]$, and the reduced implementable disclosure Hamiltonian satisfies the affine representation \eqref{eq:affine_representation};

\item[(ii)] the information-production technology satisfies the linear--Gaussian specification $dS_t=a_t\theta\,dt+\sigma\,dW_t$, with posterior precision evolving according to $\dot\rho_t=\delta_ta_t^2/\sigma^2$;

\item[(iii)] $a_t>0$ for almost every $t$ along the relevant optimal implementable continuation path;

\item[(iv)] the compensated-envelope identity \eqref{eq:Delta_gaussian_revised} holds along that path.
\end{enumerate}

Then, for almost every $t$,
\begin{equation}
\label{eq:bang_bang_gaussian_revised}
\delta_t^*
=
\begin{cases}
1, & \text{if }J_\rho(t,\rho_t,x_t)>0,\\[0.4em]
0, & \text{if }J_\rho(t,\rho_t,x_t)<0,\\[0.4em]
\text{any value in }[0,1], & \text{if }J_\rho(t,\rho_t,x_t)=0.
\end{cases}
\end{equation}

\end{corollary}

\begin{proof}
See Appendix~\ref{app:proof_bang_bang_gaussian}.
\end{proof}

Under the compensated-envelope identity, Gaussian disclosure is governed by the sign of $J_\rho$. By the decomposition in Subsection~\ref{subsec:marginal_value_information}, this sign reflects the balance between the decision benefit of precision and its endogenous implementation and feasibility costs.

\subsection{Endogenous Two-Phase Disclosure and the Positive-Value Disclosure Principle}
\label{subsec:two_phase_disclosure}

Under unrestricted disclosure, the affine reduced-Hamiltonian representation, the linear--Gaussian information-production technology, and the compensated-envelope identity \eqref{eq:Delta_gaussian_revised}, optimal disclosure is governed by
\[
J_\rho(t,\rho_t,x_t)
=
B_\rho(t,\rho_t,x_t)
-
K_\rho(t,\rho_t,x_t).
\]
A single crossing of this net marginal value generates an endogenous transition from opacity to transparency. Positive direct informational value alone need not imply disclosure because $B_\rho$ excludes the implementation costs generated by incentives, continuation promises, and feasibility.

\begin{theorem}[Endogenous Two-Phase Disclosure]
\label{thm:two_phase_disclosure}

Suppose Assumptions~\ref{ass:regularity_primitives}--\ref{ass:quasilinear} hold. Suppose further that the following conditions are satisfied:
\begin{enumerate}
\item[(i)] the feasible disclosure set is $\mathcal D=[0,1]$, and the reduced implementable disclosure Hamiltonian satisfies the affine representation \eqref{eq:affine_representation};

\item[(ii)] the information-production technology satisfies the linear--Gaussian specification, and the compensated-envelope identity \eqref{eq:Delta_gaussian_revised} holds along the relevant optimal implementable continuation path;

\item[(iii)] the information-producing action satisfies $a_t>0$ for almost every $t$ along the relevant optimal implementable continuation path;

\item[(iv)] the composite net marginal value of posterior precision,
\[
J_\rho(t,\rho_t,x_t)
=
B_\rho(t,\rho_t,x_t)
-
K_\rho(t,\rho_t,x_t),
\]
is continuous on $[0,\bar T]$ and continuously differentiable on $(0,\bar T)$;

\item[(v)] the net marginal value satisfies\footnote{This endogenous single-crossing condition is stated at the equilibrium level to preserve generality. Primitive sufficient conditions can be derived under additional structure, such as linear payoff functions; Section~7 provides a related primitive-based illustration in the closed-form application.}
\[
J_\rho(0,\rho_0,x_0)<0,
\qquad
J_\rho(\bar T,\rho_{\bar T},x_{\bar T})>0,
\]
and
\[
\frac{d}{dt}
J_\rho(t,\rho_t,x_t)
>
0
\qquad
\text{for all }t\in(0,\bar T).
\]
\end{enumerate}

Then there exists a unique cutoff date $t^*\in(0,\bar T)$ satisfying
\[
J_\rho(t^*,\rho_{t^*},x_{t^*})
=
0.
\]
Moreover, an optimal disclosure process can be selected to satisfy
\[
\delta_t^*
=
\begin{cases}
0, & t<t^*,\\[0.4em]
\text{any value in }[0,1], & t=t^*,\\[0.4em]
1, & t>t^*.
\end{cases}
\]

\end{theorem}

\begin{proof}
See Appendix~\ref{app:proof_two_phase}.
\end{proof}

At the cutoff, $B_\rho=K_\rho$. Implementation costs exceed informational benefits before the cutoff and informational benefits exceed implementation costs afterward. Feasibility affects this comparison through the continuation state, the shadow value of future opportunities, and the implementable action path.

We next ask whether positive direct informational value alone is sufficient to induce disclosure.

\begin{definition}[Positive-Value Disclosure Principle]
\label{def:positive_value_disclosure}

Under the unrestricted disclosure benchmark $\mathcal D=[0,1]$, the \textbf{positive-value disclosure principle} holds if, along every optimal implementable continuation path and for almost every date at which $a_t>0$,
\[
B_\rho(t,\rho_t,x_t)>0
\quad\Longrightarrow\quad
\delta_t^*>0.
\]

\end{definition}

Under the affine reduced-Hamiltonian representation and the linear--Gaussian information-production technology, $\delta_t^*>0$ is equivalent to full disclosure at every nonindifference date. At an indifference date, every disclosure intensity is optimal. The failure result below occurs on an interval over which the net marginal value is strictly negative.

\begin{corollary}[Failure of the Positive-Value Disclosure Principle]
\label{cor:positive_value_disclosure_failure}

Suppose the conditions of Theorem~\ref{thm:two_phase_disclosure} hold and
\begin{equation}
\label{eq:positive_direct_information_value}
B_\rho(t,\rho_t,x_t)>0
\qquad
\text{for all }t\in[0,\bar T].
\end{equation}
Then the positive-value disclosure principle fails. In particular,
\[
0
<
B_\rho(t,\rho_t,x_t)
<
K_\rho(t,\rho_t,x_t)
\qquad
\text{for every }t<t^*,
\]
while
\[
\delta_t^*=0
\qquad
\text{for almost every }t<t^*.
\]
Thus information has strictly positive direct value throughout the initial opacity phase but is nevertheless optimally withheld.

\end{corollary}

\begin{proof}
See Appendix~\ref{app:proof_positive_value_disclosure_failure}.
\end{proof}

Theorem~\ref{thm:two_phase_disclosure} characterizes the timing of the transition, while Corollary~\ref{cor:positive_value_disclosure_failure} shows that positive direct informational value need not induce disclosure when $K_\rho>B_\rho$. Conversely, disclosure is favored when $B_\rho>K_\rho$, in particular when $K_\rho=0$.

This result is related to \citet{HornerSamuelson2026}, who show that withholding information can improve incentives despite an allocative cost. The mechanism here is different: information is endogenously produced through action and jointly governed by disclosure, incentives, and feasibility, so information with positive direct decision value may nevertheless be optimally withheld because of its implementation cost.

\subsection{Comparative Statics of Disclosure Timing}
\label{subsec:disclosure_comparative_statics}

Under the affine representation of the reduced implementable disclosure Hamiltonian, the linear--Gaussian information-production technology, and the compensated-envelope identity \eqref{eq:Delta_gaussian_revised}, disclosure timing is determined by the cutoff at which the net marginal value of posterior precision crosses zero. The results below derive the general cutoff formula and apply it to signal informativeness, action costs, and risk aversion.

\subsubsection{Cutoff-Based Comparative Statics under Gaussian Learning}

Let $\lambda$ denote an economic parameter and define the composite marginal value along the endogenous optimal state path by $M(t,\lambda):=J_\rho\bigl(t,\rho_t(\lambda),x_t(\lambda);\lambda\bigr)$. The cutoff satisfies
\begin{equation}
\label{eq:cutoff_condition_general}
M(t^*,\lambda)=0.
\end{equation}

\begin{proposition}[Comparative Statics of the Disclosure Cutoff]
\label{prop:comparative_statics_revised}

Suppose Assumptions~\ref{ass:regularity_primitives}--\ref{ass:quasilinear} hold. Suppose further that:
\begin{enumerate}
\item[(i)] the feasible disclosure set is $\mathcal D=[0,1]$, the reduced implementable disclosure Hamiltonian satisfies the affine representation \eqref{eq:affine_representation}, the information-production technology satisfies the linear--Gaussian specification, and the compensated-envelope identity \eqref{eq:Delta_gaussian_revised} holds along the relevant optimal implementable continuation path;

\item[(ii)] the disclosure cutoff date $t^*$ is interior;

\item[(iii)] for an economic parameter $\lambda$, the composite marginal value $M(t,\lambda)$ is continuously differentiable in $(t,\lambda)$ in a neighborhood of $(t^*,\lambda)$ and satisfies $M_t(t^*,\lambda)>0$.
\end{enumerate}

Then the cutoff date is locally differentiable in $\lambda$, and
\begin{equation}
\label{eq:dtstar_general_revised}
\frac{dt^*}{d\lambda}
=
-
\frac{M_\lambda(t^*,\lambda)}
{M_t(t^*,\lambda)}
<0
\quad\Longleftrightarrow\quad
M_\lambda(t^*,\lambda)>0.
\end{equation}

\end{proposition}

\begin{proof}
See Appendix~\ref{app:proof_comparative_statics}.
\end{proof}

Proposition~\ref{prop:comparative_statics_revised} shows that a parameter advances disclosure precisely when it raises the composite marginal value of precision at the cutoff. Using $J_\rho=B_\rho-K_\rho$, we have $M_\lambda=dB_\rho/d\lambda-dK_\rho/d\lambda$ along the endogenous optimal state path. Thus disclosure timing depends on the parameter's total effects on informational benefits and implementation costs, including the induced changes in posterior beliefs, feasibility, and the implementable action path.

\subsubsection{Learning Technologies and Action Costs}

We first apply the general cutoff formula to signal informativeness and then to action costs.

\begin{proposition}[Signal Informativeness and Disclosure Timing]
\label{prop:signal_informativeness}

Suppose the conditions of Proposition~\ref{prop:comparative_statics_revised} hold with $\lambda=\chi:=\sigma^{-2}$. Then the cutoff date is locally differentiable in $\chi$, and
\begin{equation}
\label{eq:chi_cutoff_sign_revised}
\frac{dt^*}{d\chi}
=
-
\frac{M_\chi(t^*,\chi)}
{M_t(t^*,\chi)}
<0
\quad\Longleftrightarrow\quad
M_\chi(t^*,\chi)>0.
\end{equation}
Consequently,
\begin{equation}
\label{eq:chi_benefit_cost_condition_revised}
\frac{dt^*}{d\chi}<0
\quad\Longleftrightarrow\quad
\frac{d}{d\chi}
B_\rho\bigl(t^*,\rho_{t^*}(\chi),x_{t^*}(\chi);\chi\bigr)
>
\frac{d}{d\chi}
K_\rho\bigl(t^*,\rho_{t^*}(\chi),x_{t^*}(\chi);\chi\bigr).
\end{equation}

\end{proposition}

\begin{proof}
See Appendix~\ref{app:proof_signal_informativeness}.
\end{proof}

Greater signal informativeness does not mechanically imply earlier disclosure. It advances disclosure only when its total effect on informational benefits exceeds its total effect on implementation costs at the cutoff. This local timing result is distinct from the informational-monotonicity result in Subsection~\ref{subsec:informational_monotonicity}, which compares optimized values across information-production technologies.

\begin{proposition}[Cost Parameters and Disclosure Timing]
\label{prop:effort_cost_curvature_revised}

Let $\kappa$ parameterize the action-cost function $c(a;\kappa)$, with $c_{a\kappa}>0$. Suppose the conditions of Proposition~\ref{prop:comparative_statics_revised} hold with $\lambda=\kappa$. Suppose further that, at the cutoff, the total effect of $\kappa$ on the direct informational-benefit component is weakly negative and its total effect on the implementation-cost component is weakly positive, with at least one inequality strict; that is, $dB_\rho/d\kappa\leq0$ and $dK_\rho/d\kappa\geq0$ along the endogenous optimal state path at $t^*$.

Then $M_\kappa(t^*,\kappa)<0$, and therefore
\begin{equation}
\label{eq:dtstar_kappa_revised}
\frac{dt^*}{d\kappa}
=
-
\frac{M_\kappa(t^*,\kappa)}
{M_t(t^*,\kappa)}
>
0.
\end{equation}
Hence a higher action-cost parameter delays disclosure under the stated component-ordering conditions.

\end{proposition}

\begin{proof}
See Appendix~\ref{app:proof_effort_cost_curvature}.
\end{proof}

Proposition~\ref{prop:effort_cost_curvature_revised} is a sufficient-condition result. The condition $c_{a\kappa}>0$ means that $c(a;\kappa)$ has increasing differences in $(a,\kappa)$, or equivalently that $-c(a;\kappa)$ has decreasing differences. This primitive condition alone does not determine disclosure timing because actions, transfers, posterior dynamics, and feasibility adjust endogenously. The additional component-ordering condition determines the resulting effect on the net marginal value of precision.

For example, if $c(a;\kappa)=\kappa a^2/2$, then $c_{a\kappa}=a>0$ for positive actions and $c_{aa\kappa}=1>0$. Under the stated component ordering, a higher $\kappa$ lowers the net marginal value of precision at the cutoff and delays disclosure.

\subsubsection{Risk Aversion and Disclosure Timing}

Risk aversion affects disclosure through both the direct informational benefit and the adjustments required to preserve implementability. Suppose the agent's instantaneous utility is $u(T_t,\mu_t,x_t)-c(a_t)$, where $u_T>0$ and $u_{TT}\leq0$. Under quasi-linear utility, $u(T,\mu,x)=T+\tilde u(\mu,x)$ and $u_{TT}=0$, whereas risk aversion with respect to transfers corresponds to $u_{TT}<0$. Under strict concavity, posterior fluctuations generate continuation-utility risk and interact with intertemporal insurance; see \citet{SpearSrivastava1987,ThomasWorrall1990,PhelanTownsend1991,Sannikov2008}.

Because quasi-linearity is not imposed here, we use the general marginal-value decomposition from Subsection~\ref{subsec:marginal_value_information}. Let $(\rho_t^{RA},x_t^{RA})$ and $(\rho_t^{RN},x_t^{RN})$ denote the respective optimal implementable continuation paths under risk aversion and risk neutrality, and let $J^{RA}$ and $J^{RN}$ denote the corresponding principal continuation values. Let $\mathcal B_\rho^{RA}$ and $\mathcal B_\rho^{RN}$ denote the corresponding direct-benefit components, and let $\mathcal C_\rho^{RA}$ and $\mathcal C_\rho^{RN}$ denote the corresponding implementability-adjustment components.

\begin{proposition}[Risk Aversion and Disclosure Timing]
\label{prop:risk_aversion}

Suppose Assumptions~\ref{ass:regularity_primitives}--\ref{ass:markov_state_representation} hold. Suppose further that:
\begin{enumerate}
\item[(i)] in both regimes, the feasible disclosure set is $\mathcal D=[0,1]$, the reduced implementable disclosure Hamiltonian satisfies the affine representation \eqref{eq:affine_representation}, the common information-production technology satisfies the linear--Gaussian specification, and the compensated-envelope identity \eqref{eq:Delta_gaussian_revised} holds along the relevant optimal implementable continuation path;

\item[(ii)] the agent's utility under risk aversion is strictly concave in transfers, $u_{TT}(T,\mu,x)<0$;

\item[(iii)] along the respective optimal implementable continuation paths,
\[
\mathcal B_\rho^{RA}(t,\rho_t^{RA},x_t^{RA})
\leq
\mathcal B_\rho^{RN}(t,\rho_t^{RN},x_t^{RN}),
\qquad
\mathcal C_\rho^{RA}(t,\rho_t^{RA},x_t^{RA})
\geq
\mathcal C_\rho^{RN}(t,\rho_t^{RN},x_t^{RN})
\]
for all $t\in[0,\bar T]$.
\end{enumerate}

Then
\[
J_\rho^{RA}(t,\rho_t^{RA},x_t^{RA})
\leq
J_\rho^{RN}(t,\rho_t^{RN},x_t^{RN})
\qquad
\text{for all }t\in[0,\bar T].
\]

If both composite marginal-value paths are continuous, strictly increasing, and admit unique interior disclosure cutoffs satisfying
\[
J_\rho^{RA}(t_{RA}^*,\rho_{t_{RA}^*}^{RA},x_{t_{RA}^*}^{RA})=0,
\qquad
J_\rho^{RN}(t_{RN}^*,\rho_{t_{RN}^*}^{RN},x_{t_{RN}^*}^{RN})=0,
\]
then $t_{RA}^*\geq t_{RN}^*$. Hence risk aversion weakly delays disclosure under the stated component-ordering conditions.

\end{proposition}

\begin{proof}
See Appendix~\ref{app:proof_risk_aversion}.
\end{proof}

Under the stated ordering conditions, risk aversion lowers the net marginal value of precision, weakly delays disclosure, and prolongs the initial withholding phase.

\section{A Closed-Form Illustration: Delegated Experimentation with Scarce Testing Capacity}
\label{subsec:closed_form_experimentation}

This section serves both as a closed-form application of the general framework and as a transparent illustration of its economic mechanism. Using the delegated-experimentation environment introduced in Section~3, a principal delegates a pilot project to an agent who chooses an experimentation intensity. Experimentation produces evidence about an unknown technology but is costly and consumes scarce testing capacity, while the principal governs how the evidence is validated, disclosed, and used in posterior beliefs and compensation. The closed-form solution makes transparent how information production, incentives, disclosure, and feasibility jointly determine the optimal path of experimentation and informational use.

\begin{example}[Delegated Experimentation and Endogenous Disclosure]
\label{ex:closed_form_gaussian_revised}

Let $\theta$ denote the informational-environment state, interpreted as an unknown characteristic of the technology, such as its reliability, productivity, or suitability for large-scale adoption. The state $\theta$ does not enter current payoffs directly, but governs the distribution of experimental signals.

At each date $t\in[0,\bar T]$, the agent chooses an experimentation intensity $a_t\in[0,\bar a]$. The raw experimental evidence is represented by the signal process
\[
dS_t
=
a_t\theta\,dt
+
\sigma\,dW_t,
\qquad
\chi:=\sigma^{-2}.
\]

This raw evidence does not automatically enter the principal's formal information system. Let $\delta_t\in[0,1]$ denote the intensity with which the evidence is validated, disclosed, and incorporated into public posterior beliefs. Under the linear--Gaussian information-production technology,
$\dot\rho_t=\delta_t\chi a_t^2$.
Thus posterior precision is jointly determined by the action that produces evidence and the disclosure rule that makes it publicly usable.

Experimentation consumes scarce testing capacity:
\[
\dot x_t
=
-a_t,
\qquad
x_0\in[0,\bar x],
\]
and the induced path must satisfy $x_t\geq0$.

The agent is risk neutral and has flow utility
\[
T_t+\eta x_t-\frac{\kappa}{2}a_t^2
\]
and terminal payoff
\[
T_{\bar T}+\eta_Tx_{\bar T},
\]
where $\kappa>0$ and $\eta,\eta_T\geq0$.

The agent is compensated for producing verifiable information. Transfers are affine in public posterior precision:
\[
T_t=w+q\rho_t,
\qquad
T_{\bar T}=q_T\rho_{\bar T},
\]
where $q,q_T>0$. The constant payment $w$ is chosen to satisfy the agent's participation constraint.

The principal's flow utility is $b\rho_t-T_t$, and the principal's terminal payoff is
\[
\beta\rho_{\bar T}
+
\varphi x_{\bar T}
-
T_{\bar T},
\]
where $b,\beta,\varphi>0$.

Using the continuation-value representation in Section~4, for any fixed continuation action--disclosure path the agent's continuation value is affine in the current informational and feasibility states:
\[
U(t,\rho,x)
=
\big[q_T+q(\bar T-t)\big]\rho
+
\big[\eta_T+\eta(\bar T-t)\big]x
+
\widetilde U(t),
\]
where $\widetilde U(t)$ collects terms independent of the current values of $\rho$ and $x$. Hence
\[
U_\rho(t)
=
q_T+q(\bar T-t),
\qquad
U_x(t)
=
\eta_T+\eta(\bar T-t).
\]

The action-relevant continuation-value Hamiltonian is
\[
\mathcal H_t^A(a)
=
-\frac{\kappa}{2}a^2
+
\delta_t\chi a^2U_\rho(t)
-
aU_x(t).
\]

When $\delta_t=0$, $\mathcal H_t^A(a)=-\kappa a^2/2-aU_x(t)$, so $a_t=0$ is optimal. Under full disclosure, maximal experimentation is implementable whenever
\begin{equation}
\label{eq:high_action_implementability}
\chi\bigl[q_T+q(\bar T-t)\bigr]
\geq
\frac{\kappa}{2}
+
\frac{\eta_T+\eta(\bar T-t)}{\bar a}.
\end{equation}
This condition implies that $\mathcal H_t^A$ is convex and that $\bar a$ weakly dominates zero.

Consider an inactive initial phase followed by full disclosure and maximal experimentation:
\[
a_t^*
=
\begin{cases}
0, & t<t^*,\\[0.3em]
\bar a, & t>t^*.
\end{cases}
\]
The induced state paths are
\[
\rho_t=\rho_0,
\qquad
x_t=x_0,
\qquad
t<t^*,
\]
and
\[
\rho_t
=
\rho_0+\chi\bar a^2(t-t^*),
\qquad
x_t
=
x_0-\bar a(t-t^*),
\qquad
t>t^*.
\]

Let $\bar U$ denote the agent's reservation payoff. Binding participation requires
\[
\bar U
=
\int_0^{\bar T}
\left(
w+q\rho_t+\eta x_t-\frac{\kappa}{2}(a_t^*)^2
\right)dt
+
q_T\rho_{\bar T}
+
\eta_Tx_{\bar T},
\]
so
\begin{equation}
\label{eq:fixed_transfer_binding_ir}
w
=
\frac{1}{\bar T}
\left[
\bar U
-
q\int_0^{\bar T}\rho_t\,dt
-
\eta\int_0^{\bar T}x_t\,dt
-
q_T\rho_{\bar T}
-
\eta_Tx_{\bar T}
+
\frac{\kappa}{2}\bar a^2(\bar T-t^*)
\right].
\end{equation}

Substituting into the principal's objective gives
\begin{align}
\Pi
={}&
b\int_0^{\bar T}\rho_t\,dt
+
\beta\rho_{\bar T}
+
\varphi x_{\bar T}
-
\bar U
\notag\\
&+
\eta\int_0^{\bar T}x_t\,dt
+
\eta_Tx_{\bar T}
-
\frac{\kappa}{2}\bar a^2(\bar T-t^*).
\label{eq:principal_reduced_payoff}
\end{align}
Because transfers are unrestricted, the payments generated by $q$ and $q_T$ are offset through $w$. They remain essential for implementation through \eqref{eq:high_action_implementability}, but do not enter the reduced objective conditional on the implemented path.

Using the continuation-value formulation in Section~5, in the interior region where the feasibility constraint is slack, the principal's HJB becomes
\begin{equation}
\label{eq:closed_form_principal_hjb}
-\partial_tJ(t,\rho,x)
=
b\rho+\eta x
+
\sup_{\substack{a\in[0,\bar a],\,\delta\in[0,1]\\
a\ \text{implementable under }\delta}}
\left\{
\delta\chi a^2J_\rho(t,\rho,x)
-
aJ_x(t,\rho,x)
-
\frac{\kappa}{2}a^2
\right\},
\end{equation}
with terminal condition
\[
J(\bar T,\rho,x)
=
\beta\rho+(\varphi+\eta_T)x.
\]
Because the reduced problem is affine in $\rho$ and $x$, write
\[
J(t,\rho,x)
=
A(t)+B(t)\rho+C(t)x.
\]
Matching coefficients in \eqref{eq:closed_form_principal_hjb} and using the terminal condition gives
\[
J_\rho(t,\rho,x)
=
\beta+b(\bar T-t),
\qquad
J_x(t,\rho,x)
=
\varphi+\eta_T+\eta(\bar T-t).
\]

Comparing inactivity with full-disclosure maximal experimentation, the corresponding Hamiltonian gain from activating experimentation at date $t$ is
\begin{align}
\Gamma(t)
={}&
\chi\bar a^2J_\rho(t,\rho,x)
-
\bar aJ_x(t,\rho,x)
-
\frac{\kappa}{2}\bar a^2
\notag\\
={}&
\beta\chi\bar a^2
-
(\varphi+\eta_T)\bar a
-
\frac{\kappa}{2}\bar a^2
-
\bigl(
\eta\bar a-b\chi\bar a^2
\bigr)(\bar T-t).
\label{eq:marginal_activation_value}
\end{align}
Thus, whenever maximal experimentation is implementable, $\Gamma(t)$ gives the net gain from activating full experimentation and disclosure rather than remaining inactive.

Let $\tau:=\bar T-t^*$ denote the active duration. Feasibility requires
\[
0
\leq
\tau
\leq
\bar\tau,
\qquad
\bar\tau
:=
\min
\left\{
\bar T,\frac{x_0}{\bar a}
\right\}.
\]
Along the two-phase path,
\[
\rho_{\bar T}
=
\rho_0+\chi\bar a^2\tau,
\qquad
x_{\bar T}
=
x_0-\bar a\tau,
\]
and
\[
\int_0^{\bar T}\rho_t\,dt
=
\rho_0\bar T
+
\frac{\chi\bar a^2}{2}\tau^2,
\qquad
\int_0^{\bar T}x_t\,dt
=
x_0\bar T
-
\frac{\bar a}{2}\tau^2.
\]
Therefore,
\begin{equation}
\label{eq:principal_tau_objective}
\Pi(\tau)
=
C
+
\frac{1}{2}
\bigl(
b\chi\bar a^2-\eta\bar a
\bigr)\tau^2
+
\left[
\beta\chi\bar a^2
-
(\varphi+\eta_T)\bar a
-
\frac{\kappa}{2}\bar a^2
\right]\tau,
\end{equation}
where $C$ is independent of $\tau$. Moreover,
\[
\Pi'(\tau)
=
\Gamma(\bar T-\tau),
\]
so the reduced-duration problem and the Hamiltonian switching condition give the same cutoff.

Suppose
\begin{equation}
\label{eq:cutoff_concavity_condition}
\eta\bar a
>
b\chi\bar a^2,
\qquad\text{equivalently,}\qquad
\eta
>
b\chi\bar a.
\end{equation}
Then $\Pi(\tau)$ is strictly concave, and the interior active duration is
\[
\tau^*
=
\frac{
\beta\chi\bar a^2
-
(\varphi+\eta_T)\bar a
-
\frac{\kappa}{2}\bar a^2
}{
\eta\bar a
-
b\chi\bar a^2
}.
\]
Hence
\begin{equation}
\label{eq:closed_form_cutoff_revised}
t^*
=
\bar T
-
\frac{
\beta\chi\bar a^2
-
(\varphi+\eta_T)\bar a
-
\frac{\kappa}{2}\bar a^2
}{
\eta\bar a
-
b\chi\bar a^2
}.
\end{equation}
Equivalently,
\[
t^*
=
\bar T
-
\frac{
\beta\chi\bar a
-
(\varphi+\eta_T)
-
\frac{\kappa}{2}\bar a
}{
\eta
-
b\chi\bar a
}.
\]

The solution is interior and feasible when
\begin{equation}
\label{eq:interior_cutoff_conditions}
0
<
\beta\chi\bar a^2
-
(\varphi+\eta_T)\bar a
-
\frac{\kappa}{2}\bar a^2
<
\bigl(
\eta\bar a
-
b\chi\bar a^2
\bigr)\bar\tau.
\end{equation}
Otherwise, the optimum occurs at $\tau=0$ or $\tau=\bar\tau$.

Dividing the switching value by the precision flow $\chi\bar a^2$ gives the net value per unit of precision generated through experimentation:
\begin{equation}
\label{eq:closed_form_information_index}
\frac{\Gamma(t)}{\chi\bar a^2}
=
\beta
+
b(\bar T-t)
-
\frac{\kappa}{2\chi}
-
\frac{
\varphi+\eta_T+\eta(\bar T-t)
}{
\chi\bar a
}.
\end{equation}
The first two terms, $J_\rho(t)=\beta+b(\bar T-t)$, give the direct value of additional precision, while the remaining terms capture the experimentation and feasibility costs required to generate it.

Under \eqref{eq:cutoff_concavity_condition},
\[
\frac{d}{dt}
\left[
\frac{\Gamma(t)}{\chi\bar a^2}
\right]
=
\frac{\eta}{\chi\bar a}
-
b
>
0.
\]
Thus $\Gamma(t)$ is strictly increasing and, under \eqref{eq:interior_cutoff_conditions}, crosses zero uniquely.

The cutoff policy is also globally optimal within the full implementable set. For any disclosure intensity, $\mathcal H_t^A(a)$ is either decreasing on $[0,\bar a]$ or convex in $a$, so the agent's implemented action is always $0$ or $\bar a$. Conditional on activating experimentation at $\bar a$, full disclosure raises posterior precision without increasing the real experimentation or feasibility costs and also strengthens implementation. Full disclosure is therefore optimal at active dates. Since $\Gamma(t)$ is increasing, those dates form a terminal interval.

Provided that \eqref{eq:high_action_implementability} holds throughout $[t^*,\bar T]$, an optimal disclosure process can be selected as
\[
\delta_t^*
=
\begin{cases}
0, & t<t^*,\\[0.3em]
\text{any value in }[0,1], & t=t^*,\\[0.3em]
1, & t>t^*.
\end{cases}
\]
The corresponding action is zero before the cutoff and $\bar a$ afterward.

Because $a_t^*=0$ before the cutoff, disclosure does not affect posterior precision during the inactive phase. Thus $\delta_t^*=0$ is an optimal selection, but it need not be the unique pointwise disclosure choice before $t^*$.

The initial inactive phase arises despite $b,\beta>0$ because verified information requires costly experimentation and depletes valuable testing capacity. As the terminal adoption decision approaches, the remaining cost of capacity depletion declines and the net value of activation rises.

If action does not affect information production, disclosure no longer changes experimentation incentives. If evidence is automatically incorporated into public beliefs, disclosure ceases to be an independent instrument. If feasibility has no value, so that $\eta=\eta_T=\varphi=0$, condition~\eqref{eq:cutoff_concavity_condition} cannot hold when $b>0$, and the interior cutoff disappears.

\end{example}

The application makes the solution procedure explicit: continuation values determine implementable experimentation, binding participation yields the reduced payoff, feasibility restricts the active duration, and the single crossing of the marginal activation value determines the cutoff.

\section{Extensions and Robustness of Endogenous Informational Evolution}

Several of the sharper results in Sections~5 and~6 impose additional structure when needed, including quasi-linearity, affine disclosure, or a linear--Gaussian information-production technology. This section returns to the general recursive framework and considers two extensions: private payoff-relevant information and limited commitment. The first adds truthful-reporting constraints and private-information-indexed continuation values; the second replaces ex ante participation with continuation participation constraints. In both cases, the central recursive structure is preserved: actions produce information and affect feasibility, disclosure governs how generated information enters public beliefs, and continuation values determine implementability. Unless otherwise stated, the extensions maintain Assumptions~\ref{ass:regularity_primitives}--\ref{ass:markov_state_representation}, with additional conditions imposed only when required.

\subsection{Private Payoff-Relevant Types and Informational-Environment Uncertainty}
\label{subsec:payoff_relevant_information_states}

The benchmark model isolates endogenous information production by assuming that the informational-environment state $\theta$ governs the information-production technology but does not enter primitive payoffs directly. We now introduce a payoff-relevant private component, so that reporting and action incentive constraints interact within the same recursive informational system.

Let the primitive state be $z=(\omega,\theta)\in\Omega\times\Theta$, where $\omega$ is the agent's payoff-relevant private type and $\theta$ is the latent informational-environment state governing information production. The common prior $F_0\in\Delta(\Omega\times\Theta)$ is known to both the principal and the agent. The principal does not observe the realization of $z$, while the agent privately observes a signal $\xi\in\Xi$ with conditional distribution $I(\cdot\mid\omega)$ and likelihood $i(\xi\mid\omega)$. Independence is obtained as the special case $F_0=F_0^\Omega\times F_0^\Theta$. Direct observation of the type is obtained when $\Xi=\Omega$ and $I(B\mid\omega)=\mathbf 1_{\{\omega\in B\}}$ for every measurable $B\subseteq\Omega$.

The timing extends the benchmark in the standard way. The agent first observes the private signal $\xi$. The principal then commits to the report-contingent transfer and disclosure rules, and the agent decides whether to participate. Conditional on participation, the agent reports and subsequently chooses the action. The report and the realized action-generated information then determine posterior updating and feasibility evolution under the committed rules.

As in the benchmark model, actions produce raw information. Conditional on the primitive state, the information-production technology is $Q(\cdot\mid a_t,x_t,\theta)$, or more generally $Q(\cdot\mid a_t,x_t,\omega,\theta)$. The first specification separates payoff-relevant uncertainty from informational-environment uncertainty; the second allows the payoff-relevant type also to affect information production.

The mechanism may request a report $\hat\xi$ of the agent's private signal. The disclosed signal is $Y_t=\pi_t(S_t)$, and the public history $h_t^\pi$ is generated by past reports, disclosed signals $(Y_s)_{s\leq t}$, transfers, and public state variables. At a reporting stage, let $\mu_{t^-}$ denote the public posterior immediately before the report under consideration is incorporated. The $t^-$ notation only distinguishes beliefs immediately before and after a private report; absent a reporting stage, $\mu_{t^-}$ coincides with the benchmark posterior $\mu_t$.

For each $(\omega,\theta)$, let $L_t(h_{t^-}^\pi\mid\omega,\theta)$ denote the likelihood of this pre-report public history induced by the information-production technology, past reporting and disclosure rules, transfers, actions, and public state dynamics. Then
\[
\mu_{t^-}(B)
=
\Pr(z\in B\mid h_{t^-}^\pi)
=
\frac{
\int_B L_t(h_{t^-}^\pi\mid\omega,\theta)\,dF_0(\omega,\theta)
}{
\int_{\Omega\times\Theta}
L_t(h_{t^-}^\pi\mid\omega,\theta)\,dF_0(\omega,\theta)
},
\qquad
B\subseteq\Omega\times\Theta,
\]
whenever the denominator is strictly positive.

Conditional on private signal $\xi$ and the pre-report public posterior $\mu_{t^-}$, the agent's interim posterior is
\[
\nu^\xi(B\mid\mu_{t^-})
=
\frac{
\int_B i(\xi\mid\omega)\,d\mu_{t^-}(\omega,\theta)
}{
\int_{\Omega\times\Theta}
i(\xi\mid\omega)\,d\mu_{t^-}(\omega,\theta)
},
\qquad
B\subseteq\Omega\times\Theta.
\]

The report under consideration is then incorporated into the public history, and we denote the resulting continuation public posterior by $\mu_t$. Thus $\mu_t\in\Delta(\Omega\times\Theta)$ is induced by reports, action-generated signals, disclosure, and Bayesian updating rather than directly selected by the principal.

The agent's primitive flow utility is $u(T_t,\omega,\mu_t,x_t)-c(a_t)$, and let $U_{\bar T}(\omega,\mu,x)$ denote the agent's primitive terminal payoff at date $\bar T$. From the perspective of an agent with private signal $\xi$ and pre-report public posterior $\mu$, define
\[
\bar u(T,\xi,\mu,x)
:=
\int_{\Omega\times\Theta}
u(T,\omega,\mu,x)\,
d\nu^\xi(\omega,\theta\mid\mu),
\qquad
\bar U_{\bar T}(\xi,\mu,x)
:=
\int_{\Omega\times\Theta}
U_{\bar T}(\omega,\mu,x)\,
d\nu^\xi(\omega,\theta\mid\mu).
\]

A direct mechanism specifies the report-contingent transfer and disclosure rules $T(t,\hat\xi,\mu,x)$ and $\pi(t,\hat\xi,\mu,x)$, where $\mu$ denotes the public posterior immediately before the report under consideration is incorporated. Separately, let $a(t,\hat\xi,\mu,x)$ denote the action profile the principal seeks to implement following report $\hat\xi$. These rules and the target action profile may depend on the report and the public state but not directly on the unobserved type $\omega$. The report, together with subsequent action-generated and disclosed signals, determines the continuation public posterior through Bayesian updating.

For a true private signal $\xi$ and a report $\hat\xi$, let $U(t,\xi;\hat\xi,\mu,x)$ denote the agent's maximal continuation payoff from reporting $\hat\xi$ at the pre-report public state $(\mu,x)$ and subsequently choosing an optimal admissible action strategy under the transfer, disclosure, and continuation rules assigned to that report. Along the truthful path, write $U(t,\xi,\mu,x):=U(t,\xi;\xi,\mu,x)$. Thus the truthful-reporting constraint compares truthful behavior with every false report followed by an optimal subsequent action deviation and therefore rules out joint report-and-action deviations.

Let $\mathcal L_z^{\pi,a;\xi}$ denote the infinitesimal generator of the continuation public posterior over $z=(\omega,\theta)$ from the perspective of an agent with private signal $\xi$, under disclosure rule $\pi$ and action $a$, with the current report incorporated through the report-contingent continuation rule. The following proposition extends the recursive implementability characterization to the composite-state environment.

\begin{proposition}[Private Payoff-Relevant Information and Recursive Informational Evolution]
\label{prop:composite_states_recursive}

Suppose Assumptions~\ref{ass:regularity_primitives}--\ref{ass:markov_state_representation} hold in the composite-state environment. Suppose further that the public posterior $\mu_t\in\Delta(\Omega\times\Theta)$ induced by the public history is a controlled Markov public state.

Then a report-contingent action profile $a(t,\xi,\mu,x)$ is implementable under a direct report-contingent mechanism if and only if there exists a family of private-information-indexed continuation values $U(t,\xi,\mu,x)$, $\xi\in\Xi$, satisfying the following conditions:
\begin{enumerate}
\item[(i)] For each $\xi\in\Xi$, the truthful continuation value satisfies
\[
\begin{aligned}
-\partial_t U(t,\xi,\mu,x)
={}&
\bar u(T(t,\xi,\mu,x),\xi,\mu,x)
-
c(a(t,\xi,\mu,x))
\\
&+
\mathcal L_z^{\pi,a(t,\xi,\mu,x);\xi}
U(t,\xi,\mu,x)
+
U_x(t,\xi,\mu,x)G(x,a(t,\xi,\mu,x)),
\end{aligned}
\]
with terminal condition
\[
U(\bar T,\xi,\mu,x)
=
\bar U_{\bar T}(\xi,\mu,x).
\]

\item[(ii)] For each $\xi\in\Xi$, the target action satisfies
\[
a(t,\xi,\mu,x)
\in
\arg\max_{\tilde a\in\mathcal A}
\left\{
-c(\tilde a)
+
\mathcal L_z^{\pi,\tilde a;\xi}
U(t,\xi,\mu,x)
+
U_x(t,\xi,\mu,x)G(x,\tilde a)
\right\}.
\]

\item[(iii)] Truthful-reporting incentive compatibility requires
\[
U(t,\xi,\mu,x)
\geq
U(t,\xi;\hat\xi,\mu,x)
\qquad
\text{for all }\xi,\hat\xi\in\Xi,
\]
where the deviation payoff on the right-hand side allows optimal subsequent action deviations after report $\hat\xi$.
\end{enumerate}

Moreover, the principal's continuation problem remains recursive in the public state $(\mu,x)$. With the report-contingent action profile entering the composite-state posterior generator, the reduced implementable disclosure characterization of Subsection~\ref{subsec:optimal_disclosure_general} extends directly to this environment.

\end{proposition}

\begin{proof}
See Appendix~\ref{app:proof_composite_states_recursive}.
\end{proof}

Proposition~\ref{prop:composite_states_recursive} enlarges the public informational state to a posterior over $(\omega,\theta)$ and indexes the agent's continuation problem by private information. The additional restriction is truthful reporting against joint report-and-action deviations. Actions continue to determine information production and feasibility, disclosure governs public posterior evolution, and continuation values support implementation.

\subsection{Limited Commitment}
\label{subsec:limited_commitment}

We next consider limited commitment. In the benchmark model, the principal can commit to future disclosure and transfer rules. Here, the agent may terminate the relationship whenever continuation utility falls below an outside option. Continuation values therefore serve not only as incentive devices but also as participation devices sustaining the informational relationship over time; see, for example, \citet{ThomasWorrall1988,Kocherlakota1996,AlbuquerqueHopenhayn2004}.

The agent's continuation value must satisfy the dynamic participation constraint
\begin{equation}
\label{eq:PC_LC}
U(t,\mu_t,x_t)
\geq
\bar U(t,\mu_t,x_t)
\qquad
\text{for all }t\in[0,\bar T],
\end{equation}
where $\bar U(t,\mu,x)$ is the agent's outside option and may depend on both informational and feasibility states. Thus participation must be satisfied pointwise along the continuation path, rather than only at the initial date.

The recursive implementability characterization in Sections~4--5.1 remains valid, but the feasible continuation plans are further restricted by \eqref{eq:PC_LC}. Under full commitment, continuation values can be adjusted subject to incentive compatibility and the initial participation requirement. Under limited commitment, every adjustment must also preserve participation along the continuation path. Posterior fluctuations may therefore require additional transfers or continuation promises whenever the dynamic participation constraint binds.

The effect on disclosure follows from the marginal-value decomposition in Subsection~\ref{subsec:marginal_value_information} and the cutoff characterization in Theorem~\ref{thm:two_phase_disclosure}. Limited commitment weakly delays disclosure when it lowers the direct marginal benefit of precision, raises its implementation cost, or both, thereby lowering the net marginal value of precision along the relevant optimal path.

\begin{proposition}[Limited Commitment and Disclosure Timing]
\label{prop:limited_commitment}

Consider limited-commitment and full-commitment regimes with the same primitive payoffs, information-production technology, feasibility dynamics, and initial state. For each regime $r\in\{LC,FC\}$, let $(\rho_t^r,x_t^r)$ denote the relevant optimal implementable continuation path and define
\[
M^r(t)
:=
J_\rho^r(t,\rho_t^r,x_t^r)
=
B_\rho^r(t)
-
K_\rho^r(t).
\]

Suppose Assumptions~\ref{ass:regularity_primitives}--\ref{ass:quasilinear} hold. Suppose further that:
\begin{enumerate}
\item[(i)] in both regimes, the feasible disclosure set is $\mathcal D=[0,1]$, the reduced implementable disclosure Hamiltonian satisfies the affine representation \eqref{eq:affine_representation}, the common information-production technology satisfies the linear--Gaussian specification, and the compensated-envelope identity \eqref{eq:Delta_gaussian_revised} holds along the relevant optimal implementable continuation path;

\item[(ii)] along the relevant optimal continuation paths,
\begin{equation}
\label{eq:limited_commitment_component_ordering}
B_\rho^{LC}(t)
\leq
B_\rho^{FC}(t),
\qquad
K_\rho^{LC}(t)
\geq
K_\rho^{FC}(t)
\qquad
\text{for all }t\in[0,\bar T].
\end{equation}
\end{enumerate}

Then
\begin{equation}
\label{eq:limited_commitment_marginal_value_ordering}
M^{LC}(t)
\leq
M^{FC}(t)
\qquad
\text{for all }t\in[0,\bar T].
\end{equation}

If both marginal-value paths are continuous, strictly increasing, and admit unique interior disclosure cutoffs satisfying $M^r(t_r^*)=0$, then $t_{LC}^*\geq t_{FC}^*$. Hence limited commitment weakly delays disclosure under the stated component-ordering conditions.

\end{proposition}

\begin{proof}
See Appendix~\ref{app:proof_limited_commitment}.
\end{proof}

Proposition~\ref{prop:limited_commitment} is a sufficient-condition result. The component ordering in condition~(ii) determines the comparison of disclosure timing. If one additionally requires that the dynamic participation constraint \eqref{eq:PC_LC} bind on a set of positive measure along the limited-commitment path, then limited commitment is economically active rather than merely a nonbinding additional constraint. When posterior perturbations tighten continuation participation constraints, the implementation cost of precision rises because informational fluctuations must be accommodated while preserving both incentive compatibility and continued participation. If the direct informational benefit does not offset this additional cost, the net marginal value of precision is lower, its zero occurs later, and the initial opacity phase is weakly prolonged.

\section{Conclusion}

This paper studies economic systems as endogenous informational adjustment systems. The central idea is that an economy does not merely allocate resources within a given informational environment. Economic activity may produce information; institutional rules determine how that information is disclosed and used; learning changes posterior beliefs; beliefs affect incentives and continuation decisions; and the resulting actions alter both future information and future feasibility. Resource allocation and informational evolution are therefore jointly determined. Building on Hurwicz's informational perspective, the paper moves from allocation and implementation within a specified informational environment to the endogenous evolution of the informational environment itself. The corresponding methodological implication is a system-level one: partial endogenization is sufficient when omitted interactions are economically inessential, but when information production, disclosure, learning, incentives, actions, and feasibility feed back on one another, those interactions must be analyzed jointly.

The main results show that this change in the level of analysis has substantive consequences. First, informational efficiency under exogenous feasibility is generally incompatible with endogenous dynamic feasibility when information-producing activity consumes scarce future opportunities. Second, incentive and disclosure design cannot generally be separated because disclosure changes the continuation return to information-producing action. Third, a more informative information-production technology need not increase maximal implementable value when institutional restrictions prevent replication of outcomes available under a less informative technology. Each negative result also identifies a positive boundary: the paper characterizes when the exogenous-feasibility informational frontier remains attainable, when sequential incentive--disclosure design is valid, and when implementation replicability restores informational monotonicity.

The disclosure analysis provides a complementary positive implication. Optimal disclosure is determined not by the direct value of information alone, but by its contribution to the optimized implementable continuation value. Under general disclosure technologies, disclosure may therefore be interior or extreme; under the additional conditions stated in the paper, the optimal path can take a bang--bang form and, more specifically, exhibit an endogenous transition from opacity to transparency. Information can consequently be worth acquiring and useful for decision-making while still being optimally withheld when the implementation cost of making it public is sufficiently large. The delegated-experimentation application makes this mechanism explicit: experimentation generates information but also consumes scarce capacity, so the optimal policy may preserve capacity and remain inactive early before switching to intensive experimentation and disclosure as the terminal decision approaches.

The extensions show that the recursive informational structure survives when private payoff-relevant information coexists with informational-environment uncertainty and when commitment is limited. These extensions also reinforce the broader point of the paper. The relevant distinction is not between models with ``exogenous information'' and models with ``endogenous information'' in the abstract. Different literatures endogenize different informational margins for legitimate analytical reasons. The system-level question is instead whether the margins held fixed interact with the endogenous ones in ways that alter feasibility, implementation, or welfare. When they do, results obtained under partial endogenization may cease to hold once the feedback structure is restored.

This perspective opens several questions that seem particularly important for a more general theory of endogenous informational economic systems. One concerns the \emph{boundary of partial analysis}: when can one endogenize information production, disclosure, incentives, or learning separately without changing the conclusions of the full system, and what conditions characterize valid decomposition? The separation and replicability results in this paper provide first examples of such boundary conditions, but a general theory of when system-level interactions can be ignored remains to be developed.

A second open question concerns \emph{implementation when the informational environment is itself endogenous}. Standard implementation theory asks whether desired outcomes can arise under a specified informational and strategic environment. Once information production and informational evolution become endogenous, the object to be implemented is richer: actions, beliefs, disclosure, feasibility, and continuation incentives must be jointly generated by the mechanism. This raises natural questions about the appropriate concepts of informational implementability, full implementation, robustness, and informational efficiency when the informational state is itself part of the equilibrium outcome.

A third set of questions arises once information is produced by multiple strategic agents. With dispersed private information, decentralized experimentation, communication, and informational externalities, one agent's action may affect not only her own information and incentives but also the informational opportunities available to others. Competition over scarce informational opportunities, free riding in information production, strategic concealment, and the design of decentralized communication then become parts of the same system. Such environments would connect endogenous informational evolution more directly to organizations, teams, networks, markets, and multi-agent mechanism design.

A fourth direction concerns the interaction between informational evolution and institutional evolution. The present analysis takes the basic mechanism-design environment as given while allowing information, incentives, and feasibility to evolve endogenously within it. In many economic systems, however, the rules governing disclosure, communication, participation, commitment, and enforcement may themselves change in response to accumulated information and past performance. Allowing mechanisms and institutional constraints to evolve jointly with beliefs and actions would move the analysis from endogenous informational evolution toward a broader theory of endogenous institutional and informational adjustment.

Finally, there is a large class of technical and applied extensions. These include genuinely hidden information-producing actions and noisy performance signals, multidimensional actions and feasibility states, stochastic resource dynamics, richer information-production technologies, heterogeneous agents, and alternative commitment structures. Applications to innovation, organizational learning, regulation, finance, digital platforms, and other environments in which economic activity simultaneously creates information and changes future opportunities may help identify which feedbacks are quantitatively important.

These questions suggest a broader research agenda. Modern economic theory has made major progress by endogenizing individual components of economic systems while holding others fixed. The next step in environments where those components interact is not simply to combine more ingredients in a larger model. It is to identify the minimal feedback structure that changes economic conclusions, to characterize when partial analysis remains valid, and to determine how information, incentives, behavior, feasibility, and institutions jointly generate the evolution of an economic system.

\addtocontents{toc}{\protect\setcounter{tocdepth}{-1}}
\appendix

\section{Proofs}
\label{app:proofs}

\subsection{Proof of Lemma~\ref{lem:agent_local_ic}}
\label{app:proof_agent_local_ic}

Fix a mechanism $\mathcal M=(\pi_t,T_t)_{t\in[0,\bar T]}$, let $(\hat a_t)_{t\in[0,\bar T]}$ denote the target action process, and suppose the induced continuation plan is implementable. Under Assumptions~\ref{ass:regularity_primitives} and \ref{ass:markov_state_representation}, the agent's continuation problem admits a Markov continuation-value representation satisfying the dynamic programming principle. Hence
\begin{equation}
\label{eq:proof_hjb}
-\partial_t U(t,\mu,x)
=
\sup_{a\in\mathcal A}
\left\{
u(T(t,\mu,x),\mu,x)
-
c(a)
+
\mathcal L^{\pi,a}U(t,\mu,x)
+
U_x(t,\mu,x)G(x,a)
\right\}.
\end{equation}

Along an implementable continuation path, the target action must attain the pointwise supremum in \eqref{eq:proof_hjb}. Since $u(T(t,\mu,x),\mu,x)$ does not depend directly on the current action, the relevant pointwise problem is to maximize
\[
-c(a)
+
\mathcal L^{\pi,a}U(t,\mu_t,x_t)
+
U_x(t,\mu_t,x_t)G(x_t,a)
\]
over $a\in\mathcal A$. Consequently, $(\hat a_t)$ must satisfy \eqref{eq:local_ic_general} for almost every $t$.

If the maximizer is interior and the action-relevant continuation-value Hamiltonian is differentiable with respect to $a$, the standard necessary first-order condition gives \eqref{eq:agent_foc_general}.

Finally, suppose the derivative of the action-relevant continuation-value Hamiltonian satisfies the stated single-crossing property around a candidate action $\hat a$, being nonnegative for $a<\hat a$ and nonpositive for $a>\hat a$. The Hamiltonian is then weakly increasing up to $\hat a$ and weakly decreasing thereafter, so $\hat a$ is a global maximizer and satisfies \eqref{eq:local_ic_general}. The corresponding one-sided inequalities give the same conclusion when the maximizer lies on the boundary of $\mathcal A$.

\qed

\subsection{Proof of Proposition~\ref{prop:general_characterization}}
\label{app:proof_general_characterization}

Suppose first that the action rule is implementable under the mechanism rules. Under Assumptions~\ref{ass:regularity_primitives} and \ref{ass:markov_state_representation}, the agent's continuation problem satisfies the dynamic programming principle and hence the HJB equation \eqref{eq:agent_hjb_general} with terminal condition \eqref{eq:agent_terminal_general}. By Lemma~\ref{lem:agent_local_ic}, the target action satisfies the pointwise incentive condition \eqref{eq:prop_agent_lic}. Substituting the target action into the HJB equation yields \eqref{eq:prop_agent_hjb}, while the informational and feasibility states follow the induced state dynamics.

Conversely, suppose there exists a continuation-value function satisfying \eqref{eq:prop_agent_hjb}, \eqref{eq:prop_agent_terminal}, and \eqref{eq:prop_agent_lic}, together with the induced state dynamics. By \eqref{eq:prop_agent_lic}, the target action attains the pointwise supremum in the agent's HJB equation. Hence, for any admissible action deviation, the HJB inequality together with the change-of-variables formula and terminal condition implies that the resulting expected payoff cannot exceed $U(t,\mu,x)$. Under the target action, the HJB inequality holds with equality, so the target attains $U(t,\mu,x)$. Therefore no admissible deviation is profitable, and the action rule is implementable under the mechanism rules.

\qed

\subsection{Proof of Proposition~\ref{prop:principal_reduction}}
\label{app:proof_principal_reduction}

By Proposition~\ref{prop:general_characterization}, an action rule is implementable under mechanism rules if and only if there exists a continuation-value function $U$ satisfying the recursive continuation-value equation, the pointwise incentive condition, and the induced state dynamics. The agent's global action incentive constraint can therefore be replaced by these recursive conditions without changing the set of implementable continuation plans.

Substituting this characterization into the principal's objective gives the recursive maximization problem in Proposition~\ref{prop:principal_reduction}. Conversely, every feasible solution to that problem satisfies the applicable participation requirement, the continuation-value equation, the pointwise incentive condition, and the induced state dynamics. Proposition~\ref{prop:general_characterization} therefore implies that it corresponds to an implementable continuation plan in the original contracting problem.

The original problem and the recursive formulation consequently have the same feasible implementable continuation plans and assign the same payoff to each continuation plan. Their optimized values coincide, and this common value is $J(t,\mu,x)$.

\qed

\subsection{Proof of Theorem~\ref{thm:informational_scarcity_impossibility}}
\label{app:proof_informational_scarcity_impossibility}

The theorem rules out a universal guarantee over the class of informational environments considered in the paper. It is therefore sufficient to identify one admissible environment in which informational efficiency under exogenous feasibility is incompatible with endogenous dynamic feasibility.

Consider the linear--Gaussian environment
\[
dS_s=a_s\theta\,ds+\sigma\,dW_s,
\qquad
\dot\rho_s=\delta_s\chi a_s^2,
\qquad
\dot x_s=-a_s,
\]
where $a_s\in[0,\bar a]$, $\delta_s\in[0,1]$, and $\chi:=\sigma^{-2}$. Let the terminal informational criterion be $\Phi(\mu_{\bar T})=\rho_{\bar T}$, and fix any continuation state $(t,\rho,x)$ satisfying $0<x<\bar a(\bar T-t)$.

In the exogenous-feasibility benchmark, full disclosure and maximal action throughout the remaining horizon attain
\[
\mathcal I^{\mathrm{exo}}(t,\rho_t,x_t)
=
\rho_t+\chi\bar a^2(\bar T-t).
\]
Thus $\rho_{\bar T}^{\,\mathrm{exo}}-\rho_t=\chi\bar a^2(\bar T-t)$.

Under any path satisfying endogenous dynamic feasibility, $a_s^2\leq\bar a a_s$. Hence
\[
\begin{aligned}
\rho_{\bar T}-\rho_t
&=
\chi\int_t^{\bar T}\delta_s a_s^2\,ds
\leq
\chi\bar a\int_t^{\bar T}a_s\,ds
\\
&=
\chi\bar a\bigl(x_t-x_{\bar T}\bigr)
\leq
\chi\bar a x_t.
\end{aligned}
\]
Because $x_t<\bar a(\bar T-t)$, it follows that
\[
\rho_{\bar T}-\rho_t
<
\chi\bar a^2(\bar T-t)
=
\rho_{\bar T}^{\,\mathrm{exo}}-\rho_t.
\]
Therefore,
\[
\mathcal I^{\mathrm{endo}}(t,\rho_t,x_t)
<
\mathcal I^{\mathrm{exo}}(t,\rho_t,x_t).
\]

No mechanism satisfying endogenous dynamic feasibility can attain the exogenous-feasibility informational frontier in this environment. Hence no dynamic mechanism can guarantee both properties throughout the class.

\qed

\subsection{Proof of Lemma~\ref{lem:general_informational_feasibility}}
\label{app:proof_general_informational_feasibility}

By definition of $\Lambda$, $h(a)\leq\Lambda g(a)$ for every $a>0$; the same inequality holds at $a=0$ because $h(0)=g(0)=0$. Since $\delta_s\in[0,1]$,
\[
z_{\bar T}-z_t
=
\int_t^{\bar T}\delta_s h(a_s)\,ds
\leq
\Lambda\int_t^{\bar T}g(a_s)\,ds
=
\Lambda\bigl(x_t-x_{\bar T}\bigr),
\]
where the equality follows from $\dot x_s=-g(a_s)$. Similarly, $h(a)\leq\bar h$ implies $z_{\bar T}-z_t\leq\bar h(\bar T-t)$. Combining these inequalities and using $x_{\bar T}\geq0$ gives
\[
z_{\bar T}-z_t
\leq
\min
\left\{
\Lambda\bigl(x_t-x_{\bar T}\bigr),
\bar h(\bar T-t)
\right\}
\leq
\min
\left\{
\Lambda x_t,
\bar h(\bar T-t)
\right\}.
\]

\qed

\subsection{Proof of Corollary~\ref{cor:gaussian_informational_frontier}}
\label{app:proof_gaussian_informational_feasibility}

For every technologically feasible path,
\[
\rho_{\bar T}-\rho
=
\chi\int_t^{\bar T}\delta_s a_s^2\,ds
\leq
\chi\bar a\int_t^{\bar T}a_s\,ds
=
\chi\bar a\bigl(x-x_{\bar T}\bigr)
\leq
\chi\bar a x.
\]
Also, $\rho_{\bar T}-\rho\leq\chi\bar a^2(\bar T-t)$. Consequently,
\[
\rho_{\bar T}-\rho
\leq
\chi\bar a
\min
\left\{
x,
\bar a(\bar T-t)
\right\}.
\]

Conversely, let $\widehat\rho\geq\rho$ satisfy \eqref{eq:gaussian_target_feasibility} and define $\tau:=(\widehat\rho-\rho)/(\chi\bar a^2)$. The target condition implies $\tau\leq\bar T-t$ and $\bar a\tau\leq x$.

Set $\delta_s=1$ and $a_s=\bar a$ on any interval of length $\tau$, and set $a_s=0$ otherwise. Total feasibility depletion is $\bar a\tau\leq x$, so $x_s\geq0$ throughout. Moreover, $\rho_{\bar T}=\rho+\chi\bar a^2\tau=\widehat\rho$.

Thus every precision level in \eqref{eq:gaussian_attainable_precision_set}, and no level outside it, is technologically attainable. If the constructed path is also implementable and individually rational, the target is implementable.

\qed

\subsection{Proof of Theorem~\ref{thm:design_nonseparability}}
\label{app:proof_design_nonseparability}

It is sufficient to identify one admissible environment in which the two properties in Theorem~\ref{thm:design_nonseparability} cannot both hold.

Consider a unit-horizon linear--Gaussian environment with constant feasibility:
\[
dS_s=a_s\theta\,ds+\sigma\,dW_s,
\qquad
\dot\rho_s=\frac{\delta_s a_s^2}{\sigma^2},
\qquad
a_s\in[0,\bar a],
\qquad
\delta_s\in\{0,1\}.
\]
The agent's payoff is $\mathbb E[T]+\eta\rho_1-\int_0^1(\kappa/2)a_s^2\,ds$, where $\eta,\kappa>0$ and $T$ denotes any admissible public-history-contingent transfer, while the principal's payoff is $M\rho_1-\mathbb E[T]$, where $M>0$.

Under nondisclosure, the public signal and posterior process are independent of action. Hence the distribution of any admissible transfer is also independent of action, so the agent uniquely chooses $a_s=0$ almost everywhere. Under full disclosure, a constant transfer is an admissible special case, and the action-dependent payoff is $\int_0^1(\eta/\sigma^2-\kappa/2)a_s^2\,ds$. If $\eta/\sigma^2>\kappa/2$, the agent uniquely chooses $a_s=\bar a$ almost everywhere. Hence an action and incentive plan that implements the action under one disclosure rule cannot be retained under the other.

This failure is also consequential for optimality. For a constant implemented action $a$ and disclosure intensity $\delta$, a constant transfer satisfying binding participation is
$T=\bar U-\eta(\rho_0+\delta a^2/\sigma^2)+(\kappa/2)a^2$, so the principal's payoff is
$(M+\eta)(\rho_0+\delta a^2/\sigma^2)-(\kappa/2)a^2-\bar U$.
Therefore full disclosure and maximal action improve the principal's payoff relative to nondisclosure and zero action by
\[
\left[
\frac{M+\eta}{\sigma^2}
-
\frac{\kappa}{2}
\right]\bar a^2>0,
\]
where the inequality follows from $M>0$ and $\eta/\sigma^2>\kappa/2$.

Thus, in this admissible environment, an action and incentive-continuation plan specified independently of disclosure cannot remain implementable as disclosure varies while also attaining the value of the joint incentive-and-disclosure problem. Hence no sequential dynamic design procedure can guarantee both properties throughout the class.

\qed

\subsection{Proof of Proposition~\ref{prop:separation_criterion}}
\label{app:proof_separation_criterion}

For part~(i), a prescribed action $a^I$ satisfies the pointwise incentive condition under disclosure rule $\pi$ with continuation plan $U$ if and only if $a^I \in \mathcal A^I(t,\mu,x;U,\pi)$. Hence the local separation property between $\pi$ and $\pi'$ holds for $a^I$ if and only if
\[
a^I
\in
\mathcal A^I(t,\mu,x;U,\pi)
\cap
\mathcal A^I(t,\mu,x;U,\pi'),
\]
which establishes \eqref{eq:separation_invariance}.

For part~(ii), suppose
\[
\mathcal A^I(t,\mu,x;U,\pi)
=
\mathcal A^I(t,\mu,x;U,\pi')
\]
for every relevant pair of disclosure rules. Every action satisfying the pointwise incentive condition under one admissible rule then satisfies it under every other rule with the same continuation plan. Hence the pointwise incentive condition is invariant to disclosure.

For part~(iii), suppose that, at every reachable continuation state, the retained incentive-continuation plan remains feasible and implementable as disclosure varies. Suppose also that the principal's feasible set and continuation objective are separable between that plan and the disclosure rule. The subsequent disclosure choice then neither removes the retained plan from the feasible set nor changes its incentive constraints. The joint problem can therefore be written as an iterated maximization over the same feasible set, so the incentive-continuation plan may be selected first and disclosure optimized afterward without loss.

For part~(iv), if
\[
\mathcal A^I(t,\mu,x;U,\pi)
\cap
\mathcal A^I(t,\mu,x;U,\pi')
=
\varnothing,
\]
then no action satisfies the pointwise incentive condition under both disclosure rules with the same continuation plan. The local separation property therefore fails for every prescribed action at that state.

\qed

\subsection{Proof of Corollary~\ref{cor:gaussian_nonseparability}}
\label{app:proof_gaussian_nonseparability}

Under disclosure intensity $\delta$, the positive interior action $a$ satisfies the pointwise incentive condition and therefore the necessary first-order condition
\[
c'(a)
-
U_x(t,\rho,x)G_a(x,a)
=
\frac{2\delta a}{\sigma^2}
U_\rho(t,\rho,x).
\]

If the same action also satisfies the pointwise incentive condition under disclosure intensity $\delta'$, then, because it is interior, it necessarily satisfies
\[
c'(a)
-
U_x(t,\rho,x)G_a(x,a)
=
\frac{2\delta' a}{\sigma^2}
U_\rho(t,\rho,x).
\]
Subtracting the two conditions yields
\[
(\delta-\delta')aU_\rho(t,\rho,x)=0,
\]
which proves necessity.

Conversely, suppose
\[
(\delta-\delta')aU_\rho(t,\rho,x)=0.
\]
Then the first-order condition under $\delta'$ coincides with the first-order condition already satisfied under $\delta$. By the additional shape condition imposed under $\delta'$, this first-order condition is sufficient for global pointwise optimality. Hence the same action $a$ satisfies the pointwise incentive condition under disclosure intensity $\delta'$. This proves sufficiency and establishes \eqref{eq:gaussian_separation_condition}.

\qed

\subsection{Proof of Theorem~\ref{thm:general_decomposition}}
\label{app:proof_general_decomposition}

Fix an optimal implementable continuation plan and consider an admissible compensated informational perturbation around it. Under the envelope condition stated in Subsection~\ref{subsec:marginal_value_information}, the first-order effect of further reoptimization around the optimal unperturbed continuation plan vanishes, so the derivative of the optimized continuation value is obtained from the direct informational effect and the implementability-preserving adjustments induced by the perturbation.

Under Assumption~\ref{ass:quasilinear}, the recursive evaluation equation is
\[
-\partial_t J(t,\mu,x)
=
v(a_t,\mu,x)-T_t
+
\mathcal L^{\pi,a_t}J(t,\mu,x)
+
J_x(t,\mu,x)G(x,a_t),
\]
with terminal condition
\[
J(\bar T,\mu,x)
=
v_{\bar T}(\mu,x)-T_{\bar T}(\mu,x).
\]

Since $\mathcal L^{\pi,a}$ governs the belief component and $\dot x_t=G(x_t,a_t)$, the change-of-variables formula gives
\[
\begin{aligned}
dJ(t,\mu_t,x_t)
=
\Bigl[
&\partial_t J(t,\mu_t,x_t)
+
\mathcal L^{\pi,a_t}J(t,\mu_t,x_t)
\\
&+
J_x(t,\mu_t,x_t)G(x_t,a_t)
\Bigr]dt
+
dM_t,
\end{aligned}
\]
where $(M_t)$ is a local martingale. The maintained integrability conditions ensure that the relevant stopped martingale terms are uniformly integrable, so their conditional expectations vanish. Substitution of the recursive evaluation equation yields
\[
dJ(t,\mu_t,x_t)
=
-\bigl(v(a_t,\mu_t,x_t)-T_t\bigr)dt
+
dM_t.
\]
Integrating from $t$ to $\bar T$, applying the terminal condition, and taking conditional expectations gives
\[
J(t,\mu_t,x_t)
=
\mathbb E_t
\left[
\int_t^{\bar T}
\bigl(v(a_s,\mu_s,x_s)-T_s\bigr)\,ds
+
v_{\bar T}(\mu_{\bar T},x_{\bar T})
-
T_{\bar T}
\right].
\]

Now consider the admissible differentiable compensated posterior perturbation specified in Subsection~\ref{subsec:marginal_value_information}. The implemented action, disclosure, and induced feasibility paths are held fixed, while transfers and continuation promises adjust to preserve incentive compatibility and participation. The maintained regularity permits differentiation under the expectation, giving
\[
\begin{aligned}
D_\mu J(t,\mu_t,x_t)
={}&
\left.
\frac{d}{d\varepsilon}
\mathbb E_t
\left[
\int_t^{\bar T}
v(a_s,\mu_s^\varepsilon,x_s)\,ds
+
v_{\bar T}(\mu_{\bar T}^\varepsilon,x_{\bar T})
\right]
\right|_{\varepsilon=0}
\\
&-
\left.
\frac{d}{d\varepsilon}
\mathbb E_t
\left[
\int_t^{\bar T}
T_s^\varepsilon\,ds
+
T_{\bar T}^\varepsilon
\right]
\right|_{\varepsilon=0}.
\end{aligned}
\]
The first derivative is $B_\mu(t,\mu_t,x_t)$ and the second is $K_\mu(t,\mu_t,x_t)$. Hence
\[
D_\mu J(t,\mu_t,x_t)
=
B_\mu(t,\mu_t,x_t)
-
K_\mu(t,\mu_t,x_t),
\]
which establishes \eqref{eq:general_value_decomposition}. The sign equivalence in \eqref{eq:general_value_boundary} follows immediately.

\qed

\subsection{Proof of Proposition~\ref{prop:gaussian_decomposition}}
\label{app:proof_gaussian_decomposition}

Under the linear--Gaussian specification, posterior beliefs remain Gaussian. When continuation payoffs depend on beliefs through posterior precision, the informational state is summarized by $\rho_t$, and the principal's optimized continuation value can be written as $J(t,\rho,x)$.

An admissible compensated posterior perturbation can therefore be represented by a differentiable family $(\rho_s^\varepsilon)_{s\in[t,\bar T]}$ satisfying $\rho_s^0=\rho_s$ and normalized so that $\left.d\rho_t^\varepsilon/d\varepsilon\right|_{\varepsilon=0}=1$. Applying Theorem~\ref{thm:general_decomposition} to this reduced belief-state representation and using the normalization preceding \eqref{eq:J_rho_def} gives
\[
J_\rho(t,\rho_t,x_t)
=
B_\rho(t,\rho_t,x_t)
-
K_\rho(t,\rho_t,x_t),
\]
where $B_\rho$ and $K_\rho$ are the direct-benefit and implementation-cost derivatives generated by the compensated precision perturbation, as defined in Proposition~\ref{prop:gaussian_decomposition}.

\qed

\subsection{Proof of Theorem~\ref{thm:informational_monotonicity_impossibility}}
\label{app:proof_informational_monotonicity_impossibility}

As in the preceding impossibility results, it is sufficient to identify one admissible environment in which informational monotonicity fails. Consider a unit-horizon linear--Gaussian environment with constant action:
\[
dS_s
=
a\theta\,ds+\sigma\,dW_s,
\qquad
a\in[0,\bar a],
\qquad
s\in[0,1],
\qquad
\chi:=\sigma^{-2}.
\]

Disclosure is restricted to $\delta\in[\underline\delta,1]$, where $\underline\delta>0$. The effective public signal satisfies $dY_s
= \sqrt{\delta}\,a\theta\,ds+\sigma\,dW_s^Y,$
and the posterior-precision increment is $\rho_1-\rho_0=\delta\chi a^2$. The case $\underline\delta=1$ corresponds to automatic disclosure.

Set the feasibility state constant so that the example isolates information production, disclosure, and implementation. Let $m:=\mathbb E[\theta]>0$. The risk-neutral agent's expected payoff is
\[
\mathbb E[T]
-
\frac{\kappa}{2}a^2
-
q(\rho_1-\rho_0),
\]
where $\kappa,q>0$. The principal's expected payoff is
\[
ra+b(\rho_1-\rho_0)-\mathbb E[T],
\]
where $r,b>0$ and $q>b$. Let $\bar U$ denote the agent's reservation payoff.

For any target $(a,\delta)$, consider the public-signal-contingent transfer
\[
T=w+sY_1,
\qquad
s
=
\frac{(\kappa+2q\delta\chi)a}{\sqrt{\delta}\,m}.
\]
If the agent chooses an alternative action $\widetilde a$, expected utility is
\[
w
+
s\sqrt{\delta}\,m\widetilde a
-
\frac{\kappa}{2}\widetilde a^2
-
q\delta\chi\widetilde a^2.
\]
Its derivative is $s\sqrt{\delta}\,m-(\kappa+2q\delta\chi)\widetilde a$, and its second derivative is $-(\kappa+2q\delta\chi)<0$. By the choice of $s$, the unique maximizer is $\widetilde a=a$.

Choosing $w$ so that participation binds gives
\[
\mathbb E[T]
=
\bar U
+
\frac{\kappa}{2}a^2
+
q\delta\chi a^2.
\]
The principal's maximal implementable value under informativeness $\chi$ is therefore
\begin{equation}
\label{eq:optimized_value_informativeness}
J(\chi)
=
\max_{\substack{a\in[0,\bar a]\\ \delta\in[\underline\delta,1]}}
\left\{
ra
-
\frac{\kappa}{2}a^2
-
(q-b)\delta\chi a^2
\right\}
-
\bar U.
\end{equation}

Because $q>b$, the objective is strictly decreasing in $\delta$ whenever $a>0$, so $\delta^*(\chi)=\underline\delta$. Hence
\[
J(\chi)
=
\max_{a\in[0,\bar a]}
\left\{
ra
-
\frac{\kappa}{2}a^2
-
(q-b)\underline\delta\chi a^2
\right\}
-
\bar U.
\]
The objective is strictly concave in $a$. Choose $\bar a$ sufficiently large that the maximizing action is interior for both technologies considered below. The first-order condition therefore gives
\[
a^*(\chi)
=
\frac{r}
{\kappa+2(q-b)\underline\delta\chi}.
\]
Substitution yields
\begin{equation}
\label{eq:optimized_value_informativeness_reversal}
J(\chi)
=
\frac{r^2}
{2\bigl[\kappa+2(q-b)\underline\delta\chi\bigr]}
-
\bar U.
\end{equation}
Because $q>b$ and $\underline\delta>0$, this value is strictly decreasing in $\chi$.

Choose $\chi_H>\chi_L$. Conditional on every fixed action path, the Gaussian raw-signal experiment under $\chi_H$ has a higher signal-to-noise ratio and Blackwell dominates the experiment under $\chi_L$, so $Q^H\succeq_B Q^L$. Nevertheless, \eqref{eq:optimized_value_informativeness_reversal} implies $J(\chi_H)
< J(\chi_L).$
Thus the more informative technology strictly reduces the principal's maximal implementable continuation value.

\qed

\subsection{Proof of Proposition~\ref{prop:replicability_informational_monotonicity}}
\label{app:proof_gaussian_replication_monotonicity}

Consider any implementable mechanism under $Q^L$, with disclosure process $(\delta_s^L)$ and induced action process $(a_s)$. Under $Q^H$, define
\[
\delta_s^H
:=
\delta_s^L\frac{\chi_L}{\chi_H}.
\]
Since $\chi_H>\chi_L$ and $\delta_s^L\in[0,1]$, we have $0\leq\delta_s^H\leq\delta_s^L\leq1$. Moreover, if $\delta_s^L$ is adapted to the public history, then $\delta_s^H$ is adapted as well.

By construction,
\[
\delta_s^H\chi_Ha_s^2
=
\delta_s^L\chi_La_s^2.
\]
Thus the effective public signal-to-noise ratio, and hence the posterior drift and innovation law under Gaussian filtering, are identical under the two technologies for the same action path. The two mechanisms therefore induce the same public posterior process in distribution.

Retain the same action, transfer, and continuation rules under $Q^H$, replacing only $\delta^L$ by $\delta^H$. Because payoffs and implementation constraints depend on the technology only through the induced public posterior process, the agent faces the same continuation incentives and obtains the same payoff distribution. The mechanism remains implementable and induces the same joint distribution of actions, posterior beliefs, feasibility states, transfers, and payoff-relevant outcomes.

Every implementable mechanism under $Q^L$ is therefore reproducible under $Q^H$. Feasible-set inclusion implies
\[
J(t,\mu,x;Q^H)
\geq
J(t,\mu,x;Q^L).
\]

\qed

\subsection{Proof of Theorem~\ref{theorem:general_disclosure}}
\label{app:proof_general_disclosure}

Fix almost every date $t$ along an optimal implementable continuation path and condition on the current state $(\mu_t,x_t)$ and prescribed action $a_t$. For each $\delta\in\mathcal D$, the set $\Gamma^I(t,\mu_t,x_t;a_t,\delta)$ is nonempty, and the supremum defining $\mathfrak H^I(t,\mu_t,x_t;a_t,\delta)$ is attained.

By the dynamic programming principle, the disclosure intensity and the associated implementability-preserving transfer and continuation choices must maximize the reduced implementable disclosure Hamiltonian. Otherwise, there would exist some $\widehat\delta\in\mathcal D$ and an associated continuation choice in $\Gamma^I(t,\mu_t,x_t;a_t,\widehat\delta)$ such that
\[
\mathfrak H^I(t,\mu_t,x_t;a_t,\widehat\delta)
>
\mathfrak H^I(t,\mu_t,x_t;a_t,\delta_t^*).
\]
Using $\widehat\delta$ and its associated continuation choice would yield an admissible deviation that preserves implementation of the prescribed action while increasing the principal's continuation value, contradicting optimality. Therefore,
\[
\delta_t^*
\in
\arg\max_{\delta\in\mathcal D}
\mathfrak H^I(t,\mu_t,x_t;a_t,\delta).
\]

If $\delta_t^*\in\operatorname{int}\mathcal D$, differentiability of the reduced implementable disclosure Hamiltonian gives
\[
\left.
\frac{\partial}{\partial\delta}
\mathfrak H^I(t,\mu_t,x_t;a_t,\delta)
\right|_{\delta=\delta_t^*}
=
0.
\]
This establishes \eqref{eq:optimal_delta_general_revised} and \eqref{eq:delta_FOC_revised}.

\qed

\subsection{Proof of Proposition~\ref{thm:bang_bang_general}}
\label{app:proof_bang_bang_general}

By Theorem~\ref{theorem:general_disclosure}, optimal disclosure maximizes the reduced implementable disclosure Hamiltonian over $\delta\in[0,1]$. Under the affine representation \eqref{eq:affine_representation},
\[
\mathfrak H^I(t,\mu_t,x_t;a_t,\delta)
=
\mathfrak H^I(t,\mu_t,x_t;a_t,0)
+
\delta\Delta^{I,a_t}\mathfrak H(t,\mu_t,x_t).
\]
The first term is independent of $\delta$. If $\Delta^{I,a_t}\mathfrak H(t,\mu_t,x_t)>0$, the reduced Hamiltonian is strictly increasing in $\delta$, so the unique maximizer is $\delta_t^*=1$. If $\Delta^{I,a_t}\mathfrak H(t,\mu_t,x_t)<0$, it is strictly decreasing, so the unique maximizer is $\delta_t^*=0$. If $\Delta^{I,a_t}\mathfrak H(t,\mu_t,x_t)=0$, it is constant in $\delta$, and every $\delta\in[0,1]$ is optimal. This establishes \eqref{eq:bangbang_rule_revised}.

\qed

\subsection{Proof of Corollary~\ref{cor:bang_bang_gaussian_revised}}
\label{app:proof_bang_bang_gaussian}

Under the linear--Gaussian information-production technology, posterior precision evolves according to $\dot\rho_t =\delta_t\frac{a_t^2}{\sigma^2}$.
Hence the generator contribution from moving from nondisclosure to full disclosure is
\[
\Delta^{a_t}J(t,\rho_t,x_t)
=
\frac{a_t^2}{\sigma^2}
J_\rho(t,\rho_t,x_t).
\]
By the compensated-envelope identity \eqref{eq:Delta_gaussian_revised},
\[
\Delta^{I,a_t}\mathfrak H(t,\rho_t,x_t)
=
\frac{a_t^2}{\sigma^2}
J_\rho(t,\rho_t,x_t).
\]
Since $a_t>0$ and $\sigma>0$, the coefficient $a_t^2/\sigma^2$ is strictly positive. Therefore, the switching term $\Delta^{I,a_t}\mathfrak H$ and $J_\rho$ have the same sign. Proposition~\ref{thm:bang_bang_general} then yields \eqref{eq:bang_bang_gaussian_revised}.

\qed

\subsection{Proof of Theorem~\ref{thm:two_phase_disclosure}}
\label{app:proof_two_phase}

By continuity and the endpoint conditions in part~(v), the intermediate value theorem implies that there exists at least one $t^*\in(0,\bar T)$ such that $J_\rho(t^*,\rho_{t^*},x_{t^*})=0$. Since $dJ_\rho(t,\rho_t,x_t)/dt>0$ for all $t\in(0,\bar T)$, the composite net marginal-value path is strictly increasing. Hence the zero is unique, with $J_\rho(t,\rho_t,x_t)<0$ for $t<t^*$ and $J_\rho(t,\rho_t,x_t)>0$ for $t>t^*$.

By Corollary~\ref{cor:bang_bang_gaussian_revised}, an optimal disclosure process can therefore be selected to satisfy
\begin{equation}
\label{eq:two_phase_disclosure_rule}
\delta_t^*
=
\begin{cases}
0, & t<t^*,\\[0.4em]
\text{any value in }[0,1], & t=t^*,\\[0.4em]
1, & t>t^*.
\end{cases}
\end{equation}
The first and third cases hold for almost every date in the corresponding intervals.

\qed

\subsection{Proof of Corollary~\ref{cor:positive_value_disclosure_failure}}
\label{app:proof_positive_value_disclosure_failure}

Theorem~\ref{thm:two_phase_disclosure} implies that $J_\rho(t,\rho_t,x_t)<0$ for every $t<t^*$ and that $\delta_t^*=0$ for almost every $t<t^*$. Combining the decomposition $J_\rho=B_\rho-K_\rho$ with \eqref{eq:positive_direct_information_value} yields $0<B_\rho(t,\rho_t,x_t)<K_\rho(t,\rho_t,x_t)$ for every $t<t^*$. Thus the direct informational value is strictly positive throughout the initial opacity phase, while optimal disclosure is zero almost everywhere on that interval. The positive-value disclosure principle therefore fails.

\qed

\subsection{Proof of Proposition~\ref{prop:comparative_statics_revised}}
\label{app:proof_comparative_statics}

By definition, $M(t,\lambda)=J_\rho\bigl(t,\rho_t(\lambda),x_t(\lambda);\lambda\bigr)$, and the cutoff satisfies $M(t^*,\lambda)=0$. Condition~(iii) states that $M$ is continuously differentiable near $(t^*,\lambda)$ and that $M_t(t^*,\lambda)>0$.

The implicit function theorem therefore implies that $t^*$ is locally differentiable in $\lambda$ and
\[
\frac{dt^*}{d\lambda}
=
-
\frac{M_\lambda(t^*,\lambda)}
{M_t(t^*,\lambda)}.
\]
Since the denominator is positive, $dt^*/d\lambda<0$ if and only if $M_\lambda(t^*,\lambda)>0$. This establishes \eqref{eq:dtstar_general_revised}.

\qed

\subsection{Proof of Proposition~\ref{prop:signal_informativeness}}
\label{app:proof_signal_informativeness}

For signal informativeness $\chi$, define $M(t,\chi)=J_\rho\bigl(t,\rho_t(\chi),x_t(\chi);\chi\bigr)$. Since the cutoff satisfies $M(t^*,\chi)=0$ and condition~(iii) gives $M_t(t^*,\chi)>0$, Proposition~\ref{prop:comparative_statics_revised} yields
\[
\frac{dt^*}{d\chi}
=
-
\frac{M_\chi(t^*,\chi)}
{M_t(t^*,\chi)}.
\]
Hence $dt^*/d\chi<0$ if and only if $M_\chi(t^*,\chi)>0$, proving \eqref{eq:chi_cutoff_sign_revised}.

Along the endogenous optimal state path, $J_\rho=B_\rho-K_\rho$. Taking the total derivative with respect to $\chi$ gives
\[
\begin{aligned}
M_\chi(t,\chi)
={}&
\frac{d}{d\chi}
B_\rho\bigl(t,\rho_t(\chi),x_t(\chi);\chi\bigr)
\\
&-
\frac{d}{d\chi}
K_\rho\bigl(t,\rho_t(\chi),x_t(\chi);\chi\bigr).
\end{aligned}
\]
Therefore, $dt^*/d\chi<0$ if and only if
\[
\frac{d}{d\chi}
B_\rho\bigl(t^*,\rho_{t^*}(\chi),x_{t^*}(\chi);\chi\bigr)
>
\frac{d}{d\chi}
K_\rho\bigl(t^*,\rho_{t^*}(\chi),x_{t^*}(\chi);\chi\bigr),
\]
which establishes \eqref{eq:chi_benefit_cost_condition_revised}.

\qed

\subsection{Proof of Proposition~\ref{prop:effort_cost_curvature_revised}}
\label{app:proof_effort_cost_curvature}

For the action-cost parameter $\kappa$, define $M(t,\kappa)=J_\rho\bigl(t,\rho_t(\kappa),x_t(\kappa);\kappa\bigr)$. Since $M(t^*,\kappa)=0$ and condition~(iii) gives $M_t(t^*,\kappa)>0$, Proposition~\ref{prop:comparative_statics_revised} yields
\[
\frac{dt^*}{d\kappa}
=
-
\frac{M_\kappa(t^*,\kappa)}
{M_t(t^*,\kappa)}.
\]

Using $J_\rho=B_\rho-K_\rho$ along the endogenous optimal path,
\[
\begin{aligned}
M_\kappa(t^*,\kappa)
={}&
\frac{d}{d\kappa}
B_\rho\bigl(t^*,\rho_{t^*}(\kappa),x_{t^*}(\kappa);\kappa\bigr)
\\
&-
\frac{d}{d\kappa}
K_\rho\bigl(t^*,\rho_{t^*}(\kappa),x_{t^*}(\kappa);\kappa\bigr).
\end{aligned}
\]
The stated component-ordering condition makes the first term weakly negative and the second weakly positive, with at least one inequality strict. Hence $M_\kappa(t^*,\kappa)<0$. Since $M_t(t^*,\kappa)>0$, it follows that $dt^*/d\kappa>0$. This establishes \eqref{eq:dtstar_kappa_revised}.

\qed

\subsection{Proof of Proposition~\ref{prop:risk_aversion}}
\label{app:proof_risk_aversion}

Along the respective optimal implementable continuation paths, the general marginal-value decomposition gives
\[
J_\rho^{RA}
=
\mathcal B_\rho^{RA}
-
\mathcal C_\rho^{RA},
\qquad
J_\rho^{RN}
=
\mathcal B_\rho^{RN}
-
\mathcal C_\rho^{RN},
\]
with each term evaluated at its corresponding state $(t,\rho_t^{RA},x_t^{RA})$ or $(t,\rho_t^{RN},x_t^{RN})$. Condition~(iii) gives
\[
\mathcal B_\rho^{RA}
\leq
\mathcal B_\rho^{RN},
\qquad
\mathcal C_\rho^{RA}
\geq
\mathcal C_\rho^{RN}.
\]
Therefore,
\[
J_\rho^{RA}(t,\rho_t^{RA},x_t^{RA})
\leq
J_\rho^{RN}(t,\rho_t^{RN},x_t^{RN})
\]
for every $t\in[0,\bar T]$.

At the risk-neutral cutoff,
\[
J_\rho^{RN}(t_{RN}^*,\rho_{t_{RN}^*}^{RN},x_{t_{RN}^*}^{RN})=0.
\]
The preceding ordering therefore implies that the risk-averse marginal value at the same date is weakly negative. Since the risk-averse composite marginal-value path is strictly increasing and crosses zero uniquely at $t_{RA}^*$, its zero cannot occur before $t_{RN}^*$. Hence $t_{RA}^*\geq t_{RN}^*$.

\qed

\subsection{Proof of Proposition~\ref{prop:composite_states_recursive}}
\label{app:proof_composite_states_recursive}

At a reporting stage, Bayes' rule applied to the common prior $F_0$ and the pre-report public history yields the public posterior $\mu_{t^-}\in\Delta(\Omega\times\Theta)$. Conditional on the privately observed signal $\xi$, the agent combines $\xi$ with $\mu_{t^-}$ to obtain the interim posterior $\nu^\xi(\cdot\mid\mu_{t^-})$ and the corresponding interim expected utilities. The current report is then incorporated into the public history according to the report-contingent mechanism and equilibrium reporting strategy, and subsequent action-generated signals enter the continuation public posterior through the disclosure rule and Bayesian updating. Thus the private signal is used only once in forming the agent's interim belief, while the report determines the continuation branch of the public process.

By the controlled Markov public-state assumption, conditional on the relevant public state $(\mu,x)$, the agent's private signal $\xi$, and current controls, future payoff-relevant distributions depend on past histories only through $(\xi,\mu,x)$. Hence the agent's continuation problem admits a private-information-indexed recursive representation. Applying the dynamic programming principle along the truthful report-contingent continuation branch yields the continuation equation and terminal condition in part~(i).

For action incentive compatibility, fix $\xi$ and the truthful continuation branch and consider an alternative action $\tilde a$. The deviation changes the current action cost, the feasibility drift, and the law of the continuation public posterior generated by future action-dependent signals and disclosure. Pointwise optimality therefore gives the action condition in part~(ii).

Because $\xi$ is privately observed, implementability also requires truthful reporting. By definition, $U(t,\xi;\hat\xi,\mu,x)$ is the maximal continuation payoff of an agent with true signal $\xi$ who reports $\hat\xi$ at the current public state and subsequently chooses an optimal admissible action strategy under the transfer, disclosure, and continuation rules assigned to that report. Hence the inequalities in part~(iii) compare truthful behavior with every false report followed by an optimal subsequent action deviation and therefore rule out joint report-and-action deviations.

Conversely, suppose there exists a family of continuation values satisfying the continuation equations, terminal conditions, action incentive conditions, and truthful-reporting inequalities. The action incentive conditions rule out profitable action deviations along each truthful continuation branch. The truthful-reporting inequalities compare the truthful continuation payoff with the maximal payoff attainable after every false report, including all subsequent admissible action deviations. Therefore no deviation in reports, actions, or both is profitable, and the direct report-contingent mechanism is implementable.

Finally, the principal conditions continuation rules on the public posterior and feasibility state rather than on the unobserved primitive state. The principal's continuation problem therefore remains recursive in the public state $(\mu,x)$. Conditional on the report-contingent action profile induced by the truthful mechanism, disclosure changes the continuation public-posterior dynamics and may also change the continuation choices required to preserve implementation. Accordingly, the reduced implementable disclosure characterization of Subsection~\ref{subsec:optimal_disclosure_general} applies with the composite-state posterior generator and the report-contingent action profile. Thus the recursive disclosure logic is preserved in the composite-state environment.

\qed

\subsection{Proof of Proposition~\ref{prop:limited_commitment}}
\label{app:proof_limited_commitment}

For each commitment regime $r\in\{LC,FC\}$, the marginal-value decomposition gives $M^r(t)=B_\rho^r(t)-K_\rho^r(t)$. By \eqref{eq:limited_commitment_component_ordering},
\[
B_\rho^{LC}(t)\leq B_\rho^{FC}(t),
\qquad
K_\rho^{LC}(t)\geq K_\rho^{FC}(t),
\]
so $M^{LC}(t)\leq M^{FC}(t)$ for all $t\in[0,\bar T]$, establishing \eqref{eq:limited_commitment_marginal_value_ordering}.

When both marginal-value paths are continuous, strictly increasing, and cross zero uniquely, $M^{FC}(t_{FC}^*)=0$ implies $M^{LC}(t_{FC}^*)\leq0$. Since $M^{LC}$ crosses zero uniquely at $t_{LC}^*$, it follows that $t_{LC}^*\geq t_{FC}^*$.

\qed

\bibliographystyle{plainnat}

\pagebreak
\bibliography{Endogenous-information-references}

\end{document}